\documentclass[12pt,a4paper]{article}%
\usepackage{amsfonts,amssymb,amsmath,amsthm,latexsym,a4wide,amsfonts,mathrsfs,bm,color,graphicx,amsthm,epsfig,multirow}
\usepackage{rotating}
\usepackage{lscape}
\usepackage{graphics}
\usepackage{mathrsfs,color}
\usepackage[round]{natbib}
\usepackage{booktabs}
\usepackage{setspace}
\usepackage{amsmath}
\usepackage{amsfonts}
\usepackage{amssymb}
\usepackage{graphicx}%
\providecommand{\U}[1]{\protect\rule{.1in}{.1in}}
\numberwithin{equation}{section}
\newtheorem{theorem}{Theorem}

\newtheorem{algorithm}[theorem]{Algorithm}

\newtheorem{remark}{Remark}

\newtheorem{lemma}{Lemma}
\theoremstyle{definition}

\DeclareMathOperator{\E}{\text{E}}

\renewcommand{\Pr}{\mathbb{P}}

\newtheorem{asu}{Assumption}
\newcounter{subassumption}[asu]
\renewcommand{\thesubassumption}{(\textit{\roman{subassumption}})}
\makeatletter
\renewcommand{\p@subassumption}{\theasu}
\newcommand{\subasu}{
	\refstepcounter{subassumption}
	\thesubassumption}

\numberwithin{equation}{section}
\numberwithin{theorem}{section}
\numberwithin{lemma}{section}

\allowdisplaybreaks

\title{Recovering latent linkage structures and spillover effects with structural breaks in panel data models\thanks{The authors would like to thank Geert Dhaene, Chen Huang, Chih-sheng Hsieh, Elena Manresa, Liangjun Su, participants of seminars at University of Cyprus, Charles III University of Madrid, CUNEF University, University of Southampton, Academica Sinica, IAAE2022, AMEST2022, CMStatistics 2022, Kansai Keikyou Keizaigaku Workshop 2023, SETA2023, and University of Queensland for valuable comments. Okui acknowledges financial support from JSPS KAKENHI Grant numbers 22H00833, 22K20154, 23H00804, and 23K25501. Sun acknowledges the support from the National Natural Science Foundation of China Grant Number 72203032. All authors acknowledge support from Tinbergen Institute.}}

\author{Ryo Okui\thanks{Faculty of Economics, Graduate School of Economics, University of Tokyo. Email: okuiryo@e.u-tokyo.ac.jp}\quad\qquad  Yutao Sun\thanks{Institute for Advanced Economic Research (IAER) at Dongbei University of Finance and Economics. Email: yutao.sun@dufe.edu.cn}\quad\qquad   
	Wendun Wang\thanks{Econometric Institute, Erasmus University Rotterdam; Tinbergen Institute. Email: wang@ese.eur.nl}}

\begin{document}
	
	\maketitle
	
	\begin{abstract}
		This paper introduces a framework to analyze time-varying spillover effects in panel data. We consider panel models where a unit's outcome depends not only on its own characteristics (private effects) but also on the characteristics of other units (spillover effects). The linkage of units is allowed to be latent and may shift at an unknown breakpoint. We propose a novel procedure to estimate the breakpoint, linkage structure, spillover and private effects. We address the high-dimensionality of spillover effect parameters using penalized estimation, and estimate the breakpoint with refinement. 
		We establish the super-consistency of the breakpoint estimator, ensuring that inferences about other parameters can proceed as if the breakpoint were known. The private effect parameters are estimated using a double machine learning method. The proposed method is applied to estimate the cross-country R\&D spillovers, and we find that the R\&D spillovers become sparser after the financial crisis.

		
	\end{abstract}
	
	\clearpage

	\section{Introduction}
	Innovation is widely acknowledged as a cornerstone of technological progress and plays a pivotal role in enhancing productivity \citep{Romer1990}. Empirical research consistently demonstrates that innovation not only bolsters domestic productivity but also generates spillover effects, positively influencing productivity in other countries \citep[see, e.g.,][]{Coe1995, Potterie2001, Coe2009, Ertur2007}. Despite the extensive body of literature on technology spillovers, the dynamics of these effects have received surprisingly little attention. This lack of focus contrasts with empirical evidence, which strongly suggests that spillover effects are not static but instead exhibit significant time-varying patterns. 
	For instance, \citet{OECD2012} highlights that the global financial crisis and the public debt crisis had a pronounced negative impact on business R\&D and innovation activities across a wide range of countries in 2008. 
	Also documented is a remarkable decrease in foreign direct investment  (FDI) and a disproportionate collapse in trade during the financial and debt crisis \citep{OECD2010, OECD2010outlook}.\footnote{\citet{OECD2010} reports a significant 15\% decrease in the volume of OECD imports and exports in the first two quarters of 2009 compared to the same period in 2008. This decline is attributed to the disproportionate fall in domestic demand, outputs, trade of capital goods, and the temporary drying up of short-term trade finance. \citet{OECD2010outlook} shows that FDI inflows sharply declined between 2008 and 2009 in many countries.} Given the disruption in R\&D activities by the crisis and considering that FDI and international trade are commonly regarded as key channels for technology spillovers \citep{keller2004, Coe1995, Potterie2001}, a natural question is if and how economic and political shocks influence R\&D spillovers?
	
	Addressing this empirical question presents several econometric challenges. First, how to identify the latent spillover structure, i.e., which countries generate spillovers and to whom? Existing studies typically assume a predefined functional form of the spillover structure, constructing the adjacency matrix that summarizes cross-country spillovers as an exponential or linear function of an observable bilateral factor, e.g., geographic distance, trade volume, FDI, or linguistic similarity. While the use of each observed factor may be justified by a specific theory, spillover channels are likely influenced by multiple factors simultaneously and also by unobserved variables. Furthermore, the relationship between these factors and spillover structures may be nonlinear and complex. Consequently, the postulated adjacency matrix based on a predetermined function of an observable can lead to misspecification, introducing biases in parameter estimation \citep{hardy2020}. 
	Second, even with a correctly specified spillover structure, estimating the magnitude of pair-specific spillover effects poses additional challenges, because the dimensionality of these effects grows with the number of countries, often surpassing the available time-series observations. 
	Lastly, capturing the temporal dynamics of the latent spillover structure is critical. Most existing studies on R\&D spillovers assume that spillover structures and effects are static over time. In practice, cross-country R\&D spillovers may undergo abrupt changes in response to economic and political shocks. Failing to account for such time variation obscures important insights into network dynamics and may render subsequent analyses based on the estimated spillover structure misleading.
		
	To address these challenges, we propose a novel model and estimation approach that identifies time-varying latent spillover structures without relying on predefined determinants or functional forms of network formation. Specifically, we consider a linear panel data model where the outcome of each unit is influenced by the characteristics of other unit, whose coefficients are referred to as spillover effects, as well as covariates that do not generate spillover effects, whose coefficients are referred to as private effects. Both spillover effects and private effects, or either one, may vary over time. To capture this temporal variability, we model it as abrupt structural breaks. Our approach allows the spillover structure to remain fully unspecified, introducing high-dimensional parameters. Consequently, our methodology is designed to detect structural breaks in these high-dimensional parameters.
		
	Our estimation strategy involves multiple steps. First, we apply a penalized least squares method to obtain preliminary estimates of the breakpoint and other high-dimensional parameters, under the assumption of sparse spillover structures. In the second step, we refine the breakpoint estimates using non-penalized least squares to improve accuracy. With the refined breakpoint estimate, we re-estimate the spillover and private effects within each regime.	To eliminate the regulation bias and enable proper inference for the private effect estimator, we employ double machine learning (DML) in the similar essence of \citet{Belloni2014} and \citet{Chernozhukov2018}. Departing from conventional DML, we propose a post-double-Lasso (least absolute shrinkage and selection operator) procedure to estimate the parameters needed for constructing the Neyman orthogonal score, and this additional post-double-Lasso step is essential for achieving a fast convergence rate of the private effect estimator. 
	
	We show that the proposed breakpoint estimator exhibits super-consistency, meaning that the probability of the estimator coinciding exactly with the true breakpoint approaches one as the cross-sectional dimension ($N$) and time-series dimension ($T$) grow. This property ensures that the estimation of unknown breakpoints has a negligible impact on the accuracy of spillover and private effect estimators. Consequently, the effects estimated under the unknown breakpoint are asymptotically equivalent to those obtained under the true breakpoint. 
	Furthermore, we rigorously analyze the asymptotic properties of the private effect estimator. When the private effect is assumed to be homogeneous across units and constant over time, the estimation can theoretically leverage $NT$ observations. However, the accuracy of this estimation is influenced by the error in estimating the spillover effects, which diminishes at a rate no faster than $\sqrt{T}$. This limitation inherently restricts the convergence rate of the private effect estimator under standard DML inference theory. Hence, we employ post-double-Lasso and establish the conditions under which the proposed private effect estimator achieves $\sqrt{NT}$-consistency.
	
	Using the proposed techniques to examine the relationship between R\&D expenditure and total factor productivity (TFP), we identify a structural break in the cross-country R\&D spillover structure in 2009, aligning closely with the onset of the global financial crisis. Notably, the spillover structure becomes markedly sparser after this break, with many R\&D spillover channels originating from continental Europe, such as those from France, Germany, and Italy, disappearing. This finding indicates that financial crises weaken technology diffusion among countries, resulting in fewer spillover channels in the post-crisis regime. 
	Further analysis reveals that the diminished R\&D spillovers following the financial crisis can be partly attributed to reduced R\&D expenditures in key European countries that had previously served as major spillover generators. This finding aligns with evidence from OECD reports \citep{OECD2012} and other studies using different datasets \citep{Babina2022,hardy-sever2021}. It is also justified by microeconomic foundation that smaller and private firms with weaker balance sheets, as well as firms facing higher financial constraints, tend to cut back disproportionately on R\&D investments during crisis episodes \citep{peia2022}.
	Our study is the first to analyze the time-varying dynamics of cross-country \emph{latent} R\&D spillovers, demonstrating that financial crises affect the structure and magnitude of these spillovers. Importantly, our estimated spillover structure, while correlated with factors such as geographical distance, language similarity, and bilateral economic activities to some extent, cannot be fully explained by these observables. This suggests that the spillover structure is complex, driven by a combination of both observable and unobservable factors.
	
 	In addition to its empirical contribution, our work also enriches multiple strands of econometric literature. First, it relates to the growing body of research on modeling and estimating unknown spillover effects, or more broadly, recovering latent networks. Several identification and estimation techniques have been developed to uncover latent network structures, including network formation models \citep[see, e.g.,][]{Goldsmith-Pinkham2013, Hsieh2016, Hsieh2020}, penalized estimation methods \citep{manresa2016,Paula:2020}, and other approaches \citep{hardy2020, Lewbel2023}, among others. Notably, these studies typically assume that the adjacency matrix remains constant over time.
 	
 	Recently, efforts have been made to relax the static assumption of networks. For instance, \citet{Comola2021} explore the identification and estimation of treatment effects when networks change after an intervention, though their framework assumes that the change points and network structures are known. In different contexts involving network or graph data, several studies have focused on break detection using various methods, such as tensor decomposition of longitudinal network data \citep{Park2020}, break detection in observed networks \citep{WangYuRinaldo2021}, and modeling smooth or abrupt changes in Markov random fields \citep{Kolar2010}. 
 	We differ from these studies by considering break detection in a regression model with a latent network structure, without pre-specifying the network models.

	Two closely related studies in terms of topic, though differing greatly in techniques, are \citet{Han2021} and \citet{Chen2022}. \citet{Han2021} model a time-varying adjacency matrix within a spatial panel data framework, capturing temporal variation by linking the binary adjacency matrix to a set of time-varying observed variables via a logistic function, along with unobserved factors modeled by a Markov process. They employ a Bayesian approach for estimation. In contrast, our model allows for both weighted and unweighted adjacency matrices and does not require specifying the functional form of network formation or the dynamics of unobservables, though it imposes sparsity on the adjacency matrix. \citet{Chen2022} focus on high-dimensional vector autoregressive (VAR) models, using penalized local linear approaches to estimate smoothly varying transition and precision matrices. We differ from this work by considering discrete structural changes. To the best of our knowledge, this paper is the first to analyze structural breaks in a panel regression model with a latent network structure.
	
	Our work also contributes to the growing literature on detecting structural breaks in panel data models. Various methods have been proposed to estimate common structural breaks in the (low-dimensional) slope coefficients, where breaks affect all units simultaneously \citep[see, e.g.,][among others]{bai2010, baltagi&kao&liu2017, qian&su2016, li&qian&su}. \citet{okui&wang2018} propose a penalized estimation technique to detect heterogeneous structural breaks, where the magnitudes and timing of breaks may vary across latent groups. \citet{Lumsdaine2022} extend this by studying break detection when group membership also changes after the break. Our paper differs from these studies by addressing breaks in high-dimensional parameters, which requires a distinct estimation approach and theoretical treatment. 
	
	The topic of structural break detection in high-dimensional models has recently received increasing attention, primarily in the statistical literature. For example, \citet{li2020} explore break detection in large covariance matrices. \citet{safikhani2020} detect break in large VAR models using fused-type penalization with a focus on computational efficiency, which differs from ours in both estimation and theoretical analysis.
	\citet{Lee2016} provide a theoretical analysis of high-dimensional break detection using Lasso under 
	Gaussian errors, while \citet{Lee2018} extend this framework to quantile regressions. 
	Our break detection and asymptotic analysis of the breakpoint estimator builds on and extends that of \citet{Lee2018} to a panel setting. A key contribution of our paper is the establishment of the super-consistency of the estimated breakpoint by exploring the cross-sectional variation, which is crucial for the inference of spillover and private effect estimators. The use of cross-sectional information also alleviates the requirement on the length of time series to a large extent. 
	
	Our work is related to the burgeoning literature on DML \citep[see][among others]{Belloni2014,Chernozhukov2018}. A key difference is that the low-dimensional parameter in our framework, namely the private effect, can be estimated using both cross-sectional and time-series observations, much larger than the number of observations used to estimate the high-dimensional parameters. This increased sample size provides the potential to achieve a faster rate of convergence for the private effect estimator. Building on the DML literature, our estimator solves the Neyman orthogonal score and involves cross-fitting. However, we further innovate by proposing the use of a post-double-Lasso procedure as a preliminary estimate to compute the Neyman orthogonal score. Remarkably, we demonstrate that the proposed private effect estimator attains $\sqrt{NT}$-consistency, even though the spillover effect estimators converge at a slower rate no faster than $\sqrt{T}$.

	The remainder of the paper is organized as follows. Section~\ref{sec:model} introduces the model. Section~\ref{sec:estimation} presents the estimation method. Section~\ref{sec:theory} examines the asymptotic properties of the proposed estimators. Extensions, including the specification of break types and the heterogeneity of private effects, are explored in Section~\ref{sec:extension}. Section~\ref{sec:simulation} evaluates the finite sample performance via simulation. Section~\ref{sec:application} presents the empirical study, and Section~\ref{sec:conclusion} provides concluding remarks.
	
	\section{Econometric models for cross-country R\&D spillover}\label{sec:model}
	
	A widely used model for linking the TFP and innovation employs the Cobb-Douglas function \citep[see, e.g.,][]{Potterie2001, Coe2009} as
	\begin{equation*}
	f_{it}=\exp(\alpha_i+u_{it})\Big(S_{it}^d\Big)^{\gamma}\left(S_{it}^f\right)^{\tau}\Big(H_{it}\Big)^{\delta},
	\end{equation*}
	where $f_{it}$ represents the TFP of country $i$ in year $t$, $\alpha_i$ accounts for unobserved country-specific heterogeneity, and $u_{it}$ denotes the idiosyncratic error term. The determinants of TFP include the domestic and foreign R\&D capital stock, represented by $S_{it}^d$ and $S_{it}^f$, respectively, as well as human capital, denoted by $H_{it}$. These determinants are allowed to correlate with $\alpha_i$, and their contributions to the TFP are captured by the parameters $\gamma$, $\tau$, and $\delta$. The foreign R\&D capital stock, $S_{it}^f$, is typically modeled as a weighted average of R\&D capital stock of other countries, i.e., $S_{it}^f=\sum_{j\neq i}\omega_{ij}S_{jt}^d$, where $\omega_{ij}$ reflects the potential cross-country R\&D spillover channels. Existing literature often assumes a specific, time-invariant spillover structure, meaning that $\omega_{ij}$ is predefined and remains constant over time. For example, the spillover effect are commonly modeled as a function of geographic distance \citep{Ertur2017}, language similarity \citep{Musolesi2007}, or a time aggregated measure of bilateral economic activities, such as foreign direct investment (FDI) \citep{Potterie2001} and international trade \citep{Coe2009}. However, in practice, R\&D spillovers are likely influenced by a combination of observable and unobservable factors, including geopolitical dynamics, historical ties, and other complex relationships. The relative contribution of domestic and foreign R\&D to the TFP also vary markedly across countries. Furthermore, both the spillover channels and effects can shift due to major economic and political events. 
	
	To address these complexities, we propose a model that allows the spillover structure to remain latent and evolve dynamically over time, while accounting for country-specific R\&D effects. Specifically, we extend the TFP Cobb-Douglas  function as
	\begin{equation}\label{eq:CD-function}
		f_{it}=\exp(\alpha_i+u_{it})\Big(S_{it}^d\Big)^{\gamma_{it}}\left(S_{it}^f\right)^{\tau_{it}}\Big(H_{it}\Big)^{\delta},
	\end{equation} 
	where $S_{it}^f=\sum_{j\neq i}\omega_{ij,t}S_{jt}^d$, with $\omega_{ij,t}$ being unknown and time-varying. The R\&D effects, $\gamma_{it}$ and $\tau_{it}$, are allowed to vary across countries and over time. 
	Taking the logarithm of~\eqref{eq:CD-function} results in the following linear model: 
	\begin{equation*}
		\log f_{it} = \alpha_i+\sum_{j=1}^N\gamma_{ij,t}\log S_{jt}^d+\delta\log H_{it}+u_{it},
	\end{equation*}
	where $\gamma_{ii,t}=\gamma_{it}$ and $\gamma_{ij,t}=\tau_{it}\omega_{ij,t}$. With a slight abuse of terminology, we refer to $\{\gamma_{ij,t}\}$ as the spillover parameter, as it encapsulates the impact of the spillover covariate, with $\gamma_{ii,t}$ capturing the heterogeneous direct effect of a country's own R\&D, and $\gamma_{ij,t}$ capturing the pair-specific spillover effect from $j$ to $i$ for $i\neq j$. In contrast, we define $\delta$ as the private effect parameter, since it reflects the influence of non-spillover covariates. Following the majority of studies on the determinants of TFP \citep[see, e.g.,][]{Miller2000, Vandenbussche2006}, we adopt the change specification, namely using the above model to explore the relationship among the changes in the variables.\footnote{Under the change specification, the slope parameters capture the short-run relationship, or put it differently, the deviations from the long-run equilibrium \citep{Coe1995, Potterie2001}.} For notational simplicity, we let $y_{it}=\Delta\log f_{it}$, $x_{jt}=\Delta\log S_{jt}^d$, and $z_{it}=\Delta\log H_{it}$. We assume that the time variation of spillover effects is characterized by structural breaks. For the demonstration purpose, we consider the case of a single break occurring at an unknown time $b$.
	Thus, we obtain the following benchmark econometric model: 
	\begin{equation}\label{eq:model}
		y_{it} = \alpha_i + \sum_{j=1}^N x_{jt} \gamma_{ij,B}\mathbf{1} (t \leq b)+ \sum_{j=1}^N x_{jt} \gamma_{ij,A}\mathbf{1} (t > b) + z_{it}\delta + u_{it},
	\end{equation}
	where $\mathbf{1}(\cdot)$ is an indicator function, $\gamma_{ij,B}$ and $\gamma_{ij,A}$ are the spillover effects before and after the break, respectively. The spillover network is typically sparse due to, e.g., distance decay \citep{Keller2013}, and thus many elements of $\{\gamma_{ij,B}, \gamma_{ij,A}\}$ are expected to be zero. While $x_{it}$ and $z_{it}$ are both scalars, the estimation method and theory can easily accommodate a vector of covariates. In Section~\ref{sec:extension}, we extend model~\eqref{eq:model} by allowing the private effect to be time-varying and heterogeneous (with a latent group pattern) and discuss the specification of breaks.   %	

	To write the model in a more compact form, we denote $X_{it}:= (x_{1t}, \dots, x_{Nt})'$ as the vector of all units' spillover covariate that potentially affects unit $i$'s outcome,  $X_{it} (b):= (X_{it}'\mathbf{1} (t \leq b ) , X_{it}' \mathbf{1}(t > b))'$, and the coefficient of $X_{it} (b)$ is  $\gamma_{i}:=(\gamma_{i1, B}, \dots, \gamma_{iN, B}, \gamma_{i1, A}, \dots, \gamma_{iN, A})^\prime$. Further, denote $W_{it} (b):= (1, X_{it}(b)', z_{it})'$, whose associated coefficient is denoted as $\beta_i:= (\alpha_i, \gamma_{i}^\prime, \delta)'$ for $i=1,\ldots,N$. 
	Then, model~\eqref{eq:model} can be written as
	\begin{align}\label{eq:model-compact}
		y_{it} = W_{it} (b)' \beta_i + u_{it}.
	\end{align}
	This representation indicates that fixed effects are addressed by incorporating individual-specific dummy variables. Our objective is to estimate the unknown breakpoint $b$, the spillover and private-effect parameters in $\beta_i$, with their true values denoted by $b^0$ and $\beta_i^0$, respectively. 
	
	This model generalizes the widely studied panel data models with common structural breaks \citep[see, e.g.,][]{baltagi&qu&kao2016,baltagi&kao&liu2017,qian&su2016} by incorporating spillover effects. Unlike standard panel data models, where breaks affect a small number of parameters, our framework addresses breakpoint detection in high-dimensional parameters, adding complexity to both estimation and theoretical analysis. Additionally, our model can be viewed as a generalization of certain network models, such as network autoregressive models \citep{Zhu:2017}, by allowing unit connections to be both unknown and time-varying. Furthermore, our model extends the work of \citet{manresa2016} by permitting spillover effects to vary over time.

	\section{Estimation method}\label{sec:estimation}
	Our estimation strategy is implemented in multiple stages. First, we estimate the breakpoint. Given the estimated breakpoint, we then re-estimate the spillover effect for each regime and the private effect. Each stage consists of several sub-steps.
	
	We begin by estimating the breakpoint. To achieve this, we first obtain preliminary estimates of the breakpoint, fixed effects, and slope parameters by Lasso. Specifically, denoting $\beta:=(\beta_1',\dots,\beta_N')'$, we solve the following optimization problem:
	\begin{eqnarray}\label{eq:obj-break-point-pls}
		(\hat \beta, \hat b) = \arg\min_{\beta, b} V_{NT} (\beta, b) ,
	\end{eqnarray}
	where 
	\begin{eqnarray*}	
		V_{NT} (\beta, b) =	\frac{1}{NT}\sum_{i=1}^N \sum_{t=1}^T (y_{it} - W_{it} (b)' \beta_i )^2 + D(\beta, b),
	\end{eqnarray*}
	and $D(\beta, b)$ incorporates the $L_1$-penalty, given by
	\begin{align*}
		D(\beta, b) = \lambda_{NT,B} \sum_{i=1}^N \sum_{j=1}^N |\gamma_{ij, B} | \phi_{ij,B} + \lambda_{NT,A} \sum_{i=1}^N \sum_{j=1}^N |\gamma_{ij, A} | \phi_{ij,A}.
	\end{align*}
	Here, $\lambda_{NT,B}$ and $\lambda_{NT,A}$ are tuning parameters, and $\phi_{ij,B}$ and $\phi_{ij,A}$ are adaptive weights, which can be computed, e.g., as $\phi_{ij,B}^2=\left[\sum_{t=1}^{b}x_{jt}^2/b\right]^{-1}\quad \textrm{and}\quad \phi_{ij,A}^2=\left[\sum_{t=b+1}^Tx_{jt}^2/(T-b)\right]^{-1}$. The penalty term $D(\beta, b)$ depends on $b$ through $\phi_{ij,B}^2$ and $\phi_{ij,A}^2$. Although Lasso effectively handles high-dimensional parameters and yields a breakpoint estimator $\hat{b}$ near its true value, $\hat{b}$ is not guaranteed to be consistent.
 	
	To obtain a breakpoint estimator with desirable theoretical properties, we refine it in a second step by re-estimating the breakpoint using least squares, based on the preliminary estimator $\hat{\beta}$,\footnote{In this refining step, we may also use the post-Lasso estimator of $\beta$, namely re-estimating $\beta$ based on the selected nonzero spillover effects from~\eqref{eq:obj-break-point-pls}. While this approach is expected to yield a breakpoint estimator with properties similar to those stated below, it requires adjustments in the theoretical analysis. To maintain theoretical transparency, we update the breakpoint estimator using $\hat{\beta}$ as in~\eqref{eq:obj-break-point-refined}.} without imposing any penalty terms, i.e.,
	\begin{eqnarray}\label{eq:obj-break-point-refined}
		\tilde b = \arg\min_{b} \frac{1}{NT}\sum_{i=1}^N \sum_{t=1}^T (y_{it} - W_{it} (b)' \hat \beta_i )^2.
	\end{eqnarray}
	\begin{remark}
	The choice of adaptive weights in the first step is less critical, because the break can be (super-)consistently estimated as long as $\beta$ lies within a neighborhood of its true value (see the proof of Theorem \ref{th-b-superconsistency}). Thus, it is even feasible to use Lasso with non-adaptive weights in~\eqref{eq:obj-break-point-pls} or other penalized estimation methods. Naturally, a more precise estimate of $\beta$ is expected to enhance the accuracy of $\tilde{b}$	in finite samples.	
	\end{remark}
	
	In the next stage, with the updated breakpoint estimate, we re-estimate the spillover effects of each regime using adaptive Lasso and estimate the private effect using DML in a similar spirit of \citet{CHHW2021} \citep[see also][]{Belloni2014, manresa2016, Chernozhukov2018}. This approach helps mitigate the bias of the private-effect estimator introduced by penalized estimation.
	
	To implement DML, we further split each of the two regimes into two sub-samples: the main sample and the auxiliary sample. Let $\mathcal{T}_m$ be the main sample containing both the pre- and post-break regimes and $\mathcal{T}_a$ be the auxiliary sample. We apply post double Lasso separately to each sample. Specifically, using the auxiliary sample, we perform sparse regressions of $z_{it}$ and $y_{it}$ on $X_{it}(\tilde{b})$, respectively, i.e., 
	\begin{equation}\label{eq:double-lasso}
		z_{it}=\eta_i+X_{it}(\tilde{b})^\prime\nu_{i}+e_{it},\quad y_{it} = \alpha_i^* + X_{it} (\tilde b)'\gamma_{i}^* + \tilde u_{it}, \quad t\in\mathcal{T}_a, 
	\end{equation}
	where $\eta_i$ and $\alpha_i^*$ are the individual fixed effect, $\nu_{i}$ and $\gamma^*_i$ are both $2N\times 1$ vectors, and $e_{it}$ and $\tilde{u}_{it}$ are the error terms. We assume that $\nu_i$ is sparse, i.e., $\sum_{j\neq i}\textbf{1}(\nu_{i}\neq 0)<s_\nu$ for a positive and small constant $s_\nu$, and that $\E(e_{it}|X_{i1},\ldots,X_{iT})=0$. 
	This relationship between $z_{it}$ and $X_{it}(\tilde{b})$ implies that $\alpha_i^* = \alpha_i + \eta_i'\delta$, $\gamma_i^* = \gamma_i + \nu_i'\delta$ and $\tilde u_{it} = u_{it} + e_{it}'\delta$. The sparsity of $\gamma_i$ and $\nu_i$ implies that $\gamma_i^*$ is also sparse. 
	The double Lasso procedure identifies the sets of relevant covariates for each of the two regressions in~\eqref{eq:double-lasso} using adaptive Lasso, with the adaptive weights computed based on the assumed structure of the error covariance. For example, if the errors are heteroscedastic and uncorrelated with $x_{it}$, the adaptive weights can be calculated as
	$\varphi_{ij}^2=1/|\mathcal{T}_a|\left(\sum_{t\in\mathcal{T}_a}x_{jt}\hat{e}^*_{it}\right)^2$,
	where $|\mathcal{T}_a|$ denotes the cardinality of $\mathcal{T}_a$, and $\hat{e}^*_{it}$ is a consistent estimator of $e_{it}$, obtained from a preliminary estimation using conservative weights that depend only on $x_{it}$ and $z_{it}$. 
	If the errors are both heteroscedastic and autocorrelated, a HAC-type estimator may be used instead; see \citet{manresa2016} for a more detailed discussion. Let $\hat{\mathbf{s}}_i^{v*}$ and $\hat{\mathbf{s}}_i^{g*}$ be the sets of covariates selected by the adaptive Lasso for the two models in~\eqref{eq:double-lasso}. 
	Then we combine these sets by taking their union to construct the final set of covariates used for estimation, namely $\hat{\mathbf{s}}_i :=  \hat{\mathbf{s}}_i^{v*} \cup  \hat{\mathbf{s}}_i^{g*}$. Denote $X_{it, \mathbf{s}_i}(b)$ as the subvector of $X_{it}(b)$ containing the elements indexed by $\mathbf{s}_i$. Finally, we perform a post-Lasso step, re-estimating the two models in~\eqref{eq:double-lasso} using OLS with the covariates selected in $ \hat{\mathbf{s}}_i$, i.e., 
	\begin{align}\label{eq:double-lasso-post}
				z_{it} = \eta_i + X_{it,\hat{\mathbf{s}}_i} (\tilde b)'\nu_{i,\hat{\mathbf{s}}_i} + e_{it},\quad y_{it} = \alpha_i^* + X_{it,\hat{\mathbf{s}}_i} (\tilde b)'\gamma_{i,\hat{\mathbf{s}}_i}^* +   \tilde u_{it},\quad t\in \mathcal{T}_a.
	\end{align}
	The resulting estimators are denoted as $\tilde{\alpha}_i^a$, $\tilde{\gamma}_i^a$, $\tilde{\eta}_i^a$ and $\tilde{\nu}_i^a$, where the superscript ``$a$'' represents estimates obtained from the auxiliary sample.  
	
	With these estimates from the auxiliary sample, we estimate the private effect $\delta$ with the main sample $\mathcal{T}_m$ based on the Neyman orthogonal score. Specifically, we solve the following equationfor $\delta$
	\begin{align*}
		\sum_{i=1}^N \sum_{t\in \mathcal{T}_m} \left( y_{it} - \tilde{\alpha}_i^a - X_{it} (\tilde b)' \tilde{\gamma}_i^a - (z_{it}  - \tilde{\eta}_i^a - X_{it} (\tilde b)'\tilde{\nu}_i^a )' \delta \right) 
		\left(z_{it} - \tilde{\eta}_i^a - X_{it} (\tilde b)'\tilde{\nu}_i^a \right) =0,
	\end{align*}
	and can obtain a closed form solution 
	\begin{align}\label{eq:tilde-delta}
		\tilde \delta^m = &
		\left( \sum_{i=1}^N \sum_{t\in \mathcal{T}_m} \left(z_{it} - \tilde{\eta}_i^a - X_{it, \hat{\mathbf{s}}_i} (\tilde b)'\tilde{\nu}_i^a \right)  \left(z_{it} - \tilde{\eta}_i^a - X_{it, \hat{\mathbf{s}}_i} (\tilde b)'\tilde{\nu}_i^a \right)' \right)^{-1} \\ \nonumber
		& \times   \sum_{i=1}^N \sum_{t\in \mathcal{T}_m}  \left(z_{it}  - \tilde{\eta}_i^a - X_{it, \hat{\mathbf{s}}_i} (\tilde b)'\tilde{\nu}_i^a \right) \left( y_{it}  - \tilde{\alpha}_i^a - X_{it, \hat{\mathbf{s}}_i} (\tilde b)' \tilde{\gamma}_i^a  \right), 	
	\end{align}
	where the superscript ``$m$'' in $\tilde \delta^m$ indicates that the estimator is obtained with the main sample. 
	We then swap the role of the main and auxiliary samples to obtain $\tilde{\delta}^a$. 
	The final estimator of $\delta$ is obtained as the average of $\tilde\delta^m$ and $\tilde\delta^a$, namely $\tilde \delta = ( \tilde \delta^m + \tilde \delta^a)/2$.
	In the final step, we estimate the spillovers effects using post-Lasso with $y_{it}-z_{it}^\prime\tilde{\delta}$ as the outcome variable. 

	We summarize the estimation procedure in the following algorithm. 
	
	\begin{algorithm}
		\label{alg:main}\
		\begin{itemize} 
			\item[]\textbf{Step 1}: Estimate the coefficient parameters and breakpoint using (adaptive) Lasso as~\eqref{eq:obj-break-point-pls}. 
			
			\item[]\textbf{Step 2}: Refine the breakpoint estimate using least squares as in~\eqref{eq:obj-break-point-refined}, with the preliminary coefficient estimates obtained in Step 1. Use this refined estimate to segment the sample into pre- and post-break regimes.
			
			\item[]\textbf{Step 3}: Divide the time periods in each of the pre- and post-break regimes equally into main and auxiliary sub-samples.
			
			\item[]\textbf{Step 4}: Use the auxiliary sample to estimate the two models in~\eqref{eq:double-lasso} via adaptive Lasso. Combine the sets of covariates with nonzero estimated coefficients from both models. Re-estimate these two models using OLS with the combined set of covariates as in~\eqref{eq:double-lasso-post}.
			
			\item[]\textbf{Step 5}: Estimate the private effect $\delta$ with the main sample as in~\eqref{eq:tilde-delta}. 
			
			\item[]\textbf{Step 6}:  Swap the role of the main and auxiliary samples. Re-do Step 4 using the main sample and Step 5 using the auxiliary sample. Compute the double machine learning estimate of the private effect as $\tilde{\delta}=(\tilde{\delta}^{m}+\tilde{\delta}^a)/2$, where $\tilde{\delta}^{m}$ and $\tilde{\delta}^a$ are the estimate obtained from the main and auxiliary sample, respectively. 
			
			\item[]\textbf{Step 7:} Estimate the spillover effects by regressing $y_{it}-z_{it}^\prime\tilde{\delta}$ on $X_{it}(\tilde{b})$ using post Lasso. 
		\end{itemize}
	\end{algorithm}
	\begin{remark}
		Our double Lasso estimation in Step 4 differs from double selection described by \citet{Belloni2014} and \citet{manresa2016}. Similar to the standard DML procedure \citep[see, e.g.,][]{Chernozhukov2018}, our method employs cross-fitting, dividing the sample into a main and an auxiliary subsample, which are used to estimate the high-dimensional spillover effect and the finite-dimensional private effect $\delta$ separately. We split each of the pre- and post-break samples into two subsamples for the fact that $\delta$ is constant across the two regimes. Unlike existing DML procedures, including \citet{CHHW2021}, our method introduces an additional post-Lasso step, namely estimating~\eqref{eq:double-lasso-post} with double selected variables using OLS. This step mitigates the uncertainty of the variable selection process, enabling the derivation of the asymptotic distribution of $\tilde{\delta}$ and facilitating statistical inference on $\delta$. A more detailed theoretical discussion is provided in Section \ref{sec:dml-theory}.
	\end{remark}
	
	\section{Asymptotic properties}\label{sec:theory}

	This section studies the asymptotic properties of the proposed method. Our key result establishes the super-consistency of the breakpoint estimator defined in \eqref{eq:obj-break-point-refined}. This super-consistency property allows us to treat the breakpoint as known when justifying the penalized estimator of spillover effects and the DML estimator of the private effect. We then analyze the asymptotic distribution of the DML estimator for private effects. We use the superscript ``0'' to denote true values, e.g., $\beta_i^0$ represents the true value of $\beta_i$. 
	
	\subsection{Consistency of breakpoint estimator}
	Some regularity conditions are needed. Recall that $W_{it} (b):= (1, X_{it}(b)', z_{it} ')'$, and denote $W_{it} := (1, X_{it}', z_{it}')'$. 				
	
	\begin{asu}[Data and parameters]\ \\
		\label{a: basic}
		\vspace*{-0.7cm}
		\begin{enumerate}
			\item[] \subasu\label{as-exogeneity}
			$E[W_{it}(b) u_{it}]=0$ for all $i$, $t$, and any choice of $b$. For any $t$ and some constant $M$, $E[ (\sum_{i=1}^N u_{it} X_{it}'(\gamma_{i,A}^0 - \gamma_{i,B}^0)/N)^2] < M / N$.
			
			\item[] \subasu \label{as-tail}
			Let $w_{it}$ be any element of $(x_{it}, z_{it}, u_{it})$. There exist constants $c_1>0$ and $d_1>0$ such that $\sup_{i,t} P(|w_{it}| > a) \leq \exp (- (a/c_1)^{d_1})$ for sufficiently large $a$. 
			
			\item[] \subasu\label{as-mixing}
			$\{ (x_{it}, z_{it},u_{it})\}_{i=1}^N$ is a strong mixing sequence over $t$ with the mixing coefficient $ a_N[t] \leq A \exp (-c_2 t^{d_2})$ for some strictly positive constants $A$, $c_2$, and $d_2$ that are independent of $N$. 
			
			\item[] \subasu\label{as-compactness} 
			The parameter space is $\mathcal{B} = \{ \{\beta_i\}_{i=1}^N: \sup_{i,k} |\beta_{ik} |< B\}$, where $\beta_{ik}$ is the $k$-th element of $\beta_i$ and $B>0$ is a constant independent of both 
			$N$ and $T$.
			
			\item[] \subasu \label{as-max-eigenvalue}
			Let $\lambda_{\max}(A)$ be the maximum eigenvalue of a matrix $A$. There exists some positive constant $M$ such that $\max_{1 \leq i \leq N }\max_{1 \leq t \leq T}\lambda_{\max} (W_{it}W_{it}') < M$ with probability one. 
		\end{enumerate}
	\end{asu}
	Assumption~\ref{as-exogeneity} requires that the covariates are exogenous. Given the construction of $W_{it}(b)$, this assumption can be decomposed into two parts: the error term having a zero mean, and the absence of correlation between the error term and any covariate, i.e., $E(u_{it})=0$ and $E(x_{jt}u_{it}) =0$ for any $j=1,\ldots,N$. Note that this assumption only imposes restrictions on the contemporaneous correlation between the error term and covariates, but not strict exogeneity or predeterminedness of the covariates. It also limits the degree of cross-sectional correlation by ensuring that the variance of the cross-sectional average of $u_{it} X_{it}$ is of order $1/N$.
	Assumption \ref{as-tail} requires thin tails for the distributions of both the covariates and the error term, implying the existence of finite moments of any order.
	Assumption \ref{as-mixing} guarantees a limited degree of serial dependence in the covariates and the error term.  
	Assumption \ref{as-compactness} resembles the compactness assumption commonly applied to the parameter space in asymptotic analysis of extremum estimators. However, due to the high-dimensional nature of the problem, it is expressed differently in this context.
	Assumption \ref{as-max-eigenvalue} bounds the maximum eigenvalue of $W_{it}W_{it}'$ uniformly over $i$ and $t$. This condition precludes scenarios where any element of $W_{it}$ diverges, even as the dimension of $W_{it}$ grows with the sample size.

	Conditions related to the structural break are also required. Let $\beta_{i,B} := (\alpha_i, \gamma_{i,B}', \delta')'$ and $\beta_{i,A} := (\alpha_i, \gamma_{i,A}', \delta')'$ denote the parameter vectors before and after the break, respectively.

	\begin{asu}[Structural break]\ \\
		\label{a: break}
		\vspace*{-0.7cm}
		\begin{enumerate}
			\setcounter{subassumption}{0}
			\item[] \subasu \label{as-semi-stationary}
			$\max_{ b^0 < t \leq b} E \left[ W_{it} ' (\beta_{i,B}^0 - \beta_{i,B} )\right]^2  < C/b^0 \sum_{t=1}^{b^0} E\left[W_{it} ' (\beta_{i,B}^0 - \beta_{i,B}) \right]^2  $  for any $b > b^0$, and $\max_{ b \leq t < b^0} E \left[W_{it} ' (\beta_{i,A}^0 - \beta_{i,A} )\right]^2  < C/(T-b^0) \sum_{t=b^0+1}^{T} E \left[W_{it} ' (\beta_{i,A}^0 - \beta_{i,A} )\right]^2 $ for any $b< b^0$, for some positive constant $C$ which does not depend on $N$ nor $T$.
			
			\item[] \subasu \label{as-break-identification}
			$\min_t 1/N\sum_{i=1}^N E \left[ X_{it} ' (\gamma_{i,A}^0 - \gamma_{i,B}^0) \right]^2   \geq \underline{m} $ from some constant $\underline{m}>0$ which does not depend on $N$ nor $T$.
			
			\item[] \subasu \label{as-break-date}
			$\tau_{\min} < b^0/T < 1-\tau_{\min}$ for some $\tau_{\min} >0 $ for any $T$.
		\end{enumerate}			
	\end{asu}
	Assumption~\ref{as-semi-stationary} is essential for establishing the asymptotic properties of the preliminary penalized coefficient estimators derived from~\eqref{eq:obj-break-point-pls}. It ensures that the error introduced by estimating the unknown breakpoint has a limited impact on the coefficient estimators. This assumption holds, for instance, when $W_{it}$ is strictly stationary. A similar assumption is also used in \citeauthor{Lee2018} (2018, Assumption A.6).
	Assumption~\ref{as-break-identification} is needed for identifying the breakpoint. It requires that the (true) coefficients before and after the breakpoint differ sufficiently, and that a significant number of units are affected by the structural break.
	Finally, Assumption~\ref{as-break-date} excludes structural breaks near the sample boundaries. It ensures sufficiently long pre- and post-break periods, allowing precise estimation of spillover effects in both regimes.
	
	Since the estimation procedure involves penalization, assumptions regarding the tuning parameters and adaptive weights are necessary. Define the following quantities:
	\begin{align*}
		D_{NT} =& \max ( \lambda_{B,NT} \max_{i,j} \phi_{B,ij}, \lambda_{A,NT} \max_{i,j} \phi_{A,ij}), \\
		D_{\min} =& \min ( \lambda_{B,NT} \min_{i,j} \phi_{B,ij}, \lambda_{A,NT} \min_{i,j} \phi_{A,ij}).
	\end{align*}
	Let $p$ be the dimension of $W_{it} (b)$, which is $p=2 + 2N$ in our benchmark model~\eqref{eq:model}. The number of slope and fixed effects parameters is $N + 2N^2 + 1 $.
	Let $J(\beta)$ denote the index set of the nonzero elements of $\beta \in \mathcal{B}$, and define $s:=|J(\beta^0)|$ as the cardinality of $J(\beta^0)$, i.e., the number of nonzero elements in the true coefficient vector $\beta^0$. Let $\beta_J$ be the vector that matches $\beta$ in the entries indexed by $J(\beta^0)$ and has all other entries set to zero. Similarly, let $\beta_{J^c}$ be the vector that matches $\beta$ in the entries indexed by the complement of $J(\beta^0)$, with all other entries set to zero. Note that $\beta = \beta_{J} + \beta_{J^c}$ and that $\beta^0 = \beta_J^0$. 
	Finally, let $|\cdot |_1$ denote the $L_1$-norm and $\Vert \cdot \Vert_2$ denote the $L_2$ (Euclidean) norm.
	\begin{asu}[Penalty]\ \\
		\label{a: penalty}
		\vspace*{-0.7cm}
		\begin{enumerate}
			\setcounter{subassumption}{0}
			\item[] \subasu \label{as-penalty-order}
			$s \geq \sqrt{N}$ for $N$ sufficiently large; $D_{NT} \sqrt{T} / \log (N p^2) \to \infty $ as $N,T \to \infty$ with probability approaching one and $sD_{NT} N^{c} \to 0$ for some $c>0$; $D_{\min} >0$ for any $N,T$ with probability one;

			\item[] \subasu  \label{as-min-eigenvalue}
			There exists a constant $D_R$ such that $D_{NT} / D_{\min} < D_R$ with probability approaching one. There exists $\rho >0 $, such that for all $\beta \in \mathcal{B} $ satisfying $|\beta_{J^c} |_1 \leq 5 D_R |\beta_{J}|_1 $ it holds that $ \rho \sum_{i=1}^N \vert\beta_{R,i} \vert_1^2 /N \leq  \min_t \sum_{i=1}^N \beta_{R,i}'(E(W_{it}W_{it}') )\beta_{R,i} /N $ for $R=A,B$.
		\end{enumerate}
	\end{asu}
	Assumption \ref{as-penalty-order} imposes restrictions on the degree of sparsity, the relative scales of the cross-sectional and time dimensions, and the order of the tuning parameter in the penalty term in~\eqref{eq:obj-break-point-pls}. This assumption comprises three key parts. The first part requires that the average number of nonzero coefficients per unit is sufficiently large. The second part imposes lower and upper bounds for $D_{NT}$. On the lower bound, $D_{NT}$ must be sufficiently large, with an order greater than $\log (Np^2)/\sqrt{T}$, to ensure sparsity in the estimated coefficients. 
	On the upper bound, $D_{NT}$ is constrained by $s D_{NT} N^{c} \to 0$, where the constant $c$ can be arbitrarily small. For Gaussian $x_{it}$, $N^c$ may be replaced with $\log N$. If the support of $x_{it}$ is uniformly bounded  over $i$ and $t$, $c=0$ is also valid. Overall, this part of condition ensures that $s \cdot (N^c \log (Np^2) / \sqrt{T}) \to 0$. Given $s \geq \sqrt{N}$ and $p = O(N^2)$, a necessary condition is $(\sqrt{N/T}) N^{c}\log N  \to 0$, implying that the time series length must exceed the number of units. 
	Finally, the third part requires a lower bound on $D_{\min}$, akin to the standard Lasso. 
	Assumption~\ref{as-min-eigenvalue} is essentially a compatibility condition commonly used in the Lasso literature \citep[see, e.g.,][]{Lee2018}, but with two notable differences. First, due to the presence of a structural break, the condition is applied separately for the two regimes. Second, the compatibility condition is imposed on the cross-sectional average. A sufficient condition for this compatibility is that $\min_t \min_i \lambda_{\min} (E(W_{it}W_{it}') )$ is bounded away from zero, where $\lambda_{\min}(\cdot)$ returns the minimum eigenvalue.	Note that while $\sum_{t=1}^T W_{it}W_{it}'$ may be singular, especially when $W_{it}$ is high-dimensional, $E(W_{it}W_{it}')$ can still be non-singular. 
	
	Our goal is to analyze the properties of the refined breakpoint estimator. To this end, we first study the behavior of the preliminary estimators defined in~\eqref{eq:obj-break-point-pls}, in a similar spirit of \citet{Lee2018}. Specifically, we evaluate the excess risk of these estimators, defined as:
	\begin{align*}
		R (\beta, b):= \frac{1}{NT}\sum_{i=1}^N \sum_{t=1}^T E \left(y_{it} - W_{it} (b)' \beta_i \right)^2   - \frac{1}{NT}\sum_{i=1}^N \sum_{t=1}^T E\left(y_{it} - W_{it} (b^0)' \beta_i^0  \right)^2.
	\end{align*}
	Here, the expectations are taken with respect to $y_{it}$ and $W_{it}(b)$, while treating $b$ and $\beta_i$ as fixed.
	The following lemma establishes the order of this excess risk.
	\begin{lemma}\label{theorem:risk-consistency}
		Suppose that Assumptions \ref{as-exogeneity}--\ref{as-compactness} and \ref{as-penalty-order} hold. Then,
		\begin{align*}
			R( \hat \beta, \hat b) = O_p \left( s D_{NT} \right).
		\end{align*}
	\end{lemma}
	This lemma indicates that the preliminary estimators for the coefficients and the breakpoint lie within an appropriately defined neighborhood of their true values with probability approaching one. However, it does not guarantee that the preliminary breakpoint estimator coincides with its true value with high probability.
	
	Next, we analyze the properties of the least-squares breakpoint estimator defined in~\eqref{eq:obj-break-point-refined}. Leveraging the fact that the preliminary coefficient estimator resides in the neighborhood of its true value, we build on the (low-dimensional) panel structural break literature \citep[see, e.g.,][]{Lumsdaine2022} to study the behavior of the refined breakpoint estimator.
	
	\begin{theorem}\label{th-b-superconsistency}
		Suppose that Assumptions \ref{as-exogeneity}--\ref{as-max-eigenvalue} are satisfied. Then, as $N,T \to \infty$, 
		it holds that 
		\begin{align*}
			\Pr(\tilde b = b^0 ) \to 1.
		\end{align*}
	\end{theorem}
	This theorem demonstrates the super-consistency of the refined breakpoint estimator, meaning that the estimator converges to the true breakpoint with a probability approaching one. As a result, the DML procedure and the spillover effect estimator in the following stages can be analyzed as if the breakpoint were known.
	
	\subsection{Inference of DML estimator of private effect}
	\label{sec:dml-theory}
	
	While the super-consistency property allows us to disregard the error in breakpoint estimation when conducting inference for the DML estimator of the private effect $\tilde{\delta}$, the analysis becomes non-standard if we aim to achieve $\sqrt{NT}$-consistency. More specifically, although $NT$ observations can be utilized to estimate the private effect, $\tilde{\delta}$ is influenced by the estimation uncertainty of the spillover effects, which diminishes at a rate no faster than $\sqrt{T}$. This limitation constrains the convergence rate of 
	$\tilde{\delta}$ under standard DML inference theory. In this section, we establish conditions under which $\tilde{\delta}$ attains $\sqrt{NT}$-consistency. For clarity and to emphasize the key ideas without overcomplicating the proofs, we impose the following assumptions, some of which are stronger than necessary and can be relaxed.
	
		\begin{asu}\ \\
		\label{a:DML}
		\vspace*{-0.7cm}
		\begin{enumerate}
			\setcounter{subassumption}{0}
			\item[] \subasu \label{assume-ue-iid}
					$(u_{it}, e_{it})$ is i.i.d. over time and units, and independent of $X_{it} (b)$.  $E(e_{it})=0$. 
			
			\item[]\subasu \label{assume-sparse-dml}
				Let the sets $\mathbf{s}_i^{v0}, \mathbf{s}_i^{g0} \subset \{1,\dots, 2N\}$ collect the indices of nonzero elements in the true coefficients $\nu_i^0$ and $\gamma_i^0$ of the two regression models~\eqref{eq:double-lasso}, respectively. The corresponding cardinalities of these two sets, $s_i^{v0}$ and $s_i^{g0}$, satisfy $\sup_i s_i^{v0} < M^v$ and $\sup_i s_i^{g0} < M^g$ for some constants $M^v, M^g >0$.
		
			\item[] \subasu \label{assume-varcon}
			There exists $\mathbf{s}_i^{v*}$ and $\mathbf{s}_{i}^{g*}$, $i=1, \dots, N$, such that $\Pr (\hat{\mathbf{s}}_i^{v*}=\mathbf{s}_i^{v*},  \hat{\mathbf{s}}_i^{g*}=\mathbf{s}_i^{g*}, \forall i) \to 1$ as $T$ increases, and $\mathbf{s}_i^{v0}\subseteq \mathbf{s}_i^{v*}$ and $\mathbf{s}_i^{g0}\subseteq \mathbf{s}_i^{g*}$. 
			
			\item[] \subasu \label{assume-xx-stable}
				Let $\mathbf{s}_i^0:=  \mathbf{s}_i^{v0} \cup \mathbf{s}_i^{g0}$, and denote $||\cdot ||_2$ as the spectral norm. 
				$$\max_i \left\Vert
				\left(\sum_{t\in \mathcal{T}_a} X_{it, \mathbf{s}_i^0} (b)  X_{it, \mathbf{s}_i^0} (b)'/T\right)^{-1} - \left( \sum_{t\in \mathcal{T}_m} X_{it, \mathbf{s}_i^0} (b)  X_{it, \mathbf{s}_i^0} (b)'/T\right)^{-1}
				\right\Vert_2 
				= o_p \left( \sqrt{\frac{T}{N}} \right),$$ 
			
	\end{enumerate}
	\end{asu}

	Assumption~\ref{assume-ue-iid} imposes restrictions on the distribution of the error terms in~\eqref{eq:double-lasso}, and ensures there is no endogeneity. Assumption~\ref{assume-sparse-dml} requires the coefficients of the double Lasso models, $\gamma_i$ and $\nu_i$, to be sparse, with the numbers of nonzero elements bounded by a finite constant. Note that this sparsity assumption holds for all $i$. Even under this assumption, the set of non-zero parameters in the model is still high-dimensional. Each $i$ contains a finite number of non-zero parameters, and there are $N$ units. Thus, the number of non-zero parameters is of order $O(N)$. Assumption~\ref{assume-varcon} concerns the sets of Lasso selected covariates, $\hat{\mathbf{s}}_i^{v*}$ and $\hat{\mathbf{s}}_i^{g*}$, ensuring that these sets exhibit no  random variation in the limit under sufficiently large time-series sample size. It also requires that the selected sets include all covariates with nonzero true coefficients. Notably, we do not impose variable selection consistency; in other words, the selected sets need not exactly match the true set of nonzero coefficients and may contain redundant variables. But, of course, this condition is more easily satisfied by using methods with variable selection consistency (i.e., oracle property), such as the adaptive Lasso \citep{zou2006adaptive} and SCAD \citep{fan2001variable}. Thus, it can also be replaced by lower-level assumptions on the tuning parameters and adaptive weights.
	Assumption~\ref{assume-xx-stable} states that the main and auxiliary samples produce comparable averages of the outer products of the double Lasso selected covariates. These averages converge to their expected values, which, under the stationarity assumption, remain constant over time. It is worth noting that the rate of the left-hand side of Assumption~\ref{assume-xx-stable} is typically $O_p (\log N / \sqrt{T}) $. Hence, to achieve the rate stated in the assumption, namely $o_p( \sqrt{T/N} )$, we would need $\sqrt{N} \log N / T \to 0$, which is satisified under Assumption \ref{a: penalty}.

	\begin{theorem}
		\label{thm-dml}
		Suppose that Assumptions~\ref{a: basic}--\ref{a:DML} hold. 
		Then, as $N,T \to \infty$ and $|\mathcal{T}_m| / T \to 1/2$, 
		\begin{align*}
			\sqrt{NT} (\tilde \delta - \delta^0)
			\to_d N\left( \left[E( e_{it}^2) \right]^{-1} E(u_{it}^2 e_{it}^2 ) \left[E( e_{it}^2) \right]^{-1} \right).
		\end{align*}
	\end{theorem}
	Importantly, the theorem establishes $\sqrt{NT}$-consistency of $\tilde \delta$, despite the nuisance parameters, such as $\gamma_i$, being at most $\sqrt{T}$-consistent. This result is made possible by leveraging the sparsity assumptions in the regressions of $y_{it}$ and $z_{it}$ on $X_{it} (b)$, combined with the use of estimators that possess variable selection consistency properties. The post-Lasso OLS in Step 4 precisely serves to eliminate the uncertainty of variable selection, which is crucial to achieve $\sqrt{NT}$-consistency of $\tilde \delta$. 

	\begin{remark}
	If the private effect is individual-specific, i.e., $\delta_i$, we can utilize the results from \citet{CHHW2021} to derive the asymptotic distribution of its estimator. In this case, it is not necessary to rely on the post-Lasso estimators in Step 4. However, in contrast to \citet{CHHW2021} or other results in the DML literature, our result is not uniform because it depends on the oracle property (Assumptions \ref{assume-sparse-dml} and \ref{assume-varcon}) of the selection procedure, which does not typically hold uniformly across the entire parameter space. Specifically, the result does not hold when the parameter is close to but not precisely zero. Despite the loss of uniformity, we leverage the oracle property to establish the $\sqrt{NT}$-consistency of $\tilde \delta$. In conjunction with the cross-fitting procedure, the oracle property allows us to ignore the estimation errors in $\tilde \mu_i^a$ and $\tilde \gamma_i^a$, which are only $\sqrt{T}$ consistent. Thus, it turns the theoretical analysis into a low-dimensional problem.
	\end{remark}

	\section{Specification of breaks and heterogeneity}\label{sec:extension}
	When a structural break occurs, both the spillover and private effects may change. Moreover, considering cross-country heterogeneity in education and economic development, the private effect of human capital is also likely to vary across countries. This section discusses the estimation of time-varying private effects, the specification of break types, and the estimation of heterogeneous private effects.
	
	\subsection{Time-varying private effect}\label{sec:tv-common-effect}
	If a structural break is present, it is reasonable to consider that private effects may also change. Thus, a natural extension is to model structural breaks in both spillover and private effects as
	\begin{equation}\label{eq:model-tv-common-effect}
		y_{it} = \alpha_i^0 + \left(\sum_{j=1}^N x_{jt} \gamma_{ij,B}^0 + z_{it}\delta^0_B\right) \mathbf{1} (t \leq b^0 )+ \left(\sum_{j=1}^N x_{jt} \gamma_{ij,A}^0 + z_{it} \delta^0_A\right)\mathbf{1} (t > b^0 ) + u_{it},
	\end{equation}
	where $\delta_{B}^0$ and $\delta_{A}^0$ are the private effects before and after the break, respectively.
	To estimate this model, we can adapt Algorithm~\ref{alg:main} by redefining $z_{it}(b):= (z_{it}'\mathbf{1} (t \leq b ) , z_{it}' \mathbf{1}(t > b))'$, $W_{it} (b):= (1, X_{it}(b)', z_{it}(b) ')'$, and $\beta_i:= (\alpha_i, \gamma_{i}^\prime, \delta_B', \delta_A')'$. Using these modified variables, the breakpoint estimate can be obtained in two steps as in~\eqref{eq:obj-break-point-pls} and~\eqref{eq:obj-break-point-refined}. With the estimated breakpoint, we implement DML (Steps 3--6) separately for each regime to estimate time-varying private effects. The spillover effect can then be re-estimated by regressing $y_{it}-z_{it}^\prime\mathbf{1} (t \leq b )\tilde{\delta}_B-z_{it}^\prime\mathbf{1} (t > b )\tilde{\delta}_A$ on $X_{it}(\tilde{b})$ using post Lasso. 
	
	\subsection{Diagnosing the presence and types of structural breaks}\label{sec:test} 
	In practice, researchers often lack prior knowledge about the number of breaks or the parameters affected by them. Therefore, it is essential to detect the presence of breaks and determine which parameters are impacted.
	One approach is to use an information criterion (IC), defined as 
	\begin{equation}\label{eq:IC}
		IC =\log\widehat{Q}+n_pf(N,T),
	\end{equation}
	where $\widehat{Q}$ is the average sum of squared errors, $n_p$ is the total number of parameters, and $f(N, T)$ serves as a tuning parameter; for instance, setting $f(N, T)=\ln(NT)/(NT)$ corresponds to the Bayesian information criterion, which we employ in simulations and applications. Both $\widehat{Q}$ and $n_p$ depend on the number of breaks and the specification which parameters have a break. For example, a break in spillover effects alone implies $n_p=2N^2+N+1$, whereas breaks in both spillover and private effects yield $n_p=2(N^2+1)+N$. We evaluate the finite sample performance of the IC in determining the number and type of breaks in Section~\ref{sec:simulation}.

	\subsection{Heterogeneous private effects}
	Private effects may vary across units due to individual heterogeneity. To capture this heterogeneity, we employ a latent group approach, assuming that units within a group share the same private effect:
	\begin{equation}\label{eq:model-group-delta}
		y_{it} = \alpha_i + \sum_{j=1}^N X_{jt}(b)^\prime\gamma_{ij}+ z_{it}\delta_{g_i} + u_{it},
	\end{equation} 
	where $\delta_{g_i}$ represents the private effect for group $g_i$, and $g_i$ is the \emph{unknown} group membership of unit $i$ that takes values from $\{1,\ldots, G\}$ for some pre-specified number of groups $G$. \citet{hahn&moon2010} provide sound foundations for group structure in game theoretic models. \citet{Bonhomme&lamadon&manresa2017} argue that a group pattern can be a good approximation even in the presence of individual heterogeneity. 
	
	In our empirical application, we employ a state-of-the-art clustering algorithm, namely the Sequential Binary Segmentation Algorithm (SBSA) developed by \citet{Wang2021}, which classifies units based on preliminary estimation of individual-specific private effects. Specifically, we redefine $\beta_i:=(\alpha_i,\gamma_i^\prime,\delta_i)$ and estimate the breakpoint using Steps 1 and 2 outlined in Algorithm~\ref{alg:main}. We then split the sample at the estimated breakpoint and estimate the individual-specific private effect with DML as $\tilde\delta_i=(\tilde{\delta}_i^m+\tilde{\delta}_i^a)/2$, where $\tilde{\delta}_i^m$ and $\tilde{\delta}_i^a$ are the estimators from the main and auxiliary samples, e.g., $\tilde{\delta}_i^a$ solves
	$$
	\sum_{t\in \mathcal{T}_m} \left( y_{it} - \tilde{\alpha}_i^a - X_{it} (\tilde b)' \tilde{\gamma}_i^a - (z_{it}  - \tilde{\eta}_i^a - X_{it} (\tilde b)'\tilde{\nu}_i^a )' \delta_i \right) 
	\left(z_{it} - \tilde{\eta}_i^a - X_{it} (\tilde b)'\tilde{\nu}_i^a \right) =0.
	$$
	With individual estimators at hand, SBSA then clusters units by minimizing within-group variation, iteratively segmenting groups until $G$ clusters are formed. Specifically, for a specified number of groups $G$, SBSA begins by sorting countries based on $\tilde{\delta}_i$, dividing them into two groups to minimize within-group variation. This process continues iteratively, further segmenting existing groups until $G$ groups are formed. See SBSA 1 in \citet{Wang2021} for further details.
	Once group membership estimates $\tilde{g}_i$ are obtained, we re-estimate the group-specific private effects using DML for each group, yielding $\tilde{\delta}_{\tilde{g}_i}$. Finally, the spillover effect is estimated by regressing
	$y_{it}-z_{it}^\prime\tilde{\delta}_{\tilde{g}_i}$ on $X_{it}(\tilde{b})$ using post Lasso.
	
	The latent group approach requires specifying the number of groups $G$. To determine this, we adopt a commonly used method of minimizing an information criterion (IC), constructed analogously to~\eqref{eq:IC} \citep[see][]{bonhomme&manresa2015,okui&wang2018}.

	\section{Monte Carlo simulation}\label{sec:simulation}
	This section assesses the finite-sample performance of the proposed method. We focus on the accuracy of the estimated breakpoint. We also evaluate the estimators of spillover and private effects, as well as the performance of the IC in diagnosing break types.

	\subsection{Data generating process}
	We generate data according to \eqref{eq:model}. The spillover covariate $x_{it}$ is sampled from a standard normal distribution, while the private effect covariate $z_{it}$ is correlated with $x_{it}$ and defined as $z_{it}=1/T\sum_{t=1}^Tx_{it}+\varepsilon_{it}$, where $\varepsilon_{it}$ follows a standard normal distribution. A structural break is introduced at $b_1^0=\lfloor T/3\rfloor$, resulting in different values for $\gamma_{ij,t}$ and $\delta_t$ across the two regimes. We consider two types of networks. The first is a discrete network generated using the Erd\"{o}s-R\'{e}nyi model. The second is a continuous network in which a random subset of the spillover parameters are set to zero and the remaining parameters are drawn from a normal distribution. 

	We examine six distinct data-generating processes (DGPs), incorporating two error processes and three types of breaks. In DGP 1.X, the error term is generated as i.i.d. standard normal. In DGP 2.X, the error follows an autoregressive process, $u_{it}=0.6u_{it-1}+\epsilon_{it}$, where $\epsilon_{it}$ is i.i.d. standard normal.
	For each of the two error processes, we consider three types of breaks: 
	\begin{itemize}
		\item[]DGP X.1: A structural break occurs only in the spillover structure. For discrete networks, we generate an Erd\"{o}s-R\'{e}nyi network with a connection probability of 0.25, independently for each regime. For continuous networks, a random subset of elements are set to zero, and the remaining elements are drawn from a normal distribution with mean 1 and variance 0.5 in each regime. The sparsity level, defined as the proportion of zero elements, is set to 0.7 in the pre-break regime and 0.5 in the post-break regime. The private effect parameter is fixed at $\delta=1.5$ in both regimes.  
		
		\item[]DGP X.2: A structural break occurs only in the non-spillover effect. Here, we set $\delta_B^0=1.5$ and $\delta_A^0=-1.5$, while the network structure remains unchanged between the two regimes.
		
		\item[]DGP X.3: A structural break occurs in both the spillover and non-spillover effects. In this scenario, networks are generated as in DGP X.1, and the private effect parameters are specified as in DGP X.2.
		
	\end{itemize} 
	
	We consider cross-sectional dimensions $N=(15, 30)$ and time series dimensions $T=(50,100)$. The simulations are conducted with 1,000 replications.

	\subsection{Breakpoint, spillover, and private effect estimators}
	We first examine the estimated breakpoint. The accuracy of the breakpoint estimator is measured using the Hausdorff distance (HD), which simplifies to the absolute distance between the true and estimated breakpoints in the case of a single break, specifically $\textrm{HD}(\widehat{k},k^0)\equiv|k-k^0|.$
	Table~\ref{tab:sim-hd} reports the HD ratio relative to the length of time dimension, i.e., $100\times \textrm{HD}(\widehat{k},k^0)/T$, averaged across the replications. 
	The proposed method demonstrates a high level of accuracy in detecting breakpoints across all designs. Increasing the length of the time periods $T$ directly enhances the convergence of the breakpoint estimator and indirectly improves estimation accuracy by yielding more precise estimates of spillover effects. Consequently, the HD ratio decreases rapidly as $T$ grows.
	While an increase in cross-sectional sample size $N$ theoretically facilitates convergence, it may sometimes reduce the accuracy of breakpoint estimates when $T$ is small, particularly under autoregressive error terms. This is because a larger 
	$N$ introduces a higher number of spillover effect parameters, complicating their estimation and, consequently, lowers the accuracy of breakpoint estimates. This high-dimensionality challenge explains the increasing HD ratio observed when 
	$N$ grows from 15 to 30 and $T=50$ in some cases of DGP 2. However, when $T$ is sufficiently large (e.g., $T=100$), increasing $N$ reduces the HD ratio, confirming the consistency result in Theorem~\ref{th-b-superconsistency}.
	
	\begin{table}[htp]
		\begin{center}
			\caption{Hausdorff distance ratio of the estimated breakpoint}\label{tab:sim-hd}
			\begin{tabular}{lccccccccccc}
				\hline\hline
				
				$N$& 15&15 & 30&30  && 15&15 & 30&30\\
				$T$& 50 & 100 & 50 & 100  && 50  & 100 & 50 & 100 \\
				\hline
				& \multicolumn{9}{c}{Erd\"{o}s-R\'{e}nyi network} \\
				&\multicolumn{4}{c}{DGP 1 (i.i.d. error)} && \multicolumn{4}{c}{DGP 2 (autoregressive error)}\\
				\cline{2-5} \cline{7-10}
				DGP X.1 & 0.0077	&0.0002	&0.0082	&0.0000		&&0.0125	&0.0003	&0.0107	&0.0001    \\
				DGP X.2 & 0.0032	&0.0000	&0.0018	&0.0000		&&0.0006	&0.0001	&0.0020	&0.0000    \\
				DGP X.3 & 0.0002	&0.0000	&0.0005	&0.0000		&&0.0004	&0.0000	&0.0001	&0.0000    \\
				\hline
				
				& \multicolumn{9}{c}{Continuous network} \\
				&\multicolumn{4}{c}{DGP 1  (i.i.d. error)} && \multicolumn{4}{c}{DGP 2 (autoregressive error)}\\
				\cline{2-5} \cline{7-10}
				DGP X.1 & 0.0029	&0.0000	&0.0057	&0.0000		&&0.0056	&0.0001	&0.0064	&0.0000    \\
				DGP X.2 & 0.0065	&0.0001 &0.0028	&0.0002		&&0.0065	&0.0004	&0.0263	&0.0003    \\
				DGP X.3 & 0.0003	&0.0000	&0.0014	&0.0000		&&0.0006	&0.0000	&0.0013	&0.0000    \\
				\hline
			\end{tabular}
		\end{center}
	\end{table} 
	
	Next, we assess the accuracy of estimated spillover effects. Following \citet{Paula:2020},  we report three measures: the percentage of zero entries in the true adjacency matrix that are accurately estimated as zeros (``proportion of zeros to zeros''), the percentage of nonzero entries accurately estimated as nonzeros (``proportion of nonzeros to nonzeros''), and the overall root mean squared error (RMSE) of the adjacency matrix for each regime. The RMSE is defined as $\textrm{RMSE}= 1/N(N-1)\sum_{i\neq j}\left(\tilde{\gamma}_{ij,t}-\gamma_{ij,t}^0\right)^2$.	Table~\ref{tab:sim-network-ER} presents the accuracy of adjacency matrix estimation for Erd\"{o}s-R\'{e}nyi networks, and Table~\ref{tab:sim-network-cont} displays the result for continuous networks. The adjacency matrix is estimated with high accuracy for Erd\"{o}s-R\'{e}nyi networks, where the proportion of correctly estimated nonzero elements approaches or exceeds 90\% in most cases. Although the pre-break estimates of nonzero elements are less accurate when $N=30$ and $T=50$, this proportion improves greatly as $T$ increases. Due to uneven sample sizes before and after the break, the RMSE of pre-break spillover effect estimates is generally higher than that of post-break estimates, as expected. When the adjacency matrix is generated in a continuous structure, estimation becomes more challenging, yet the proportion of accurately identified (non)zeros remains high and improves as $T$ increases. Conversely, increasing $N$ reduces estimation accuracy, as reflected in both the proportions of (non)zeros to (non)zeros and the overall magnitude.

	\begin{table}[htp]
		\begin{center}\caption{Accuracy of network estimation: Erd\"{o}s-R\'{e}nyi network}\label{tab:sim-network-ER}
			\begin{tabular}{lccccccccccccccccccc}
				\hline\hline
				& \multicolumn{2}{c}{$N=15,T=50$} && \multicolumn{2}{c}{$N=15,T=100$} &&   \multicolumn{2}{c}{$N=30,T=50$} && \multicolumn{2}{c}{$N=30,T=100$} \\
				& Before & After && Before & After && Before & After && Before & After \\
				\hline
				&\multicolumn{10}{c}{Proportion of nonzeros to nonzeros}\\    	      	       	       	      	      	       	      	      	
				
				DGP 1.1 &0.882	&0.997		&&0.997	&1.000		&&0.670	&0.968		&&0.975	&1.000	   \\
				DGP 1.2 &0.995	&0.995		&&0.999	&0.999		&&0.994	&0.994		&&1.000	&1.000	   \\
				DGP 1.3 &0.894	&0.998		&&0.996	&1.000		&&0.676	&0.962		&&0.968	&1.000	   \\
				
				DGP 2.1 &0.863	&0.994		&&0.992	&1.000		&&0.655	&0.944		&&0.956	&0.999	  \\
				DGP 2.2 &0.998	&0.998		&&1.000	&1.000		&&0.991	&0.991		&&1.000	&1.000	   \\
				DGP 2.3 &0.870	&0.994		&&0.993	&1.000		&&0.663	&0.950		&&0.955	&0.990	   \\

				\hline		
				&\multicolumn{10}{c}{Proportion of zeros to zeros}\\
				DGP 1.1 &0.814	&0.676		&&0.819	&0.661		&&0.825	&0.711		&&0.829	&0.706   \\
				DGP 1.2 &0.922	&0.922		&&0.932	&0.932		&&0.914	&0.914		&&0.939	&0.939	   \\
				DGP 1.3 &0.816	&0.680		&&0.818	&0.663		&&0.828	&0.715		&&0.829	&0.707	   \\
				
				DGP 2.1 &0.785	&0.619		&&0.762	&0.581		&&0.813	&0.674		&&0.792	&0.636   \\
				DGP 2.2 &0.871	&0.871		&&0.871	&0.871		&&0.877	&0.877		&&0.889	&0.889	   \\
				DGP 2.3 &0.789	&0.618		&&0.763	&0.579		&&0.817	&0.678		&&0.792	&0.635	   \\

				\hline		
				&\multicolumn{10}{c}{Root mean squared error}\\

				DGP 1.1 &0.283	&0.191		&&0.157	&0.122		&&0.386	&0.254		&&0.187	&0.126   \\
				DGP 1.2 &0.102	&0.102		&&0.069	&0.069		&&0.112	&0.112		&&0.067	&0.067    \\
				DGP 1.3 &0.275	&0.189		&&0.156	&0.122		&&0.382	&0.246		&&0.186	&0.125	   \\
				
				DGP 2.1 &0.319	&0.240		&&0.200	&0.155		&&0.406	&0.302		&&0.228	&0.162   \\
				DGP 2.2 &0.137	&0.137		&&0.096	&0.096		&&0.146	&0.146		&&0.094	&0.094   \\
				DGP 2.3 &0.313	&0.237		&&0.198	&0.156		&&0.399	&0.294		&&0.228	&0.162	  \\
				\hline
			\end{tabular}
		\end{center}
	\end{table}

	\begin{table}[htp] 		\begin{center}\caption{Accuracy of network estimation: Continuous network}\label{tab:sim-network-cont} 			\begin{tabular}{lccccccccccccccccccc} 				\hline\hline 				& \multicolumn{2}{c}{$N=15,T=50$} && \multicolumn{2}{c}{$N=15,T=100$} &&   \multicolumn{2}{c}{$N=30,T=50$} && \multicolumn{2}{c}{$N=30,T=100$} \\
				
				& Before & After && Before & After && Before & After && Before & After \\
				\hline	&\multicolumn{10}{c}{Proportion of nonzeros to nonzeros}\\    	      	       	       	      	      	       	      	      					                                                                				                                                                				                                                                				
				DGP 1.1 &0.765	&0.923	&&0.908	&0.968		&&0.566	&0.832			&&0.841	&0.954	   \\ 				
				DGP 1.2 &0.902	&0.902	&&0.949	&0.949		&&0.841	&0.841			&&0.935	&0.935	   \\ 				
				DGP 1.3 &0.767	&0.926	&&0.908	&0.967		&&0.570	&0.834			&&0.843	&0.954	   \\ 				                                                                				
				DGP 2.1 &0.757	&0.918	&&0.899	&0.964		&&0.563	&0.823			&&0.833	&0.951	   \\ 				
				DGP 2.2 &0.900	&0.900	&&0.946	&0.946		&&0.836	&0.836			&&0.932	&0.932	   \\ 				
				DGP 2.3 &0.760	&0.920	&&0.901	&0.966		&&0.565	&0.825			&&0.832	&0.951	   \\ 	
				
				\hline                                   				&\multicolumn{10}{c}{Proportion of zeros to zeros}\\   				
				DGP 1.1 &0.795	&0.663	&&0.807	&0.653		&&0.800	&0.676			&&0.805	&0.690	   \\ 				
				DGP 1.2 &0.882	&0.882	&&0.915	&0.915		&&0.828	&0.828			&&0.897	&0.897	   \\ 				
				DGP 1.3 &0.798	&0.666	&&0.805	&0.656		&&0.802	&0.674			&&0.806	&0.690	   \\ 				                                                                				
				DGP 2.1 &0.772	&0.606	&&0.753	&0.575		&&0.790	&0.647			&&0.774	&0.623	   \\ 				
				DGP 2.2 &0.839	&0.839	&&0.857	&0.857		&&0.806	&0.806			&&0.854	&0.854	   \\ 				
				DGP 2.3 &0.774	&0.609	&&0.755	&0.575		&&0.795	&0.648			&&0.774	&0.625	   \\ 				                                                                							                                                                				                                                                				\hline	
				&\multicolumn{10}{c}{Root mean squared error}\\   							      	      	      	       	      	      	       	      	      					                                                                				                                                                				
				DGP 1.1  &0.329	&0.217	&&0.174	&0.129		&&	0.500	&0.347		&&0.231	&0.141	  \\ 				
				DGP 1.2  &0.155	&0.155	&&0.091	&0.091		&&	0.236	&0.236		&&0.102	&0.102    \\ 				
				DGP 1.3  &0.325	&0.215	&&0.175	&0.129		&&	0.495	&0.344		&&0.230	&0.142	  \\ 				                                                                 				
				DGP 2.1  &0.363	&0.263	&&0.216	&0.161		&&	0.515	&0.388		&&0.267	&0.178	  \\ 				
				DGP 2.2  &0.186	&0.186	&&0.115	&0.115		&&	0.259	&0.259		&&0.126	&0.126	  \\ 				
				DGP 2.3  &0.361	&0.260	&&0.216	&0.162		&&	0.510	&0.386		&&0.269	&0.178	  \\ 				
				\hline 			
			\end{tabular} 		
		\end{center} 	
	\end{table}

	Finally, we evaluate the accuracy of DML estimates of private effects. Table~\ref{tab:sim-rmse} presents the bias and RMSE of $\tilde{\delta}$ in each regime, averaged across replications. The private effect is estimated with high accuracy, and its precision depends on the accuracy of the breakpoint and adjacency matrix estimators, as discussed above. Both bias and RMSE generally decrease as $N$ and $T$ increase, supporting the validity of DML and our theoretical results.
	
	\begin{table}[htp]
		\begin{center}
			\caption{Bias and root mean squared error of estimated private effect}\label{tab:sim-rmse}
			\begin{tabular}{llrccccccccccc}
				\hline\hline
				
				& & & \multicolumn{2}{c}{$N=15,T=50$} & \multicolumn{2}{c}{$N=15,T=100$} &   \multicolumn{2}{c}{$N=30,T=50$} & \multicolumn{2}{c}{$N=30,T=100$} \\
				DGP& & $\delta^0$ & Bias & RMSE & Bias & RMSE & Bias & RMSE & Bias & RMSE \\
				\hline
				&&  & \multicolumn{8}{c}{Erd\"{o}s-R\'{e}nyi network} \\
				1.1 &B  & $ 1.5$	&0.057	&0.078		&0.032	&0.051		&0.058	&0.079		&0.028	&0.048 \\
				&A  & $ 1.5$	&0.057	&0.078		&0.032	&0.051		&0.058	&0.079		&0.028	&0.048 \\											      	       	        	       	        	       	
				1.2 &B  & $ 1.5$	&0.063	&0.115		&0.031	&0.031		&0.061	&0.109		&0.031	&0.073 \\
				&A  & $-1.5$	&0.063	&0.088		&0.030	&0.030		&0.061	&0.085		&0.030	&0.052 \\											      	       	        	       	        	       	
				1.3 &B  & $ 1.5$	&0.050	&0.107		&0.027	&0.072		&0.059	&0.107		&0.027	&0.087 \\
				&A  & $-1.5$	&0.056	&0.095		&0.028	&0.064		&0.057	&0.093		&0.035	&0.085 \\ \\
				
				2.1 &B  & $ 1.5$	&0.060	&0.082		&0.032	&0.055		&0.060	&0.081		&0.027	&0.051 \\
				&A  & $ 1.5$	&0.060	&0.082		&0.032	&0.055		&0.060	&0.081		&0.027	&0.051 \\											      	       	        	       	        	       	
				2.2 &B  & $ 1.5$	&0.058	&0.116		&0.030	&0.080		&0.049	&0.109		&0.028	&0.071 \\
				&A  & $-1.5$	&0.055	&0.089		&0.029	&0.057		&0.068	&0.092		&0.029	&0.053 \\											      	       	        	       	        	       	
				2.3 &B  & $ 1.5$	&0.058	&0.112		&0.035	&0.081		&0.062	&0.115		&0.031	&0.073 \\
				&A  & $-1.5$	&0.058	&0.103		&0.032	&0.068		&0.060	&0.097		&0.032	&0.064 \\
				\hline
				&&& \multicolumn{8}{c}{Continuous network} \\   											       	       	        	       	        	       	
				
				1.1 &B  & $ 1.5$	&0.072	&0.108		&0.041	&0.071		&0.068	&0.103		&0.036	&0.067 \\
				&A  & $ 1.5$	&0.072	&0.108		&0.041	&0.071		&0.068	&0.103		&0.036	&0.067 \\											      	       	        	       	        	       	
				1.2 &B  & $ 1.5$	&0.071	&0.173		&0.028	&0.124		&0.005	&0.199		&0.031	&0.110 \\
				&A  & $-1.5$	&0.083	&0.136		&0.035	&0.087		&0.120	&0.163		&0.051	&0.094 \\											      	       	        	       	        	       	
				1.3 &B  & $ 1.5$	&0.066	&0.137		&0.035	&0.095		&0.066	&0.135		&0.027	&0.087 \\
				&A  & $-1.5$	&0.067	&0.130		&0.034	&0.088		&0.074	&0.128		&0.035	&0.085 \\ \\
				
				2.1 &B  & $ 1.5$	&0.070	&0.106		&0.036	&0.069		&0.064	&0.099		&0.029	&0.067 \\
				&A  & $ 1.5$	&0.070	&0.106		&0.036	&0.069		&0.064	&0.099		&0.029	&0.067 \\			
				2.2 &B  & $ 1.5$	&0.067	&0.180		&0.036	&0.122		&0.003	&0.195		&0.042	&0.119 \\   
				&A  & $-1.5$	&0.088	&0.143		&0.030	&0.084		&0.123	&0.165		&0.055	&0.093 \\			
				2.3 &B  & $ 1.5$	&0.055	&0.138		&0.036	&0.103		&0.060	&0.135		&0.026	&0.097 \\   
				&A  & $-1.5$	&0.070	&0.131		&0.032	&0.088		&0.071	&0.125		&0.031	&0.085 \\   
				
				\hline
			\end{tabular}
		\end{center}
		\emph{Notes:} B stands for ``Before,'' and A stands for ``After.''
	\end{table}

	\subsection{Determining the type of break}
	We also evaluate the performance of the IC introduced in Section \ref{sec:test} for identifying the type of breaks. Table~\ref{tab:sim-IC} shows the empirical probability of selecting each break type, with the bolded values indicating the highest probability in each row for a given DGP. Overall, the IC performs effectively. Specifically, in DGP X.1 (break in spillover effects only) and DGP X.2 (break in the private effect only), the IC correctly identifies the type of break with a probability close to or equal to 1. In DGP X.3 (break in both), the IC occasionally favors a specification with fewer parameters when $T=50$. However, as $T$ increases, the empirical probability of selecting the correct specification becomes dominant.
	
	\begin{table}[htp]
		\begin{center}
			\caption{Empirical probability of selecting a specific type of break}\label{tab:sim-IC}
			\small
			\begin{tabular}{lccccccccccccc}
				\hline\hline
				& & \multicolumn{3}{c}{break in} && \multicolumn{3}{c}{break in} && \multicolumn{3}{c}{break in} \\
				$N$ & $T$& $\gamma$ & $\delta$ & both  && $\gamma$ & $\delta$ & both  && $\gamma$ & $\delta$ & both\\
				\hline
				&&  \multicolumn{11}{c}{Erd\"{o}s-R\'{e}nyi network} \\
				&&\multicolumn{3}{c}{DGP 1.1} && \multicolumn{3}{c}{DGP 1.2} && \multicolumn{3}{c}{DGP 1.3}\\
				\cline{3-5} \cline{7-9}  \cline{11-13}
				15 & 50  & \textbf{0.874}	&0.000	&0.126		&&0.000	&\textbf{1.000}	&0.000		&&0.000	&0.468	&\textbf{0.532}    \\
				15 & 100 & \textbf{0.972}	&0.000	&0.028		&&0.000	&\textbf{1.000}	&0.000		&&0.000	&0.320	&\textbf{0.680}    \\
				30 & 50  & \textbf{0.760}	&0.000	&0.240		&&0.000	&\textbf{1.000}	&0.000		&&\textbf{0.728}&0.014	&0.258    \\
				30 & 100 & \textbf{0.842}	&0.000	&0.158		&&0.000	&\textbf{1.000}	&0.000		&&0.000	&0.000	&\textbf{1.000}    \\

				&&\multicolumn{3}{c}{DGP 2.1} && \multicolumn{3}{c}{DGP 2.2} && \multicolumn{3}{c}{DGP 2.3}\\
				\cline{3-5} \cline{7-9}  \cline{11-13}   
				15 & 50   & \textbf{0.884}	&0.000	&0.116		&&0.000	&\textbf{1.000}	&0.000	&&0.002	&0.492	&\textbf{0.506}   \\
				15 & 100  & \textbf{0.942}	&0.000	&0.058		&&0.000	&\textbf{1.000}	&0.000	&&0.000	&0.256	&\textbf{0.744}   \\
				30 & 50   & \textbf{0.780}	&0.000	&0.220		&&0.000	&\textbf{1.000}	&0.000	&&\textbf{0.772}& 0.000	&0.228   \\
				30 & 100  & \textbf{0.858}	&0.000	&0.142		&&0.000	&\textbf{1.000}	&0.000	&&0.000	&0.000	&\textbf{1.000}   \\
				
				\hline
				&& \multicolumn{11}{c}{Continuous network} \\
				& &\multicolumn{3}{c}{DGP 1.1} && \multicolumn{3}{c}{DGP 1.2} && \multicolumn{3}{c}{DGP 1.3}\\
				\cline{3-5} \cline{7-9}  \cline{11-13}
				15 & 50  & \textbf{0.888}	&0.000	&0.112		&&0.000	&\textbf{1.000}	&0.000		&&0.048	&0.014	&\textbf{0.938}    \\
				15 & 100 & \textbf{0.932}	&0.000	&0.068		&&0.000	&\textbf{1.000}	&0.000		&&0.000	&0.000	&\textbf{1.000}    \\
				30 & 50  & \textbf{0.736}	&0.000	&0.264		&&0.006	&\textbf{0.992}	&0.002		&&\textbf{0.656}&0.000	&0.344   \\
				30 & 100 & \textbf{0.812}	&0.000	&0.188		&&0.000	&\textbf{1.000}	&0.000		&&0.002	&0.000	&\textbf{0.998}    \\

				&&\multicolumn{3}{c}{DGP 2.1} && \multicolumn{3}{c}{DGP 2.2} && \multicolumn{3}{c}{DGP 2.3}\\
				\cline{3-5} \cline{7-9}  \cline{11-13}
				15 & 50   & \textbf{0.874}	&0.000	&0.126		&&0.000	&\textbf{1.000}	&0.000		&&0.094	&0.002	&\textbf{0.904}  \\
				15 & 100  & \textbf{0.912}	&0.000	&0.088		&&0.000	&\textbf{1.000}	&0.000		&&0.000	&0.000	&\textbf{1.000}  \\
				30 & 50   & \textbf{0.760}	&0.000	&0.240		&&0.012	&\textbf{0.988}	&0.000		&&\textbf{0.692}&0.000	&0.308  \\
				30 & 100  & \textbf{0.786}	&0.000	&0.214		&&0.000	&\textbf{1.000}	&0.000		&&0.000	&0.000	&\textbf{1.000}  \\
				\hline
			\end{tabular}
		\end{center}

	\end{table}

	\section{Empirical analysis of cross-country R\&D spillover}\label{sec:application}
	
	The role of innovation in driving economic growth has garnered substantial attention from economists and policymakers. Innovation, fueled by knowledge gained from research and development (R\&D), expands the overall knowledge base, enabling more efficient resource use and boosting productivity within a country. Consequently, innovation is viewed as a crucial driver of an economy’s productivity growth  \citep{Romer1990}. Moreover, the knowledge embedded in exported goods can generate spillover effects, indirectly stimulating productivity growth in partner countries. Foreign R\&D also offers direct benefits, such as facilitating the acquisition of new skills and technologies, including production processes, organizational methods, and materials \citep{Coe1995}.
	Although the theory of cross-country R\&D spillovers is widely recognized, empirical estimates of R\&D spillover effects remain elusive, partly due to the ambiguity surrounding the specific channels through which these spillovers occur. 
	
	Existing studies on R\&D spillovers typically assume a predefined structure for spillover channels based on economic or geographic measures. While each of these measures can be justified by a specific theory, it is likely that spillover channels are shaped by multiple factors—both observable and unobservable—that may interact in complex, nonlinear ways. Consequently, using an adjacency matrix based on a single measure (e.g., an exponential or linear function) risks misspecifying the spillover structure. Furthermore, existing literature often assumes that spillover effects are constant over time. This assumption is also restrictive given the economic and political shocks that impact international relations and corporate strategies. Ignoring temporal variation may obscure essential information about network dynamics and lead to biased estimates of spillover effects.
	
	Motivated by these challenges, we revisit the study of cross-country R\&D spillovers to identify the latent spillover structure without imposing a specific network formation mechanism. In particular, we extend \eqref{eq:model} to enable spillover and/or private effects to vary over time, i.e.,
	\begin{equation}\label{eq:empirical-model}
		y_{it} = \alpha_i + \sum_{j=1}^N x_{jt} \gamma_{ij,t}+ z_{it}\delta_t + u_{it},
	\end{equation}
	where $y_{it}=\Delta\log f_{it}$, $x_{it}=\Delta\log S_{jt}^d$, and $z_{it}=\Delta\log H_{it}$. 
	We model the time variation in parameters using structural breaks, with the presence and type of breaks determined through diagnostic tests. Later, we will extend model~\eqref{eq:empirical-model} to incorporate potential heterogeneity in the private effect of human capital.  
	
	Our analysis focuses on 24 OECD countries, a sample widely examined in recent literature \citep[see, e.g.,][]{Potterie2001, Coe2009, Ertur2017}. We use the most up-to-date data from 1981 to 2019, sourced from the OECD and the Penn World Table version 10.0. Total factor productivity (TFP) is measured at a constant national price level, normalized to 1 for 2017. R\&D expenditure is calculated by multiplying the percentage of R\&D expenditure in GDP by the real GDP at constant 2017 national prices (in millions).\footnote{We construct R\&D expenditure in this way because the OECD's R\&D expenditure data is not normalized in the same manner as other data from the Penn World Table, while the latter does not provide expenditure level data.} To create a balanced panel, we interpolate missing R\&D expenditure values using data from \citet{Ertur2017} where available, or apply linear interpolation for values outside the sample period of \citet{Ertur2017}. Human capital is measured by the average years of schooling for the population aged 25 and older, consistent with the existing literature.

	\subsubsection*{Diagnostic analysis}

	Effective model estimation requires the knowledge of the number and type of structural breaks. We range the number of breaks from 0 (i.e., no breaks) to 2, and detect them sequentially. For each break, we consider three scenarios: breaks only in the spillover effect, only in the private effect, and in both.  Figure~\ref{fig:empirical-IC}(a) plots the IC results for different numbers and types of breaks. The results indicate that the best model features a single break in the spillover effects, while the private effects remain constant. However, the difference between this model and one that includes breaks in both the spillover and private effects is minor.
	
	\begin{figure}[htp]\caption{Information criterion to determine the number of breaks and groups}\label{fig:empirical-IC} 
		\begin{center}
			\begin{tabular}{cc}
				\includegraphics[scale=0.5]{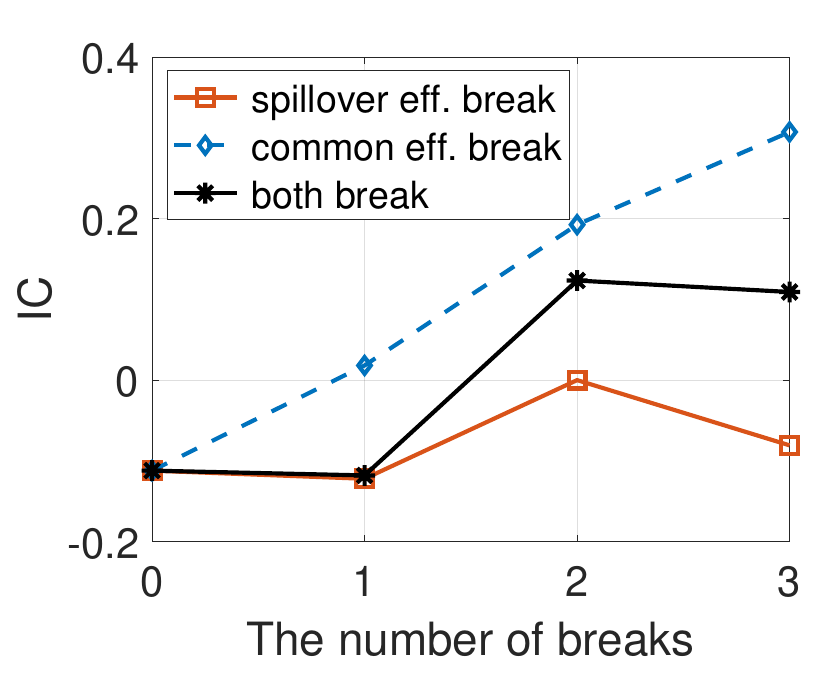}&
				\includegraphics[scale=0.5]{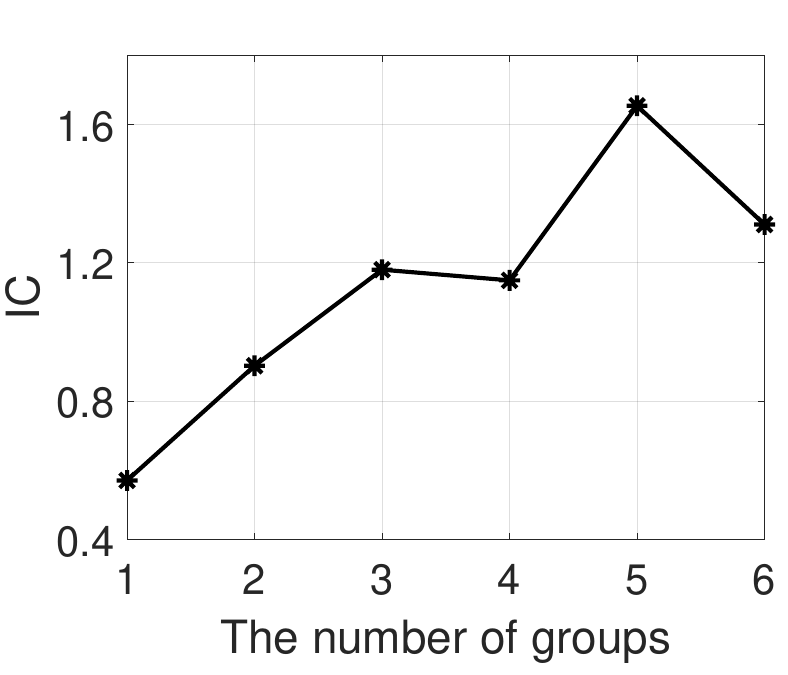}\\
				\\ (a) Choosing the number and type of breaks & (b) Choosing the number of groups
			\end{tabular}
		\end{center}
	\end{figure}
	
	\subsubsection*{Breakpoint and spillover estimates}
	We proceed to estimate the single break. Both models, whether including or excluding a break in the private effect, identify 2009 as the estimated breakpoint. This estimated breakpoint precisely matches the significant event of the global financial crisis and great recession. The crisis, which began developing in the U.S. in 2007 and was subsequently amplified by the European debt crisis from 2009 onward, impacted most of the countries in our sample. The estimated breakpoint aligns with documented evidence showing that innovation activities declined remarkably following the financial crisis of the late 2000s. For instance, \citet{OECD2009} reports that domestic and foreign companies listed on U.S. stock markets reduced their R\&D expenditures by 6.6\% in the first quarter of 2009.
	The reduced R\&D investment can be explained by tighter financial constraints caused by the crises \citep{campello2010, peia2022}. It is therefore unsurprising that R\&D expenditure and its correlation with total factor productivity (TFP) changed after the crisis. Additionally, the crisis likely altered international relationships through various channels, such as foreign direct investment (FDI) and trade, which, in turn, shifted R\&D spillover channels over time. Further explanations of the estimated breakpoint will be provided after discussing the spillover effects before and after the break.

	\begin{figure}[htp]\caption{Heatmap of R\&D spillover in the two regimes}\label{fig:empirical-network} 
		\begin{center}
			\begin{tabular}{cc}
				\includegraphics[scale=0.55]{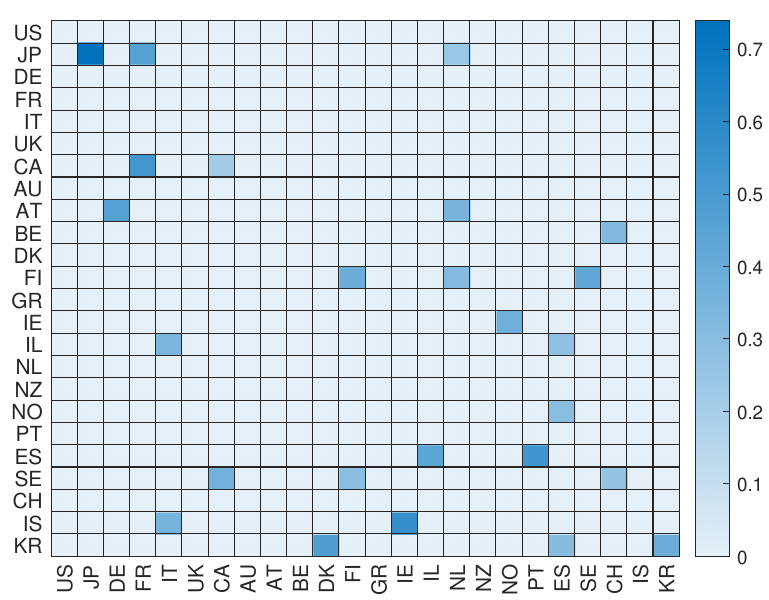} &
				\includegraphics[scale=0.55]{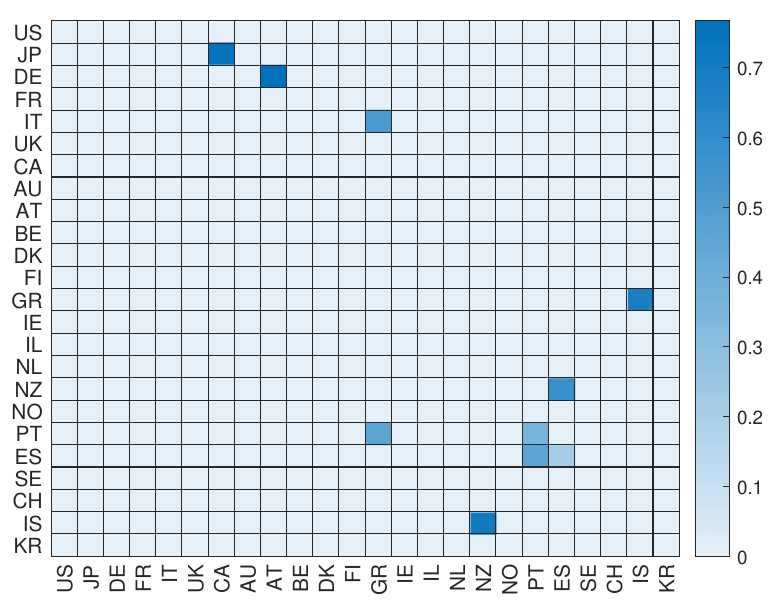}\\
				Before the break & After the break			
		\end{tabular}\end{center}
	\end{figure}

	Figure~\ref{fig:empirical-network} presents the heatmap of the estimated adjacency matrix across two distinct regimes. The results reveal that the short-run spillover structure is predominantly sparse, with only a limited number of countries exhibiting connections with others. Notably, this network becomes even more sparse following the structural break, with the network density, defined as the ratio of nonzero entries to the total possible edges, decreasing by around 51\% after the break. We discuss the spillover effects in the two regimes in turn.
	
	Before the structural break, approximately 2/3 of the countries generated R\&D spillovers, with six countries, namely France, Italy, Canada, the Netherlands, and Spain, emerging as key sources, transmitting their R\&D benefits to multiple destinations. Around half of the countries received spillovers from various sources, with Japan, Korean, and Sweden standing out as the largest beneficiary. Japan gained not only from its domestic R\&D efforts but also from those of Western European nations, e.g., France and the Netherlands. Similarly, Korea profited from its own R\&D activities and those of several European partners including Denmark and Spain. These findings underscore the strong connections between these European countries and their significant Asian allies. The identified spillover channels align with the substantial imports of machinery and transport equipment—goods that embody the highest level of technology—by Japan and Korea from Europe. Sweden receives the most R\&D spillover from Switzerland, Finland, and Canada. Notably, while the growth literature widely recognizes the private R\&D effect on a country's economic growth, our analyses indicate substantial variability across nations. For instance, strong private R\&D effects are evident in Japan, Finland, and Korea; moderate effects are observed in Canada and Spain, while for some countries, such effects remain too weak to be detected.

	Following the structural break identified in 2009, both the network structure and the magnitude of R\&D spillover effects exhibit remarkable changes. Notably, the post-break period is characterized by a more sparse linkage structure. Only about half of the countries in our sample continue to generate and receive R\&D spillovers, and most of these countries are connected to only a single partner, either as an influencer or recipient. This stands in sharp contrast to the pre-break period, which featured several major generators and recipients that interact with multiple countries. 
	In fact, a considerable number of pre-break R\&D spillover channels originated from key European countries, such as Germany, France, the Netherlands, and Italy, dissipate after the break. 
	
	\subsubsection*{Understanding the change of R\&D spillover structure}
	Further examination reveals that the R\&D expenditure greatly declined for many European countries, particularly in the years immediately following the structural break. This trend aligns with various institutional reports documenting a sharp drop in business R\&D and innovation in most OECD countries around 2009. The fact that the breakpoint year of 2009 coincides with the onset of the European debt crisis lends strong support to the theory and much evidence that financial crises can have profound and potentially long-lasting negative impacts on innovation. The negative impact of crises on R\&D expenditure can find its micro foundation, primarily through the channel of tighter financial constraints. Specifically, financial crises are known to create tighter financial constraints for firms. For instance, 
	\citet{campello2010}, using survey data from a broad range of firms in North America, Europe, and Asia collected at the end of 2008, show that firms worldwide planned significant budget reductions across nearly all policy areas for 2009 in response to the crisis. Especially, the survey reveals a plan of cut in technology spending by over 10\% in American and European firms during 2009.
	Furthermore, financial constraints are widely recognized as a critical factor influencing innovation. For example, \citet{Himmelberg1994} identify a strong, significant relationship between internal financing and R\&D activities for small high-tech firms in the U.S. 
	Considering the exacerbation of these constraints during financial crises, \citet{peia2022} argue that firms facing higher risks of financial strain and greater reliance on external financing tend to cut R\&D investments disproportionately during crisis periods. Smaller and private firms with weaker balance sheets are also particularly prone to reducing R\&D expenditures during financial disruptions.
	Historical evidence further supports this view. \citet{Babina2022}, using a century of U.S. patent data, document significant and sustained declines in patenting following the Great Depression due to disruptions in financial access, especially among young and inexperienced inventors. \citet{hardy-sever2021} reach similar conclusions, showing persistent reductions in R\&D investment during banking crises using industry-level data from developed countries.
	This body of evidence aligns with our observation of reduced R\&D expenditure in the years following the structural break. It also helps explain our finding that the primary R\&D spillover generators in the post-break period were largely countries less affected by the financial crisis, such as Canada and New Zealand.

	While there is substantial empirical evidence on how financial crises impact innovation, studies specifically examining the effects of crises on R\&D spillovers are scarce. Our analysis shows that the structure of R\&D spillovers became more sparse following the global financial crisis. This finding can be attributed partly to the decline in innovation during the crisis, as discussed earlier, and partly to the adverse effects of the crisis on cross-border economic interactions. For instance, numerous studies and reports have highlighted a sharp decline in international trade starting in late 2008. This decline was driven by multiple factors, including tighter credit conditions, disproportionate fall in domestic demand, and decreased output and trade of capital goods \citep{OECD2010outlook}. Also, severely impacted by the crisis is the FDI, as the liquidity constraints triggered by the crisis affected both firm owners and potential foreign buyers, leading to a marked decrease in cross-border mergers and acquisitions \citep{OECD2010outlook}.
	Given that cross-border economic activities are important channels for R\&D spillovers, it is not surprising that the spillover structure became more sparse due to the reduced bilateral interactions prompted by the financial crisis.

	To better understand the estimated latent spillover structure, we regress the estimated adjacency matrix on several bilateral variables commonly used to construct the adjacency matrix. These include the geographic distance (Dist), a dummy variable indicating whether two countries share a border (Contg),  a dummy variable for countries with a common official language (Lang), the trade flow (Trade) and foreign direct investment (FDI) from CEPII.\footnote{The trade flow and FDI are averaged over time in the pre- and post-break regimes, respectively, and are expressed in units of one million dollars. Distance is expressed in units of 1,000 km. Data source: http://www.cepii.fr/CEPII/en/bdd\_modele/bdd\_modele.asp}  
	Given that the adjacency matrix is bounded below by zero, we estimate a Tobit regression for each regime and present the results in Table~\ref{emp:tobit}. Also reported are the OLS estimates. Overall, we find that geographic and demographic variables partially explain the estimated adjacency matrix, but substantial variation remains unexplained. The effects of these variables also vary remarkably across regimes. In the pre-break regime, the Tobit model produces weakly significant estimates (at 15\%) of geographic distance, contiguous and common language dummies, suggesting R\&D spillover occurs more likely between countries that are closely located, share the common border and language. The bilateral economic activities explain little variation of the adjacency matrix. In the post-break regime, the significance of all proxies decreases except the contiguous dummy, partly there is less variation in the adjacency matrix due to more sparsity. This result shows that the financial crisis disrupt R\&D spillover with spillover remaining only among the highly close countries. Furthermore, the model fit indicates that, in both regimes, the short-term spillover structure cannot be fully explained by geographic, cultural, and economic variables alone, let alone by a single observable. Therefore, it is valuable to consider data-driven methods to uncover latent R\&D spillover structures.
	\begin{table}[htp]
		\begin{center}
			\caption{Regression of estimated adjacency matrix on covariates}\label{emp:tobit}\small
			\begin{tabular}{lrrrrrrrrrrrr}
				\hline\hline
				& \multicolumn{5}{c}{Tobit} &&  \multicolumn{5}{c}{OLS}\\
				
				& \multicolumn{2}{c}{Before} && \multicolumn{2}{c}{After} &&  \multicolumn{2}{c}{Before} && \multicolumn{2}{c}{After}\\
				&Coef.  & $p$-val.   && Coef.  & $p$-val.  &&Coef.  & $p$-val.   && Coef.  & $p$-val.  \\
				\cline{2-6}\cline{8-12}
				Dist  &  $-0.035$  &  0.117  &&  $ 0.013$  &  0.734    &&     $-0.001$   &  0.565    &&  $-0.001$  &  0.604 \\
				Contg &  $ 0.501$  &  0.144  &&  $ 1.288$  &  0.134    &&     $ 0.024$   &  0.164    &&  $ 0.010$  &  0.633 \\
				Lang  &  $ 0.385$  &  0.133  &&  $-0.511$  &  0.520    &&     $ 0.015$   &  0.223    &&  $ 0.008$  &  0.588 \\
				Trade &  $-0.007$  &  0.629  &&  $-0.000$  &  0.965    &&     $-0.000$   &  0.588    &&  $ 0.000$  &  0.994 \\
				FDI   &  $-19.11$  &  0.283  &&  $-3.846$  &  0.895    &&     $-0.124$   &  0.578    &&  $ 0.076$  &  0.875 \\
				Const &  $-1.217$  &  0.000  &&  $-3.449$  &  0.005    &&     $ 0.008$   &  0.158    &&  $ 0.001$  &  0.832 \\
				\hline
			\end{tabular}
		\end{center}
	\end{table}

	\subsubsection*{Estimating (heterogeneous) private effect of human capital}
	The role of human capital in productivity has garnered considerable attention in the growth literature. While human capital is generally regarded as beneficial for productivity, the magnitude and significance of the effect often vary across studies, depending on the specific samples analyzed, the quality of data, and the measures of educational outcomes used. For example, \citet{pritchett2001} finds that increases in educational attainment did not lead to corresponding growth in productivity in many developing countries, attributing this to poor institutional and educational quality. 
	In a recent review, \citet{psacharopoulos2018}, drawing on a meta-analysis of 1,120 estimates in 139 countries from 1950 to 2014, highlight that returns of education can be low or insignificant, especially in contexts where education does not translate into relevant skills or productive employment. 
	
	We revisit the role of human capital using Model~\eqref{eq:empirical-model}, controlling for time-varying R\&D spillovers. Our diagnostic analysis with IC indicates that the private effect of human capital remains largely constant after the structural break (see Figure~\ref{fig:empirical-IC}(a)). Thus, we estimate a time-invariant effect of human capital using the DML, while allowing R\&D spillovers to shift at the estimated breakpoint. We begin by estimating a homogeneous effect of human capital across countries, and will later explore potential heterogeneity. The DML estimate for the human capital effect is approximately 0.134 but not statistically significant. This result is in line with 
	\citet{Ertur2017}, who, using the same set of OECD countries (but over a shorter time horizon) and variables, also report an insignificant human capital effect and of the similar magnitude\footnote{Their estimates range from $-0.45$ to 0.23 across different subsamples and specifications.} after accounting for cross-sectional dependence and R\&D spillovers.
	As we measure human capital by years of schooling following the literature, our finding also aligns with the broader insight from the growth empirics that the quantity of education alone does not necessarily drive growth; rather, the quality of education plays a crucial role. For example, \citet{Islam1995} found an insignificant effect of education on growth using panel data from various samples, including OECD countries, attributing this to differences in the quality of education. The author also points out the disparity between panel data and cross-sectional analyses, with the former focusing on temporal relationships and often reporting an insignificant link between human capital and growth, whereas the latter often showing a positive relationship. \citet{Hanushek&Kimko2000}  also underscore that educational quality is more important for growth than educational quantity. In the similar spirit, \citet{Vandenbussche2006}, using panel data from OECD countries, argue that only skilled human capital contributes to growth, rather than the total human capital. They also find that the impact of human capital on productivity is weaker for the countries further from technology frontier. These arguments partly explain the insignificant estimate of human capital effect in our context after controlling for R\&D effects.

	Next, we examine potential heterogeneity in the human capital effect. While allowing for country-specific effects seems conceptually appealing, it suffers from the incidental parameter issue and estimation instability. 
	Traditional approaches often involve interacting human capital with an observable, such as a dummy indicating a country's G7 membership \citep{Coe2009, Ertur2017}, openness, physical capital, or labor force \citep{Miller2000}. However, the cross-country heterogeneity is likely to be driven by multiple factors jointly, both observed and unobserved, and the postulated group structure (i.e., which countries belong to which group) may not coincide with the true heterogeneity pattern. 
	To flexibly model cross-country heterogeneity and minimize potential misspecification, we adopt the latent group structure model described in \eqref{eq:model-group-delta}, assuming that the human capital effect is homogeneous within each group. Importantly, we do not impose a predefined group structure; instead, we estimate it based on the data. The number of groups $G$ is determined using the Bayesian IC (BIC) similar to \eqref{eq:IC} \citep[see, e.g.,][]{bonhomme&manresa2015,okui&wang2018}. By applying the SBSA method from \citet{Wang2021} with $G$ ranging from 1 to 6, we plot the IC in Figure~\ref{fig:empirical-IC}(b). It shows that $G=1$ is the best specification, confirming the reliability of the estimate obtained under the homogeneity assumption.
	
	To assess the robustness of our results, we re-estimate the spillover and human capital effects using $G=2$, the second-best specification indicated by the IC, which allows for some degree of heterogeneity. The resulting human capital effects are positive in both groups, but of distinct magnitudes. However, these estimates remain insignificant, which supports the IC result that favors the homogeneity assumption.
	Compared to the estimates under $G=1$, a slightly higher number of spillover links are detected in both the pre- and post-break periods. Nonetheless, the overall pattern of R\&D spillovers becoming more sparse after the break remains salient. More detailed results are available in Appendix \ref{app:two-groups}.

	\section{Conclusion}\label{sec:conclusion}
	
	This paper investigates the impact of financial crises on R\&D spillovers, introducing a novel model and estimation method for analyzing time-varying spillover effects without pre-specifying a linkage structure. Our approach recovers the latent spillover network and identifies the unknown breakpoint at which spillover and/or private effects experience a shift. We establish the super-consistency of the estimated breakpoint, which guarantees that the estimators of spillover and private effects under unknown breakpoints are asymptotically equivalent to those obtained under true breakpoints. We also establish the $\sqrt{NT}$-consistent DML estimator of the private effect, despite the fact that the spillover effects are estimated at a slower rate. 
	An analysis of data from 24 OECD countries over a 38-year period reveals that the cross-country R\&D spillover network became sparser following the 2009 financial crisis, suggesting a decline in innovation and disruption to its spillover mechanisms. 
	The ubiquity of spillovers in various economic and social contexts, such as corporate performance, crime, and educational outcomes, extends the applicability of our approach beyond the domain of technology diffusion.

	There are numerous avenues for future research, and we highlight two here. First, while we have demonstrated the super-consistency of our breakpoint estimator, this result does not enable statistical inference for the true breakpoint, as the asymptotic distribution of the estimator is not provided. According to our findings, the singleton set $\{\tilde{b}\}$ serves as a valid confidence set for $b^0$ at any confidence level, since $\Pr(\tilde{b} = b^0) > 1 - \alpha$ asymptotically for any $\alpha > 0$. However, such a confidence set fails to capture the statistical uncertainty associated with $\tilde{b}$. Developing methods for statistical inference on breakpoints in our framework would require entirely different approaches from those used for estimation.
	Second, our model assumes that the random variables in the system ($x_{it}$, $z_{it}$, and $u_{it}$) are only weakly dependent over time. However, in some applications, variables may exhibit strong serial dependence, such as in the case of random walk processes. Exploring spillover effects involving variables with unit roots or those within cointegrated systems containing multiple integrated variables would be both theoretically insightful and practically valuable.

	\section*{Declaration of generative AI and AI-assisted technologies in the writing process}

	During the preparation of this work the authors used Grammarly in order to improve the readability and language of the manuscript. After using this tool/service, the authors reviewed and edited the content as needed and takes full responsibility for the content of the published article.

	\bibliography{reference}
	\bibliographystyle{abbrvnat}

	\clearpage

	\clearpage
	\appendix

	\begin{center}
		{{\Large Supplementary appendix to\\} }	
		{\large Recovering latent linkage structures and spillover effects with structural breaks in panel data models}
	\end{center}

	%\section{Mathematical appendix}
	
	This appendix presents the mathematical proofs of the asymptotic results and additional empirical results. Throughout the proofs, $C$ represents a generic constant, the value of which may vary in different contexts. We abbreviate ``with probability approaching one'' as ``w.p.a.1''. To reduce the number of parentheses in the notation, we denote $ E((\cdot)^2) $ as $ E(\cdot)^2$.

	\section{Properties of the preliminary estimator}
	
	This section presents the proof of the property of the preliminary estimators derived in Step 1.
	Recall that we define the excess risk as 
	\begin{align*}
		R (\beta, b):= \frac{1}{NT}\sum_{i=1}^N \sum_{t=1}^T E \left(y_{it} - W_{it} (b)' \beta_i \right)^2   - \frac{1}{NT}\sum_{i=1}^N \sum_{t=1}^T E\left(y_{it} - W_{it} (b^0)' \beta_i^0  \right)^2.
	\end{align*}

	\begin{proof}[Proof of Lemma \ref{theorem:risk-consistency}]
		Note that by Assumption \ref{as-exogeneity},
		\begin{align*}
			R (\beta, b) = \frac{1}{NT}\sum_{i=1}^N \sum_{t=1}^T E  \left( W_{it} (b^0)' \beta_i^0 - W_{it} (b)' \beta_i \right)^2 .
		\end{align*}	
		Let
		\begin{align*}
			\nu_{NT} ( \beta, b) := \frac{1}{NT} \sum_{i=1}^N \sum_{t=1}^T \left[ \left(y_{it} - W_{it} (b)' \beta_i \right)^2 - E \left(y_{it} - W_{it} (b)' \beta_i \right)^2 \right].
		\end{align*}
		By the definition of $(\hat \beta, \hat b)$ in \eqref{eq:obj-break-point-pls}, it holds that $V_{NT} (\hat \beta, \hat b) \leq V_{NT} (\beta^0, b^0)$. Thus, we have 
		\begin{align}
			R(\hat \beta, \hat b)\label{eq:risk-inequality}
			&\leq  \nu_{NT} ( \beta^0, b^0) - \nu_{NT} ( \hat \beta, \hat b) + D(\beta^0, b^0) -  D(\hat \beta, \hat b)  \nonumber\\
			&=\left( \nu_{NT} ( \beta^0, \hat b) - \nu_{NT} ( \hat \beta, \hat b) \right) + \left(\nu_{NT} ( \beta^0, b^0) - \nu_{NT} (  \beta^0, \hat b) \right) \\
			&\quad + D(\beta^0, b^0) -  D(\hat \beta , \hat b) .\nonumber
		\end{align}
		We examine the three parts in~\eqref{eq:risk-inequality} separately. First, we show that the second part $\nu_{NT} ( \beta^0, b^0) - \nu_{NT} (  \beta^0, \hat b)$ is small.
		We observe that for $b \in \mathcal{T}$ and $b > b^0$, 
		\begin{align*}
			& \nu_{NT} ( \beta^0, b^0) - \nu_{NT} (  \beta^0, b) \\
			= &  \frac{1}{NT} \sum_{i=1}^N \sum_{t=1}^T \left(u_{it}^2 - E(u_{it})^2\right) \\
			& -  \frac{1}{NT} \sum_{i=1}^N \sum_{t=1}^T \left( \left[\big(W_{it}(b^0) - W_{it} (b)\big)' \beta_i^0  + u_{it} \right]^2 - E \left[\big(W_{it} (b^0) - W_{it} (b)\big)' \beta_i^0 + u_{it} \right]^2 \right) \\
			=& -  \frac{1}{NT} \sum_{i=1}^N \sum_{t=b^0+1}^{b} \left( \left[\big(W_{it}(b^0) - W_{it} (b)\big)' \beta_i^0 \right]^2 - E \left[\big(W_{it} (b^0) - W_{it} (b)\big)' \beta_i^0 \right]^2 \right) \\
			& - \frac{2}{NT} \sum_{i=1}^N \sum_{t=b^0+1}^{b} u_{it} \left[W_{it}(b^0) - W_{it} (b)\right]' \beta_i^0 \\
			= &-  \frac{1}{NT} \sum_{i=1}^N \sum_{t=b^0+1}^{b} \left( \left[ X_{it}' (\gamma_{i,A}^0 - \gamma_{i,B}^0)\right]^2 - E \left[X_{it}' (\gamma_{i,A}^0 - \gamma_{i,B}^0) \right]^2\right) \\
			& - \frac{2}{NT} \sum_{i=1}^N \sum_{t=b^0+1}^{b} u_{it}X_{it}' (\gamma_{i,A}^0 - \gamma_{i,B}^0) \\
			=&  O_p\left( \frac{1}{\sqrt{NT}}\right),
		\end{align*}
		where the last equality follows by Chebyshev inequality under Assumptions \ref{as-exogeneity}, \ref{as-tail}, \ref{as-mixing}, and \ref{as-compactness} in combination with Lemmas \ref{lem:exp-tail} and \ref{lem:mixing-g} in the Online Supplement.
		We can obtain the same results for $b< b^0$. Thus, we have 
		\begin{align*}
			\nu_{NT} ( \beta^0, b^0) - \nu_{NT} (  \beta^0, \hat b)  = O_p \left( \frac{1}{\sqrt{NT}}\right).
		\end{align*}
		
		Next, we examine the first part of~\eqref{eq:risk-inequality}. Note that for any $b\in\{1,\ldots,T\}$, we have
		\begin{align}
			& \nu_{NT} ( \beta^0,   b) - \nu_{NT} ( \beta, b) \nonumber \\
			=&   \frac{1}{NT} \sum_{i=1}^N \sum_{t=1}^T \left( \left[\big(W_{it}(b^0) - W_{it} (b)\big)' \beta_i^0  + u_{it} \right]^2 - E \left[\big(W_{it} (b^0) - W_{it} (b)\big)' \beta_i^0 + u_{it} \right]^2 \right) \nonumber \\
			& -  \frac{1}{NT} \sum_{i=1}^N \sum_{t=1}^T \left( \left[W_{it}(b^0)'\beta_i^0  - W_{it} (b)'\beta_i  + u_{it} \right]^2 - E \left[W_{it} (b^0)'\beta_i^0  - W_{it} (b)' \beta_i + u_{it} \right]^2 \right) \nonumber  \\
			=& - \frac{2}{NT} \sum_{i=1}^N \sum_{t=1}^T \Big( W_{it}(b^0)'\beta_i^0 W_{it} (b)' \left(\beta_i^0 - \beta_i\right)  - E \left[W_{it}(b^0)'\beta_i^0 W_{it} (b)' \left(\beta_i^0 - \beta_i\right) \right] \Big)\nonumber \\
			& +  \frac{1}{NT} \sum_{i=1}^N \sum_{t=1}^T \left( \left[W_{it}(b)'\beta_i^0\right]^2  - \left[W_{it} (b)'\beta_i \right]^2 - E\left[(W_{it} (b)'\beta_i^0)^2  - (W_{it} (b)' \beta_i )^2 \right]\right) \nonumber \\
			& + \frac{2}{NT} \sum_{i=1}^N \sum_{t=1}^T u_{it} W_{it} (b)' (\beta_i - \beta_i^0).\label{eq:nu-diff-1}
		\end{align}
		The right-hand side of \eqref{eq:nu-diff-1} further consists of three parts. For the first part, note that 
		\begin{align*}
			& \left|	\frac{1}{NT} \sum_{i=1}^N \sum_{t=1}^T \Big( W_{it}(b^0)'\beta_i^0 W_{it} (b)' (\beta_i^0 - \beta_i)  - E \left[W_{it}(b^0)'\beta_i^0 W_{it} (b)' (\beta_i^0 - \beta_i) \right] \Big) \right|  \\
			=&\left| \frac{1}{NT} \sum_{i=1}^N \sum_{t=1}^T \sum_{k=1}^p  \Big( W_{it}(b^0)'\beta_i^0 W_{itk} (b) - E\left[  W_{it}(b^0)'\beta_i^0 W_{itk} (b)\right]\Big) (\beta_{ik}^0 - \beta_{ik})\right|   \\
			\leq & \frac{1}{N} \sum_{i=1}^N \sum_{k=1}^p \left| \frac{1}{T} \sum_{t=1}^T  \Big( W_{it}(b^0)'\beta_i^0 W_{itk} (b) - E\left[  W_{it}(b^0)'\beta_i^0 W_{itk} (b) \right]\Big) \right| \cdot \left|  \beta_{ik} - \beta_{ik}^0  \right| ,
		\end{align*}
		where $W_{itk} (b)$ is the $k$-th element of $W_{it} (b)$.
		Applying Lemma \ref{lem:fuk-nagaev-dependent-ij-b}, whose conditions are satisfied under Assumptions  \ref{as-exogeneity}--\ref{as-compactness}, as established by Lemmas \ref{lem:exp-tail} and \ref{lem:mixing-g}, w.p.a.1, we have
		\begin{align*}
			\max_b \max_i \max_k  \left| \frac{1}{T} \sum_{t=1}^T \Big( W_{it}(b^0)'\beta_i^0 W_{itk} (b) - E\left[  W_{it}(b^0)'\beta_i^0 W_{itk} (b)\right] \Big)\right| \leq C \frac{\log (Np)}{\sqrt{T}},
		\end{align*}
		and thus
		\begin{align*}
			& \left| \frac{1}{NT} \sum_{i=1}^N \sum_{t=1}^T \Big( W_{it}(b^0)'\beta_i^0 W_{it} (b)' (\beta_i^0 - \beta_i)  - E \left[W_{it}(b^0)'\beta_i^0 W_{it} (b)' (\beta_i^0 - \beta_i) \right] \Big) \right| \\ 
			\leq & C \frac{\log (Np)}{\sqrt{T}} \frac{1}{N} \sum_{i=1}^N \sum_{k=1}^p \left|  \beta_{ik} - \beta_{ik}^0  \right|.
		\end{align*}
		A similar argument can be used to bound the third part of \eqref{eq:nu-diff-1}. Specifically, we have 
		\begin{align*}
			\left| \frac{1}{NT} \sum_{i=1}^N \sum_{t=1}^T u_{it} W_{it} (b)' (\beta_i - \beta_i^0) \right| 
			\leq  \frac{1}{N} \sum_{i=1}^N \sum_{k=1}^p \left| \frac{1}{T} \sum_{t=1}^T u_{it} W_{itk} (b) \right| \cdot \left|  \beta_{ik} - \beta_{ik}^0  \right| .
		\end{align*}
		By Lemma \ref{lem:fuk-nagaev-dependent-ij-b}, w.p.a.1, 	
		$
			\max_b \max_i \max_k  \left|1/T \sum_{t=1}^T  u_{it} W_{itk} (b) \right| \leq C \log (Np)/\sqrt{T},
		$
		and thus
		\begin{align*}
			\left| \frac{1}{NT} \sum_{i=1}^N \sum_{t=1}^T u_{it} W_{it} (b)' (\beta_i - \beta_i^0) \right| \leq C \frac{\log (Np)}{\sqrt{T}} \frac{1}{N} \sum_{i=1}^N \sum_{k=1}^p \left|  \beta_{ik} - \beta_{ik}^0  \right|.
		\end{align*}
		
		Then for the second part of~\eqref{eq:nu-diff-1}, we first note that 
		\begin{align*}
			& \left|  \frac{1}{NT} \sum_{i=1}^N \sum_{t=1}^T \Big( [W_{it}(b)'\beta_i^0]^2  - [W_{it} (b)'\beta_i]^2 - E \left[(W_{it} (b)'\beta_i^0)^2  - (W_{it} (b)' \beta_i )^2 \right]\Big)\right| \\
			= & \left|  \frac{1}{NT} \sum_{i=1}^N \sum_{t=1}^T \Big(\beta_i^0 + \beta_i \Big)' \Big(W_{it} (b) W_{it}(b)' - E\left[W_{it} (b) W_{it}(b)'\right]\Big) \Big(\beta_i^0 - \beta_i\Big)\right| \\
			=&  \left|  \frac{1}{NT} \sum_{i=1}^N \sum_{t=1}^T \sum_{k=1}^p \sum_{k'=1}^p \Big(\beta_{ik'}^0 + \beta_{ik'} \Big) \Big(  W_{itk'} (b) W_{itk}(b) - E\left[ W_{itk'} (b) W_{itk}(b)\right]\Big) \Big(\beta_{ik}^0 - \beta_{ik}\Big) \right| \\
			\leq &\frac{C}{N} \sum_{i=1}^N \sum_{k=1}^p \sum_{k'=1}^p \left| \frac{1}{T}\sum_{t=1}^T \Big(  W_{itk'} (b) W_{itk}(b) - E\left[W_{itk'} (b) W_{itk}(b)\right]\Big) \right| \cdot | \beta_{ik}^0 - \beta_{ik} |.
		\end{align*}
		Further noting that  Lemma~\ref{lem:fuk-nagaev-dependent-b} implies 
		\begin{align*}
			\max_b \max_i \max_k \max_{k'} \left|\frac{1}{T}\sum_{t=1}^T \Big(  W_{itk'} (b) W_{itk}(b) - E[W_{itk'} (b) W_{itk}(b)]\Big) \right| \leq C \frac{\log (Np^2)}{\sqrt{T}},\quad \textrm{w.p.a.1},
		\end{align*}
	 	we thus have 
		\begin{align*}
			& \left|  \frac{1}{NT} \sum_{i=1}^N \sum_{t=1}^T \Big( [W_{it}(b)'\beta_i^0]^2  - [W_{it} (b)'\beta_i ]^2 - E \left[(W_{it} (b)'\beta_i^0)^2  - (W_{it} (b)' \beta_i )^2 \right]\Big)\right|  \\
			\leq & C \frac{\log (Np^2)}{\sqrt{T}} \frac{p}{N} \sum_{i=1}^N \sum_{k=1}^p \left|  \beta_{ik} - \beta_{ik}^0  \right|.
		\end{align*}
		Combining the results of the three parts of~\eqref{eq:nu-diff-1}, we can obtain that 
		\begin{align*}
			\nu_{NT} ( \beta^0, \hat b) - \nu_{NT} ( \hat \beta, \hat b) \leq  C\frac{\log (Np^2)}{\sqrt{T}} \frac{p}{N} \sum_{i=1}^N \sum_{k=1}^p \left|  \beta_{ik} - \beta_{ik}^0  \right|, \quad \textrm{w.p.a.1}.
		\end{align*}

		Depending on how close $\hat \beta$ is to $\beta^0$, we consider two cases, namely $|\hat \beta - \beta^0|_1 \leq |\beta^0|_1$ and $|\hat \beta - \beta^0|_1> |\beta^0|_1$.
		In the first case where $|\hat \beta - \beta^0|_1 \leq |\beta^0|_1$, w.p.a.1, 
		\begin{align*}
			\nu_{NT} ( \beta^0, \hat b) - \nu_{NT} ( \hat \beta, \hat b) \leq  C\frac{\log (Np^2)}{\sqrt{T}} \frac{p}{N} \sum_{i=1}^N \sum_{k=1}^p \left|  \beta_{ik} - \beta_{ik}^0  \right| \leq  C\frac{\log (Np^2)}{\sqrt{T}} \frac{p}{N}|\beta^0|_1.
		\end{align*}
		Moreover, if $|\hat \beta - \beta^0|_1 \leq |\beta^0|_1$, we have $|\hat \beta |_1 \leq |\hat \beta - \beta^0|_1 + |\beta^0|_1 \leq 2 |\beta^0|_1 $. This implies that, for the third part of~\eqref{eq:risk-inequality}, we have
		$
			| D(\beta^0, b^0) - D(\hat \beta, \hat b)| \leq 3 D_{NT} |\beta^0|_1.
		$
		Thus, we can obtain the excess risk associated with $( \hat \beta, \hat b)$ as
		\begin{align}\label{eq:risk-case1}
			R( \hat \beta, \hat b) \leq & O_p\left( \frac{1}{\sqrt{T}} \right) + \left( C \frac{\log (Np^2)}{\sqrt{T}} \frac{p}{N} + 3 D_{NT} \right) |\beta^0|_1 \nonumber \\
			\leq &  O_p\left( \frac{1}{\sqrt{T}} \right) + \left( C \frac{\log (Np^2)}{\sqrt{T}} \frac{p}{N} + 3 D_{NT} \right) s M \nonumber \\
			=& O_p \left( s D_{NT}  \right),
		\end{align}
		by Assumptions \ref{as-compactness} and \ref{as-penalty-order}, and by noting that $p=O(N)$ and $s \log (Np^2) / \sqrt{T} $ is of larger order than $1/\sqrt{NT}$.
		
		In the second case where $|\hat \beta - \beta^0|_1 > |\beta^0|_1$, we partition a nonzero parameter $\beta$ as $\beta=(\beta_J,\beta_{J^c})$, such that $\beta_J=\{\beta_j: j\in J(\beta^0)\}$ and $\beta_{J^c}=\{\beta_j: j\notin J(\beta^0)\}$. We then have
		\begin{align*}
			R( \hat \beta, \hat b) \leq &O_p\left( \frac{1}{\sqrt{NT}} \right) + \left( C \frac{\log (Np^2)}{\sqrt{T}} \right) \frac{p}{N} |\hat\beta - \beta^0|_1 +  D(\beta^0, b^0) - D(\hat \beta, \hat b)\\
			\leq & O_p\left( \frac{1}{\sqrt{NT}} \right) + D_{NT} |\hat\beta - \beta^0|_1 +  D(\beta^0 , b^0) - D(\hat \beta, \hat b) \\
			\leq & O_p\left( \frac{1}{\sqrt{NT}} \right) + D_{NT} |(\hat\beta - \beta^0)_J|_1 +  D(\beta^0, b^0) + D_{NT} |\hat \beta_{J^c}|_1 - D(\hat \beta_{J^c}, \hat b),
		\end{align*}
		where the second inequality follows from Assumption \ref{as-penalty-order} and $p=O(N)$, and the last inequality uses the decomposition $ \hat\beta - \beta^0 = (\hat \beta - \beta^0 )_J + \hat \beta_{J^c}$. By Assumption \ref{as-compactness}, it holds that $|(\hat\beta - \beta^0)_J|_1 \leq 2 s M$ and $D(\beta^0, b^0) \leq D_{NT} s M$. Clearly, $ D_{NT} |\hat \beta_{J^c}|_1 - D(\hat \beta_{J^c}, \hat b)\geq 0$. Thus, we have 
		\begin{align}\label{eq:risk-case2}
			R( \hat \beta, \hat b) \leq O_p\left( \frac{1}{\sqrt{NT}} \right) + 3 s M D_{NT} = O_p (sD_{NT}),
		\end{align}
		where the last equality follows by noting that $ s D_{NT} \geq s \log (Np^2) / \sqrt{T} $ is of larger order than $1/\sqrt{NT}$ and also from Assumptions \ref{as-compactness} and \ref{as-penalty-order}.

		Combining the results in~\eqref{eq:risk-case1} and~\eqref{eq:risk-case2}, and using Assumption \ref{as-penalty-order}, we have
		\begin{align*}
			R( \hat \beta, \hat b) = O_p \left( s D_{NT} \right).
		\end{align*}
		
	\end{proof}

	%%%%%%%%%%%%%%%%%
	Next, we show the rate of convergence of the preliminary breakpoint estimator.
	
	\begin{lemma}\label{th-b-initial-rate}
		Suppose that Assumptions \ref{as-exogeneity}--\ref{as-break-date} are satisfied. Then
		\begin{align*}
			|\hat b - b^0 | = O_p ( s T D_{NT} ).
		\end{align*}
	\end{lemma}
	
	\begin{proof}
		We consider two cases, $b > b^0$ and $b < b^0$. In each case, we can decompose the risk as follows. When $b > b^0$, we have
		\begin{align}\label{eq:R:b>b0}
			R (\beta, b) =& \frac{1}{NT}\sum_{i=1}^N \sum_{t=1}^{b^0} E ( W_{it}' \beta_{i,B}^0 - W_{it}' \beta_{i,B} )^2   + \frac{1}{NT}\sum_{i=1}^N \sum_{t=b^0+1}^{b} E ( W_{it} ' \beta_{i,A}^0  - W_{it}' \beta_{i,B} )^2 \nonumber\\
			& + \frac{1}{NT}\sum_{i=1}^N \sum_{t=b+1}^T E ( W_{it}' \beta_{i,A}^0  - W_{it}' \beta_{i,A})^2 .
		\end{align}	
		When $b < b^0$, we have
		\begin{align}\label{eq:R:b<b0}
			R (\beta, b) =& \frac{1}{NT}\sum_{i=1}^N \sum_{t=1}^{b} E  ( W_{it}' \beta_{i,B}^0 - W_{it}' \beta_{i,B} )^2   + \frac{1}{NT}\sum_{i=1}^N \sum_{t=b+1}^{b^0} E  ( W_{it} ' \beta_{i,B}^0  - W_{it}' \beta_{i,A} )^2  \nonumber\\
			& + \frac{1}{NT}\sum_{i=1}^N \sum_{t=b^0+1}^T E (  W_{it}' \beta_{i,A}^0 - W_{it}' \beta_{i,A})^2.
		\end{align}	
		We note that all terms on the right-hand side of~\eqref{eq:R:b<b0} and~\eqref{eq:R:b>b0} are non-negative. Let $\epsilon = C s D_{NT}$ for some $C>0$. Then, Lemma \ref{theorem:risk-consistency} implies that we can make $\Pr \left(R(\hat \beta, \hat b) < \epsilon\right)$ arbitrarily close to 1 by taking a sufficiently large $C$. Hence, in the following, we will focus on the analysis in the case where $R(\hat \beta, \hat b) < \epsilon$ holds. In this case, all terms on the right-hand sides of \eqref{eq:R:b>b0} and \eqref{eq:R:b<b0} evaluated at $(\hat \beta, \hat b)$ are less than $\epsilon$. Below, we show that when $R( \beta, b) < \epsilon$, we have $|b-b^0| < C'\epsilon T$ for some $C'$.

		We first consider the case where $b> b^0$. We aim to bound 
		$
			1/(NT)\sum_{i=1}^N \sum_{t=b^0+1}^{b} E  ( W_{it} ' \beta_{i,A}^0  - W_{it}' \beta_{i,B}^0 )^2
		$
		from both sides. Note that
		\begin{align}\label{eq:diff-coef-pre-post-1}
			&\frac{1}{NT}\sum_{i=1}^N \sum_{t=b^0+1}^{b} E  ( W_{it} ' \beta_{i,A}^0  - W_{it}' \beta_{i,B}^0 )^2  \nonumber\\
			\leq &  
			\frac{1}{NT}\sum_{i=1}^N \sum_{t=b^0+1}^{b} E  ( W_{it} ' \beta_{i,A}^0 - W_{it}' \beta_{i,B} )^2   
			+	\frac{1}{NT}\sum_{i=1}^N \sum_{t=b^0+1}^{b} E ( W_{it} ' \beta_{i,B}^0 - W_{it} '\beta_{i,B} )^2.
		\end{align}
		The first term on the right-hand side of~\eqref{eq:diff-coef-pre-post-1} is less than $ \epsilon$, since $R(\beta,b)<\epsilon$.
		For the second term of~\eqref{eq:diff-coef-pre-post-1}, we have
		\begin{align*}
			\frac{1}{NT}\sum_{i=1}^N \sum_{t=b^0+1}^{b} E ( W_{it} ' \beta_{i,B}^0 - W_{it} '\beta_{i,B} )^2 
			\leq  \frac{C}{NT}\frac{b-b^0}{b^0}\sum_{i=1}^N \sum_{t=1}^{b^0} E (  W_{it}'\beta_{i,B}^0 - W_{it} '\beta_{i,B} )^2    \leq C \epsilon \frac{b-b^0}{b^0},
		\end{align*}
		where the first inequality follows from Assumption \ref{as-semi-stationary}, and the last inequality follows from 
		$
		1/(NT)\sum_{i=1}^N \sum_{t=1}^{b^0} E  ( W_{it}' \beta_{i,B}^0 - W_{it}' \beta_{i,B} )^2 < \epsilon,	
		$
		which is further due to $R(\beta,b)<\epsilon$. 
		Thus, we can bound~\eqref{eq:diff-coef-pre-post-1} from above as
		\begin{align}\label{eq:upper-bound}
			\frac{1}{NT}\sum_{i=1}^N \sum_{t=b^0+1}^{b} E  ( W_{it} ' \beta_{i,A}^0 - W_{it}' \beta_{i,B}^0 )^2  < \left(1 + C\frac{b-b^0}{b^0}\right) \epsilon.
		\end{align}
		On the other hand, Assumption \ref{as-break-identification} implies that 
		\begin{align}\label{eq:lower-bound}
			&\frac{1}{NT}\sum_{i=1}^N \sum_{t=b^0+1}^{b} E  ( W_{it} ' \beta_{i,A}^0 - W_{it}' \beta_{i,B}^0 )^2  \nonumber\\
			= & 	\frac{1}{NT}\sum_{i=1}^N \sum_{t=b^0+1}^{b} E  \left( X_{it} '  \gamma_{i,A}^0 - X_{it}'\gamma_{i,B}^0\right)^2  \geq \frac{b-b^0}{T} \underline{m}.
		\end{align}
		Combining~\eqref{eq:upper-bound} and~\eqref{eq:lower-bound}, we have
		\begin{align*}
			\frac{b-b^0}{T} \underline{m} < \left(1 + C\frac{b-b^0}{b^0}\right) \epsilon,
		\end{align*}
		and therefore
		\begin{align*}
			b- b^0 < \epsilon \left( \frac{\underline{m}}{T} - \frac{C}{b^0} \epsilon \right)^{-1}. 
		\end{align*}
		
		Next, we consider the case of $b< b^0$. Similarly, we aim to bound $1/(NT)\sum_{i=1}^N \sum_{t=b+1}^{b^0} E  ( W_{it} ' \beta_{i,B}^0  - W_{it}' \beta_{i,A}^0)^2 $ from both sides. Note that
		\begin{align}\label{eq:diff-coef-pre-post-2}
			& \frac{1}{NT}\sum_{i=1}^N \sum_{t=b+1}^{b^0} E  ( W_{it} ' \beta_{i,B}^0  - W_{it}' \beta_{i,A}^0 )^2  \nonumber\\
			\leq &
			\frac{1}{NT}\sum_{i=1}^N \sum_{t=b+1}^{b^0} E  ( W_{it} ' \beta_{i,B}^0  - W_{it}' \beta_{i,A} )^2 
			+\frac{1}{NT}\sum_{i=1}^N \sum_{t=b+1}^{b^0} E  ( W_{it} ' \beta_{i,A}^0  - W_{it}' \beta_{i,A} )^2.	
		\end{align}
		The first term on the right-hand side of~\eqref{eq:diff-coef-pre-post-2} is less than $\epsilon$ due to $R(\beta ,b ) < \epsilon$.
		For the second term, we have
		\begin{align*}
			& \frac{1}{NT}\sum_{i=1}^N \sum_{t=b+1}^{b^0} E  ( W_{it} ' \beta_{i,A}^0  - W_{it}' \beta_{i,A})^2  \\
			\leq & 	\frac{C}{NT} \frac{b^0-b}{T-b^0}\sum_{i=1}^N \sum_{t=b^0+1}^{T} E  ( W_{it} ' \beta_{i,A}^0 - W_{it}' \beta_{i,A})^2  
			\leq  C   \frac{b^0-b}{T-b^0} \epsilon,
		\end{align*}
		where the first inequality follows from Assumption \ref{as-semi-stationary}, and the last inequality follows from 
		$1/(NT)\sum_{i=1}^N \sum_{t=b^0+1}^T E ( W_{it}' \beta_{i,A}^0  - W_{it}' \beta_{i,A})^2  < \epsilon$, which is further due to $R(\beta ,b ) < \epsilon$.
		Thus, we can bound~\eqref{eq:diff-coef-pre-post-2} from above as
		\begin{align}\label{eq:upper-bound-2}
			\frac{1}{NT}\sum_{i=1}^N \sum_{t=b+1}^{b^0} E  ( W_{it} ' \beta_{i,B}^0  - W_{it}' \beta_{i,A}^0)^2 < \left( 1+  C   \frac{b^0-b}{T-b^0}\right) \epsilon.
		\end{align}
		Using similar arguments as~\eqref{eq:lower-bound}, we can show that 
		\begin{align}\label{eq:lower-bound-2}
			 \frac{1}{NT}\sum_{i=1}^N \sum_{t=b+1}^{b^0} E  ( W_{it} ' \beta_{i,B}^0  - W_{it}' \beta_{i,A}^0 )^2  
			\geq \frac{b^0-b}{T}\underline{m}.
		\end{align}
		Combining~\eqref{eq:upper-bound-2} and~\eqref{eq:lower-bound-2}, we have 
		\begin{align*}
			\frac{b^0-b}{T} \underline{m} < \left( 1+  C   \frac{b^0-b}{T-b^0}\right) \epsilon,
		\end{align*}
		which further implies that
		\begin{align*}
			b^0 - b <  \epsilon \left( \frac{ \underline{m} }{T} -  \frac{C}{T-b^0} \epsilon\right)^{-1} .
		\end{align*}
		
		Thus, when $R(\beta ,b ) < \epsilon$ and $\epsilon = o(1)$, by Assumption \ref{as-break-date}, we have that, for some $C'$,
		\begin{align*}
			| b- b^0 | < C' \epsilon T 
		\end{align*}

	\end{proof}

	%%%%%%%%%%%%%%%%%%%%
	We then study the convergence rate of the cross-sectional averaged squared errors for the preliminary coefficient estimators obtained from Step 1.
	
	\begin{lemma}\label{th-beta-initial-rate}
		Suppose that Assumptions \ref{as-exogeneity}--\ref{as-min-eigenvalue} are satisfied. Then,
		\begin{align*}
			\frac{1}{N}\sum_{i=1}^N \vert \beta_i^0 - \hat \beta_i \vert_1 = O_p (s^{1/2} D_{NT}^{1/2}).
		\end{align*}	
		
	\end{lemma}
	
	\begin{proof}
		
		Let
		\begin{align*}
			\nu_{NT} ( \beta, b) = \frac{1}{NT} \sum_{i=1}^N \sum_{t=1}^T \left[ \left(y_{it} - W_{it} (b)' \beta_i \right)^2 - E \left(y_{it} - W_{it} (b)' \beta_i \right)^2 \right].
		\end{align*}
		Since $V_{NT} (\hat \beta, \hat b) \leq V_{NT} (\beta^0, \hat b)$, and using the result of \eqref{eq:risk-inequality} in the proof of Lemma \ref{theorem:risk-consistency}, we have
		\begin{align*}
			R(\hat \beta, \hat b)
			&\leq \left( \nu_{NT} ( \beta^0, \hat b) - \nu_{NT} ( \hat \beta, \hat b) \right) + D(\beta^0, \hat b) -  D(\hat \beta , \hat b ) + R(\beta^0, \hat b) .
		\end{align*}
		The proof of Lemma \ref{theorem:risk-consistency} shows that w.p.a.1,
		\begin{align*}
			\left| \nu_{NT} ( \beta^0, \hat b) - \nu_{NT} ( \hat \beta, \hat b) \right| \leq &   C\frac{\log (Np^2)}{\sqrt{T}}  \sum_{i=1}^N \sum_{k=1}^p \left| \hat \beta_{ik} - \beta_{ik}^0  \right| 
			=   C\frac{\log (Np^2)}{\sqrt{T}} \left| \hat \beta - \beta^0 \right|_1 \\
			\leq & \frac{1}{2} D_{NT} \left| \hat \beta - \beta^0 \right|_1,
		\end{align*}
		where the second inequality follows from Assumption \ref{as-penalty-order}.
		Observe that $\beta^0 = \beta_{J}^0$ and $\hat \beta = \hat \beta_J + \hat \beta_{J^c}$. By the triangle inequality, we have 
		$
			|\hat \beta_{J,i} |_1 \geq |\beta_{J,i}^0|_1 - |\hat \beta_{J,i} - \beta_{J,i}^0|_1,
		$
		where $\beta_{J,i}$ and $\hat \beta_{J,i}$ are the $i$-th element of $\beta_{J}$ and $\hat \beta_{J}$, respectively.
		We thus have 
		\begin{align*}
			D(\beta^0, \hat b) -  D(\hat \beta , \hat b ) =& D(\beta_J^0 , \hat b) - D(\hat \beta_{J} , \hat b) - D(\hat \beta_{J^c} , \hat b) \\
			\leq & D(\beta_J^0, \hat b) - D( \beta_{J}^0, \hat b )  + D( \hat \beta_{J} - \beta_{J}^0 , \hat b) - D(\hat \beta_{J^c} , \hat b) \\
			=&  D( \hat \beta_{J} - \beta_{J}^0, \hat b ) - D(\hat \beta_{J^c}, \hat b ).
		\end{align*}
	This further implies that
		\begin{align*}
			R(\hat \beta, \hat b) - \frac{1}{2} D_{NT} \left| \hat \beta - \beta^0 \right|_1
			&\leq   D( \hat \beta_{J} - \beta_{J}^0, \hat b) - D(\hat \beta_{J^c}, \hat b) + R(\beta^0, \hat b).
		\end{align*}
		Noting that $|\hat \beta - \beta^0|_1 = |\hat \beta_J - \beta_J^0|_1 + |\hat \beta_{J^c}|_1 $, 
		\begin{align*}
			R(\hat \beta, \hat b) - \frac{1}{2} D_{NT}|\hat \beta_J - \beta_J^0|_1 - \frac{1}{2} D_{NT} |\hat \beta_{J^c}|_1
			&\leq   D( \hat \beta_{J} - \beta_{J}^0, \hat b) - D(\hat \beta_{J^c}, \hat b) + R(\beta^0, \hat b).
		\end{align*}
		By the definition of $D_{NT}$, we have 
		\begin{align*}
			R(\hat \beta, \hat b) - \frac{1}{2} D (\hat \beta_J - \beta_J^0, \hat b) - \frac{1}{2} D (\hat \beta_{J^c}, \hat b)
			&\leq   D ( \hat \beta_{J} - \beta_{J}^0, \hat b) - D(\hat \beta_{J^c}, \hat b) + R(\beta^0, \hat b).	
		\end{align*}
		Adding $D(\hat  \beta_{J^c})$ to both sides and rearranging the terms, we have 
		\begin{align*}
			R(\hat \beta, \hat b) + \frac{1}{2} D (\hat \beta_{J^c}, \hat b)
			&\leq  \frac{3}{2} D ( \hat \beta_{J} - \beta_{J}^0, \hat b) + R(\beta^0, \hat b).		
		\end{align*}
		Because $R(\hat \beta, \hat b)\geq 0$, the following inequality holds
		\begin{align}
			D (\hat \beta_{J^c}, \hat b)
			\leq 3 D ( \hat \beta_{J} - \beta_{J}^0, \hat b) + 2 R(\beta^0, \hat b). \label{eq:basic-inequality}
		\end{align}
		
		To show the convergence rate in Lemma~\ref{th-beta-initial-rate}, we consider two cases. First, we consider that $R(\beta^0, \hat b) \leq D ( \hat \beta_{J} - \beta_{J}^0, \hat b) $, which implies that
		$			
		D (\hat \beta_{J^c}, \hat b)
		\leq 5 D ( \hat \beta_{J} - \beta_{J}^0, \hat b)
		$, and furthermore
		\begin{align}
			D_{\min} |\hat \beta_{J^c}|_1
			&\leq 5 D_{NT} | \hat \beta_{J} - \beta_{J}^0|_1.	\label{eq:condition-for-a3ii}
		\end{align}
		The inequality \eqref{eq:condition-for-a3ii} demonstrates that $\hat{\beta} - \beta_0$ meets the criteria for applying Assumption \ref{as-min-eigenvalue}.
		We now consider two subcases, namely $\hat b> b^0$ and $\hat b< b^0$, and then apply Assumption \ref{as-min-eigenvalue}. In the first subcase where $\hat b> b^0$,
		if $R (\hat \beta , \hat b ) < \epsilon$, then from~\eqref{eq:R:b>b0} we have
		\begin{align*}
			\frac{1}{NT}\sum_{i=1}^N \sum_{t=1}^{b^0} E ( W_{it}' \beta_{i,B}^0 - W_{it}'  \hat \beta_{i,B} )^2 < \epsilon
			\quad\textrm{and}\quad
			\frac{1}{NT}\sum_{i=1}^N \sum_{t=\hat b+1}^T E  ( W_{it}' \beta_{i,A}^0  - W_{it}' \hat \beta_{i,A})^2  < \epsilon,
		\end{align*}
		where the expectations are taken for $W_{it}$, and $\hat \beta$ is regarded non-random in the expectation.
		By Assumption~\ref{as-min-eigenvalue} whose implication can be used because of \eqref{eq:condition-for-a3ii}, we have 
		\begin{align*}
			\frac{1}{NT}\sum_{i=1}^N \sum_{t=1}^{b^0} E  ( W_{it}' \beta_{i,B}^0 - W_{it}'  \hat \beta_{i,B} )^2  \geq \rho \frac{b^0}{T} \frac{1}{N}\sum_{i=1}^N \left\vert \beta_{i,B}^0 - \hat \beta_{i,B} \right\vert_1^2,
		\end{align*}
		and 
		\begin{align*}
			\frac{1}{NT}\sum_{i=1}^N \sum_{t=\hat b+1}^T E  ( W_{it}' \beta_{i,A}^0 - W_{it}' \hat \beta_{i,A})^2
			\geq  \rho \frac{T- \hat b}{T} \frac{1}{N}\sum_{i=1}^N \left\vert \beta_{i,A}^0 - \hat \beta_{i,A} \right\vert_1^2.
		\end{align*}
		Thus, we can obtain that 
		\begin{align*}
			\frac{1}{N}\sum_{i=1}^N \left\vert \beta_{i,B}^0 - \hat \beta_{i,B} \right\vert_1^2 < \epsilon \frac{T}{\rho b^0}
			\quad\textrm{and}\quad
			\frac{1}{N}\sum_{i=1}^N \left\vert \beta_{i,A}^0 - \hat \beta_{i,A} \right\vert_1^2 < \epsilon \frac{T}{\rho (T-\hat b)}.
		\end{align*}
		In the second subcase where $\hat b < b^0$, we can use similar arguments and obtain that
		\begin{align*}
			\frac{1}{N}\sum_{i=1}^N \left\vert \beta_{i,B}^0 - \hat \beta_{i,B} \right\vert_1^2 < \epsilon \frac{T}{\rho \hat b}
			\quad\textrm{and}\quad
			\frac{1}{N}\sum_{i=1}^N \left\vert \beta_{i,A}^0 - \hat \beta_{i,A} \right\vert_1^2 < \epsilon \frac{T}{\rho (T- b^0)}.
		\end{align*}
		We now take $\epsilon = C sD_{NT}$ for some constant $C$. By Lemma \ref{theorem:risk-consistency}, $R(\hat \beta, \hat b) < \epsilon$ with probability arbitrarily close to one if $C$ is set sufficiently large. Noting that $| \hat b - b^0| = O_p (\epsilon T)$ from Lemma~\ref{th-b-initial-rate}, and $\epsilon = o(1)$, we can obtain that 
		\begin{align*}
			\frac{1}{N}\sum_{i=1}^N \left\vert \beta_{i,B}^0 - \hat \beta_{i,B} \right\vert_1^2 = O_p (\epsilon)
			\quad\textrm{and}\quad
			\frac{1}{N}\sum_{i=1}^N \left\vert \beta_{i,A}^0 - \hat \beta_{i,A} \right\vert_1^2 = O_p (\epsilon).
		\end{align*}
		Thus, we have 
		$
			1/N\sum_{i=1}^N \left\vert \beta_i^0 - \hat \beta_i \right\vert_1^2 = O_p (\epsilon).
		$
		By the Minkowski inequality, we have 
		\begin{align*}
			\frac{1}{N}\sum_{i=1}^N \left\vert \beta_i^0 - \hat \beta_i \right\vert_1 
			\leq \left( \frac{1}{N}\sum_{i=1}^N \left\vert \beta_i^0 - \hat \beta_i \right\vert_1^2 \right)^{1/2} 
			= O_p (\epsilon^{1/2}).
		\end{align*}
		
		Next, we consider the case where $R(\beta^0, \hat b) > D ( \hat \beta_{J} - \beta_{J}^0, \hat b) $.
		Noting that $D (\hat \beta_{J^c}, \hat b) + D ( \hat \beta_{J} - \beta_{J}^0, \hat b) = D(\hat \beta - \beta^0, \hat b)$ and from \eqref{eq:basic-inequality}, we have 
		\begin{align*}
			D(\hat \beta - \beta^0, \hat b) \leq 4 D (\hat \beta_J - \beta_J^0, \hat b) + 2R(\beta^0, \hat b) \leq 6 R(\beta^0, \hat b).
		\end{align*}
		Since the bound of $D(\hat \beta - \beta^0, \hat b)$ depends on $R(\beta^0, \hat b)$, we examine $R(\beta_0, \hat b)$. Again, we consider two subcases, depending on the relationship between $b$ and $b^0$. First, for $b>b^0$, we note that
		\begin{align*}
			R(\beta^0, b) =& \frac{1}{NT} \sum_{i=1}^N \sum_{t=1}^T E \left[W_{it} (b^0) '\beta_i^0 - W_{it}(b)'\beta_i^0\right]^2 
			= \frac{1}{NT} \sum_{i=1}^N \sum_{t=b^0+1}^b E \left[W_{it} (b^0) '\beta_i^0 - W_{it}(b)'\beta_i^0\right]^2  \\
			= &\frac{1}{NT} \sum_{i=1}^N \sum_{t=b^0+1}^b E [(X_{it} ' (\gamma_{i,A}^0 - \gamma_{i,B}^0 ))^2] 
			\leq  C \frac{b-b^0}{T}. 
		\end{align*}
		Using similar arguments, we can show that in the second subcase where $b<b^0$, 
		\begin{align*}
			R(\beta^0, b) \leq C\frac{b^0 -b}{T}. 
		\end{align*}
		By Lemma \ref{th-b-initial-rate}, we have
		$
			R(\beta^0, \hat b) =O_p( \epsilon ) .
		$
		We thus have 
		$
			D(\hat \beta - \beta^0, \hat b) = O_p ( \epsilon ).
		$
		Note that 
		$
			D(\hat \beta - \beta^0, \hat b) \geq D_{\min} |\hat \beta - \beta^0|_1. % \geq  D_{\min} \Vert \hat \beta - \beta^0\Vert_2^2.
		$
		Thus, by Assumption \ref{as-penalty-order}, we have 
		$
			1/N\sum_{i=1}^N \left\vert \beta_i^0 - \hat \beta_i \right\vert_1 = O_p ( \epsilon ). %\frac{1}{N} \left\Vert \beta^0 - \hat \beta \right\Vert_2^2 =  O_p (\epsilon).
		$
		This leads to the result that
		\begin{align*}
			\frac{1}{N}\sum_{i=1}^N \left\vert \beta_i^0 - \hat \beta_i \right\vert_1 = O_p (\epsilon^{1/2} ),
		\end{align*}	
		because $\epsilon \to 0$ implies that $\epsilon < \epsilon^{1/2}$ for $N$ and $T$ sufficiently large.

	\end{proof}

	\section{Proof of Theorem \ref{th-b-superconsistency}}

	Let $\mathcal{B} (\beta^0, \epsilon)$ be a set of $\beta$ which satisfies
	$
	1/N \sum_{i=1}^N \left\vert \beta_i - \beta_i^0 \right\vert_1 < \epsilon^{1/2}. %, \text{ and } 	\frac{1}{N} \sum_{i=1}^N | \beta_i - \beta^0 | < \epsilon
	$
	Note that, for $\epsilon = C s D_{NT}$ and sufficiently large $C$, we have $\hat \beta \in \mathcal{B} (\beta^0,\epsilon)$ with probability arbitrarily close to one asymptotically.
	To simplify the notation and make the rate of convergence explicit, we omit the first argument of $\mathcal{B} (\beta^0, \epsilon)$ and write it as $\mathcal{B} (  \epsilon )$ in the following. 
	Note that $\tilde b \neq b^0$ is equivalent to
	\begin{align*}
		&	\frac{1}{NT}\sum_{i=1}^N \sum_{t=1}^T \left(y_{it} - W_{it} (b^0)' \hat \beta_i \right)^2 \\
		>&  \min_{b \in \{1,\dots, T\} \backslash b^0 } 	\frac{1}{NT}\sum_{i=1}^N \sum_{t=1}^T \left(y_{it} - W_{it} (b)' \hat \beta_i \right)^2 \\
		=&  \min \left( \min_{b < b^0} \frac{1}{NT}\sum_{i=1}^N \sum_{t=1}^T (y_{it} - W_{it} (b)' \hat \beta_i )^2, \min_{b >b^0 } \frac{1}{NT}\sum_{i=1}^N \sum_{t=1}^T (y_{it} - W_{it} (b)' \hat \beta_i )^2\right).
	\end{align*}
	
	We first consider the minimum over $b<b^0$. We observe 
	\begin{align*}
		& \Pr \left( 	\frac{1}{NT}\sum_{i=1}^N \sum_{t=1}^T (y_{it} - W_{it} (b^0)' \hat \beta_i )^2 > \min_{b < b^0} \frac{1}{NT}\sum_{i=1}^N \sum_{t=1}^T (y_{it} - W_{it} (b)' \hat \beta_i )^2 \right) \\
		\leq & 	\Pr \biggl(	\frac{1}{NT}\sum_{i=1}^N \sum_{t=1}^T (y_{it} - W_{it} (b^0)' \beta_i )^2  > \min_{b < b^0} \frac{1}{NT}\sum_{i=1}^N \sum_{t=1}^T (y_{it} - W_{it} (b)'  \beta_i )^2  \text{ for some }\beta \in \mathcal{B}(\epsilon)\biggr) + o(1)\\
		=& \Pr \biggl( \sup_{\beta \in \mathcal{B}(\epsilon) } \max_{b < b^0}	\frac{1}{NT}\sum_{i=1}^N \sum_{t=1}^T \left[(y_{it} - W_{it} (b^0)' \beta_i )^2 - (y_{it} - W_{it} (b)'  \beta_i )^2 \right] >0 \biggr) + o(1).
	\end{align*} 
	Our goal is to show that the probability on the right-hand side of the equation above converges to zero. Thus, we examine the term inside the last probability. For $b < b^0$, we can derive that
	\begin{align}\label{eq:diff-ssr}
		&	\frac{1}{NT}\sum_{i=1}^N \sum_{t=1}^T \left[ \left(y_{it} - W_{it} (b^0)' \beta_i \right)^2 - \left(y_{it} - W_{it} (b)'  \beta_i \right)^2 \right]\nonumber\\
		= & 	\frac{1}{NT}\sum_{i=1}^N \sum_{t=b+1}^{b^0} \left[ \left(y_{it} - W_{it}' \beta_{i,B} \right)^2 - \left(y_{it} - W_{it}' \beta_{i,A} \right)^2 \right]\nonumber\\
		=&  \frac{1}{NT}\sum_{i=1}^N \sum_{t=b+1}^{b^0} \left[ \left(W_{it} ' (\beta_{i,B}^0 - \beta_{i,B} ) + u_{it} \right)^2 - \left(W_{it}'\beta_{i,B}^0 - W_{it} ' \beta_{i,A} + u_{it} \right)^2 \right] \nonumber\\
		=& 2	 \frac{1}{NT}\sum_{i=1}^N \sum_{t=b+1}^{b^0}  W_{it} '(\beta_{i,A} - \beta_{i,B})  W_{it} ' (\beta_{i,B}^0 - \beta_{i,B} ) - \frac{1}{NT}\sum_{i=1}^N \sum_{t=b+1}^{b^0}   [W_{it} ' (\beta_{i,A} - \beta_{i,B})]^2  \nonumber\\
		& + 2 \frac{1}{NT}\sum_{i=1}^N  \sum_{t=b+1}^{b^0} u_{it} W_{it}'(\beta_{i,A} - \beta_{i,B}).
	\end{align}
	We examine each of the three terms on the right-hand side of~\eqref{eq:diff-ssr} in turn.
	
	For the first term, by the Cauchy-Schwarz inequality, we have
	\begin{align}\label{eq:Wbeta-pre-post-Wbeta0}
		& \left| \frac{1}{NT}\sum_{i=1}^N \sum_{t=b+1}^{b^0}  W_{it} ' (\beta_{i,A} - \beta_{i,B})  W_{it} ' (\beta_{i,B}^0 - \beta_{i,B} )\right|  \nonumber\\
		\leq  & \left( \frac{1}{NT}\sum_{i=1}^N \sum_{t=b+1}^{b^0} [ W_{it} ' (\beta_{i,A} - \beta_{i,B})]^2\right)^{1/2} 
		\left( \frac{1}{NT}\sum_{i=1}^N \sum_{t=b+1}^{b^0} [ W_{it} ' (\beta_{i,B}^0 - \beta_{i,B} )]^2\right)^{1/2}.
	\end{align}
	We observe that 
	\begin{align}\label{eq:Wbeta}
		& \frac{1}{NT}\sum_{i=1}^N \sum_{t=b+1}^{b^0} \left[ W_{it} ' (\beta_{i,A} - \beta_{i,B})\right]^2  \nonumber\\
		=& \frac{1}{NT}\sum_{i=1}^N \sum_{t=b+1}^{b^0} \left[ X_{it} ' (\gamma_{i,A} - \gamma_{i,B})-X_{it} ' (\gamma_{i,A}^0 - \gamma_{i,B}^0)\right]^2  \nonumber\\
		& + 2 \frac{1}{NT}\sum_{i=1}^N \sum_{t=b+1}^{b^0} \left[ X_{it} ' (\gamma_{i,A}^0 - \gamma_{i,B}^0)\right] \left[ X_{it} ' (\gamma_{i,A} - \gamma_{i,B})- X_{it} '(\gamma_{i,A}^0 - \gamma_{i,B}^0)\right] \nonumber\\
		& + \frac{1}{NT}\sum_{i=1}^N \sum_{t=b+1}^{b^0} \left[ X_{it} ' (\gamma_{i,A}^0 - \gamma_{i,B}^0)\right]^2.
	\end{align}
	For the first term in~\eqref{eq:Wbeta}, we have
	\begin{align*}
		& \left(\frac{1}{NT}\sum_{i=1}^N \sum_{t=b+1}^{b^0} \left[ X_{it} ' (\gamma_{i,A} - \gamma_{i,B})-X_{it}' (\gamma_{i,A}^0 - \gamma_{i,B}^0)\right]^2  \right)^{1/2} \\
		\leq &\left(	\frac{1}{NT}\sum_{i=1}^N \sum_{t=b+1}^{b^0} \left[ X_{it} ' (\gamma_{i,A} - \gamma_{i,A}^0)\right]^2  \right)^{1/2} 
		+\left(	\frac{1}{NT}\sum_{i=1}^N \sum_{t=b+1}^{b^0} \left[ X_{it} ' ( \gamma_{i,B}-  \gamma_{i,B}^0)\right]^2  \right)^{1/2} 
	\end{align*}
	For each term on the right-hand side of the inequality above, by Lemma \ref{lem:fuk-nagaev-dependent-ij-b} and $\beta \in \mathcal{B} (\epsilon)$ from Theorem~\ref{th-beta-initial-rate}, we have, w.p.a.1,
	\begin{align*}
		&\frac{1}{NT}\sum_{i=1}^N \sum_{t=b+1}^{b^0} \left[  X_{it} ' (\gamma_{i,A}^0 - \gamma_{i,A} )\right]^2  \\
		=& \frac{1}{NT}\sum_{i=1}^N \sum_{t=b+1}^{b^0} \left[  \sum_{j=1}^N x_{jt} ' (\gamma_{ij,A}^0 - \gamma_{ij,A} )\right]^2 
		\leq  \frac{1}{N}\sum_{i=1}^N  |\gamma_{i,A}^0 - \gamma_{i,A} |_1^2 \frac{1}{T} \sum_{t=b+1}^{b^0} \max_{j} x_{jt} ^2
		\leq   C \epsilon \frac{b^0-b}{T} N^c,
	\end{align*}
	and similarly,
	\begin{align*}
		\frac{1}{NT}\sum_{i=1}^N \sum_{t=b+1}^{b^0} \left[  X_{it} ' (\gamma_{i,B}^0 - \gamma_{i,B} )\right]^2 
		\leq C \epsilon \frac{b^0-b}{T} N^c.
	\end{align*}
	We bound $E( \max_{j} x_{jt}^2 )$ by using the fact that $E(x_{jt}^2)$ is bounded uniformly over $j$ and $t$ by Lemma \ref{lem:moment} and use the fact that $E(\max_{j} x_{jt}^2) \leq N^{c} \max (E(x_{jt}^{2/c}))^c$ for any $c\leq 1$.
	Thus, by Lemma~\ref{th-b-initial-rate} it holds that, w.p.a.1,
	\begin{align*}
		\frac{1}{NT}\sum_{i=1}^N \sum_{t=b+1}^{b^0} \left[ X_{it} ' (\gamma_{i,A} - \gamma_{i,B})- X_{it} '(\gamma_{i,A}^0 - \gamma_{i,B}^0)\right]^2 \leq C \epsilon \frac{b^0-b}{T} \log N.
	\end{align*}
	Next, we consider the second term on the right-hand side of~\eqref{eq:Wbeta}, and rewrite it as
	\begin{align*}
		& \left|\frac{1}{NT}\sum_{i=1}^N \sum_{t=b+1}^{b^0}  (\gamma_{i,A}^0 - \gamma_{i,B}^0)' X_{it} X_{it} ' (\gamma_{i,A} - \gamma_{i,B}- (\gamma_{i,A}^0 - \gamma_{i,B}^0))\right| \\
		\leq & \frac{1}{N}\sum_{i=1}^N | \gamma_{i,A} - \gamma_{i,B}- (\gamma_{i,A}^0 - \gamma_{i,B}^0)|_1 \max_i \max_j \left| \frac{1}{T}\sum_{t=b+1}^{b^0}  (\gamma_{i,A}^0 - \gamma_{i,B}^0)' X_{it} x_{jt} \right|\\ 
		\leq & C  \epsilon^{1/2} \frac{b^0 - b}{T} \log N,
	\end{align*}
	where the second inequality follows from $\beta \in \mathcal{B} (\epsilon)$, Assumption \ref{as-exogeneity}, and Lemma \ref{lem:fuk-nagaev-dependent-ij-b}.

	Finally, for the third term on the right-hand side of~\eqref{eq:Wbeta}, we have
	\begin{align*}
		&\frac{1}{NT}\sum_{i=1}^N \sum_{t=b+1}^{b^0} \left[ X_{it} ' (\gamma_{i,A}^0 - \gamma_{i,B}^0)\right]^2 \\
		=& 	\frac{1}{NT}\sum_{i=1}^N \sum_{t=b+1}^{b^0} E \left[ X_{it} ' (\gamma_{i,A}^0 - \gamma_{i,B}^0)\right]^2 
		 +\frac{1}{NT}\sum_{i=1}^N \sum_{t=b+1}^{b^0} \Big(\left[  X_{it} ' (\gamma_{i,A}^0 - \gamma_{i,B}^0)\right]^2 - E\left[ X_{it} ' (\gamma_{i,A}^0 - \gamma_{i,B}^0)\right]^2\Big).
	\end{align*} 
	By the Chebyshev inequality under Assumptions \ref{as-tail} and \ref{as-mixing}, noting that $\log N \to \infty$, it holds that 
	\begin{align*}
		\left| \frac{1}{NT}\sum_{i=1}^N \sum_{t=b+1}^{b^0} \Big( \left[ X_{it} ' (\gamma_{i,A}^0 - \gamma_{i,B}^0)\right]^2 - E (\left[ X_{it} ' (\gamma_{i,A}^0 - \gamma_{i,B}^0)\right]^2)\Big)	\right| \leq C \frac{\sqrt{b^0-b}\log N}{\sqrt{N} T}. 
	\end{align*}
	Thus, by Lemma \ref{lem:moment} and Assumption \ref{as-penalty-order}, we have 
	\begin{align*}
		\frac{1}{NT}\sum_{i=1}^N \sum_{t=b+1}^{b^0} \left[ X_{it} ' (\gamma_{i,A}^0 - \gamma_{i,B}^0)\right]^2 
		\leq C \left(\frac{\sqrt{b^0-b}\log N}{\sqrt{N} T} + \frac{b^0-b}{T}\right) . 
	\end{align*}
	
	We have evaluated all terms in \eqref{eq:Wbeta}, and obtain:
	\begin{align*}
		&\frac{1}{NT}\sum_{i=1}^N \sum_{t=b+1}^{b^0} \left[ W_{it} ' (\beta_{i,A} - \beta_{i,B})\right]^2 \\
		\leq & C \left( \frac{\sqrt{b^0-b}\log N}{\sqrt{N} T} + \frac{b^0-b}{T} + \epsilon^{1/2} \frac{b^0-b}{T}  \log N \right) \\ 
		\leq & C \frac{b^0-b}{T}, \quad \textrm{w.p.a.1},
	\end{align*}
	for $N,T$ sufficiently large with $\log N /\sqrt{T} \to 0$ and $\epsilon \to 0$.
	Moreover, a similar argument to that for $ \sum_{i=1}^N \sum_{t=b+1}^{b^0} [ X_{it} ' (\gamma_{i,B}^0 - \gamma_{i,B} )]^2/ (NT)$ shows that the second term in \eqref{eq:Wbeta-pre-post-Wbeta0} is bounded as 
	$
		1/(NT)\sum_{i=1}^N \sum_{t=b+1}^{b^0} [ W_{it} ' (\beta_{i,B}^0 - \beta_{i,B} )]^2 \leq C \epsilon(b^0-b)/T.
	$
	
	The above paragraph provides the bounds for the two terms in \eqref{eq:Wbeta-pre-post-Wbeta0} and implies that:
	\begin{align*}
		\left| \frac{1}{NT}\sum_{i=1}^N \sum_{t=b+1}^{b^0}  W_{it} ' (\beta_{i,A} - \beta_{i,B})  W_{it} ' (\beta_{i,B}^0 - \beta_{i,B} )\right| \leq C \epsilon \frac{b^0 -b}{T}.
	\end{align*}

	We now consider the second term in \eqref{eq:diff-ssr}. 
	It follows that
	\begin{align*}
		& \frac{1}{NT}\sum_{i=1}^N \sum_{t=b+1}^{b^0}   [W_{it} ' (\beta_{i,A} - \beta_{i,B})]^2 \\
		= &	\frac{1}{NT}\sum_{i=1}^N \sum_{t=b+1}^{b^0}   [X_{it} ' (\gamma_{i,A} - \gamma_{i,B})]^2 \\
		= &	\frac{1}{NT}\sum_{i=1}^N \sum_{t=b+1}^{b^0}   [X_{it} ' (\gamma_{i,A}^0 - \gamma_{i,B}^0)]^2 
		+\frac{1}{NT}\sum_{i=1}^N \sum_{t=b+1}^{b^0}  \left( [X_{it} ' (\gamma_{i,A} - \gamma_{i,B})]^2 -  [X_{it} ' (\gamma_{i,A}^0 - \gamma_{i,B}^0)]^2 \right) \\
		= &	\frac{1}{NT}\sum_{i=1}^N \sum_{t=b+1}^{b^0}  E [X_{it} ' (\gamma_{i,A}^0 - \gamma_{i,B}^0)]^2
		+ \frac{1}{NT}\sum_{i=1}^N \sum_{t=b+1}^{b^0}  \left(  [X_{it} ' (\gamma_{i,A}^0 - \gamma_{i,B}^0)]^2 - E [X_{it} ' (\gamma_{i,A}^0 - \gamma_{i,B}^0)]^2 \right) \\
		&+\frac{1}{NT}\sum_{i=1}^N \sum_{t=b+1}^{b^0}  \left( [X_{it} ' (\gamma_{i,A} - \gamma_{i,B})]^2 -  [X_{it} ' (\gamma_{i,A}^0 - \gamma_{i,B}^0)]^2 \right) .
	\end{align*}
	Assumption \ref{as-break-identification} implies that 
	\begin{align*}
		\frac{1}{NT}\sum_{i=1}^N \sum_{t=b+1}^{b^0} E \left[( X_{it} ' (\gamma_{i,A}^0 - \gamma_{i,B}^0)\right]^2 \geq \underline{m} \frac{b^0-b}{T}.
	\end{align*}
	By the Chebyshev inequality under Assumptions \ref{as-tail} and \ref{as-mixing}, we have 
	\begin{align*}
		\frac{1}{NT}\sum_{i=1}^N \sum_{t=b+1}^{b^0}  \left(  [X_{it} ' (\gamma_{i,A}^0 - \gamma_{i,B}^0)]^2 - E [X_{it} ' (\gamma_{i,A}^0 - \gamma_{i,B}^0)]^2 \right)/(NT) =  O_p \left(\frac{\sqrt{b^0-b}}{\sqrt{N}T}\right).
	\end{align*}
	Thus, because $\log N \to \infty$, it holds that 
	\begin{align*}
		\left|\frac{1}{NT}\sum_{i=1}^N \sum_{t=b+1}^{b^0}  \left(  [X_{it} ' (\gamma_{i,A}^0 - \gamma_{i,B}^0)]^2 - E [X_{it} ' (\gamma_{i,A}^0 - \gamma_{i,B}^0)]^2 \right) \right| \leq C \frac{\sqrt{b^0-b}\log N}{\sqrt{N} T}, \quad \text{w.p.a.1}.
	\end{align*}
	Lastly, we have 
	\begin{align*}
		& \left|\frac{1}{NT}\sum_{i=1}^N \sum_{t=b+1}^{b^0}  \left( [X_{it} ' (\gamma_{i,A} - \gamma_{i,B})]^2 -  [X_{it} ' (\gamma_{i,A}^0 - \gamma_{i,B}^0)]^2 \right)\right| \\
		=& \left|\frac{1}{NT}\sum_{i=1}^N \sum_{t=b+1}^{b^0}  (\gamma_{i,A} - \gamma_{i,B}+ (\gamma_{i,A}^0 - \gamma_{i,B}^0))' X_{it} X_{it} ' (\gamma_{i,A} - \gamma_{i,B}- (\gamma_{i,A}^0 - \gamma_{i,B}^0))\right| \\
		\leq &\left|\frac{1}{NT}\sum_{i=1}^N \sum_{t=b+1}^{b^0}  (\gamma_{i,A} - \gamma_{i,B}- (\gamma_{i,A}^0 - \gamma_{i,B}^0))' X_{it} X_{it} ' (\gamma_{i,A} - \gamma_{i,B}- (\gamma_{i,A}^0 - \gamma_{i,B}^0))\right| \\
		& + 2 \left|\frac{1}{NT}\sum_{i=1}^N \sum_{t=b+1}^{b^0}  (\gamma_{i,A}^0 - \gamma_{i,B}^0)' X_{it} X_{it} ' (\gamma_{i,A} - \gamma_{i,B}- (\gamma_{i,A}^0 - \gamma_{i,B}^0))\right| \\
		\leq & \frac{1}{N}\sum_{i=1}^N | \gamma_{i,A} - \gamma_{i,B}- (\gamma_{i,A}^0 - \gamma_{i,B}^0) |_1^2 \frac{1}{T} \sum_{t=b+1}^{b^0} \max_j x_{jt}^2 \\
		+ & \frac{1}{N}\sum_{i=1}^N | \gamma_{i,A} - \gamma_{i,B}- (\gamma_{i,A}^0 - \gamma_{i,B}^0) |_1  \max_{i} \max_j \left| \frac{1}{T} \sum_{t=b+1}^{b^0} (\gamma_{i,A}^0 - \gamma_{i,B}^0)' X_{it} x_{jt} \right| \\
		\leq & C \frac{b^0 - b}{T} (\epsilon N^c + \epsilon^{1/2}\log N),
	\end{align*}
	where the third inequality follows from $\beta \in \mathcal{B} (\epsilon)$, Assumptions \ref{as-exogeneity}, and Lemma \ref{lem:fuk-nagaev-dependent-ij-b}.
	Note that we bound $E( \max_{j} x_{jt}^2 )$ by using the fact that $E(x_{jt}^2)$ is bounded uniformly over $j$ and $t$ by Lemma \ref{lem:moment} and that $E(\max_{j} x_{jt}^2) \leq N^{c} \max_j (E(x_{jt}^{2/c}))^c$ for any $c\leq 1$.
	We thus have 
	\begin{align*}
		\frac{1}{NT}\sum_{i=1}^N \sum_{t=b+1}^{b^0}   [W_{it} ' (\beta_{i,A} - \beta_{i,B})]^2 
		\geq  \underline{m} \frac{b^0-b}{T} - C \left( (\epsilon N^c+ \epsilon^{1/2} \log N) \frac{b^0 - b}{T} +  \frac{\sqrt{b^0-b}\log N}{\sqrt{N} T} \right).
	\end{align*}
	
	The results for the first and the second term in \eqref{eq:diff-ssr} together with $\epsilon = o(1)$ imply that, w.p.a.1, we have
	\begin{align*}
		& 2	 \frac{1}{NT}\sum_{i=1}^N \sum_{t=b+1}^{b^0}  W_{it} ' (\beta_{i,A} - \beta_{i,B})  W_{it} ' (\beta_{i,B}^0 - \beta_{i,B} ) - \frac{1}{NT}\sum_{i=1}^N \sum_{t=b+1}^{b^0}   ( W_{it} ' (\beta_{i,A} - \beta_{i,B}))^2
		\\ < &  - \underline{m} \frac{b^0 - b}{3T}.
	\end{align*}
	
	Lastly, we consider the third term in \eqref{eq:diff-ssr}.
	Note that 
	\begin{align*}
		&2 \frac{1}{N} \frac{1}{T}\sum_{i=1}^N  \sum_{t=b+1}^{b^0} u_{it} W_{it}' (\beta_{i,A} - \beta_{i,B}  ) \\
		= &  2 \frac{1}{N} \frac{1}{T}\sum_{i=1}^N  \sum_{t=b+1}^{b^0} u_{it} X_{it}' (\gamma_{i,A} - \gamma_{i,B}  ) \\
		=& 	2 \frac{1}{N} \frac{1}{T}\sum_{i=1}^N  \sum_{t=b+1}^{b^0} u_{it} X_{it}' (\gamma_{i,A}^0 - \gamma_{i,B}^0 ) 
		 + 	2 \frac{1}{N} \frac{1}{T}\sum_{i=1}^N  \sum_{t=b+1}^{b^0} u_{it} X_{it}'\left[\gamma_{i,A} - \gamma_{i,B} - (\gamma_{i,A}^0 - \gamma_{i,B}^0)  \right].
	\end{align*}
	We further observe that
	\begin{align*}
		& \left| \frac{1}{NT} \sum_{i=1}^N  \sum_{t=b+1}^{b^0} u_{it} X_{it}'\left[\gamma_{i,A} - \gamma_{i,B} - (\gamma_{i,A}^0 - \gamma_{i,B}^0)\right] \right| \\
		= & \left| \frac{1}{NT} \sum_{i=1}^N \sum_{j=1}^{N} \sum_{t=b+1}^{b^0} u_{it} x_{jt} \left[\gamma_{ij,A} - \gamma_{ij,B} - (\gamma_{ij,A}^0 - \gamma_{ij,B}^0)\right] \right| \\
		\leq &\frac{1}{N} \sum_{i=1}^N \sum_{j=1}^N | \gamma_{ij,A} - \gamma_{ij,B} - (\gamma_{ij,A}^0 - \gamma_{ij,B}^0) | 
		\times \max_{i} \max_j \left|  \frac{1}{T} \sum_{t=b+1}^{b^0} u_{it} x_{jt} \right| \\
		\leq &  C \epsilon^{1/2} \frac{\log N}{\sqrt{T}},\quad \textrm{w.p.a.1},
	\end{align*} 
	where the second inequality is due to $\beta \in \mathcal{B}(\epsilon)$ and Lemma \ref{lem:fuk-nagaev-dependent-ij-b}.
	This result, together with Assumption \ref{as-penalty-order}, implies that 
	\begin{align*}
		\Pr \bigg( \sup_{\beta \in \mathcal{B}(\epsilon) } \max_{b < b^0} \left(	2 \frac{1}{N} \frac{1}{b^0-b}\sum_{i=1}^N  \sum_{t=b+1}^{b^0} u_{it} X_{it}'\left[\gamma_{i,A} - \gamma_{i,B} - (\gamma_{i,A}^0 - \gamma_{i,B}^0 )\right]    \right)    > \frac{\underline{m}}{6} \bigg) \to 0.
	\end{align*}

	Thus, we have 
	\begin{align*}
		& \Pr \left( \sup_{\beta \in \mathcal{B}(\epsilon) } \max_{b < b^0}	\frac{1}{NT}\sum_{i=1}^N \sum_{t=1}^T \left[ \left(y_{it} - W_{it} (b^0)' \beta_i \right)^2 - \left(y_{it} - W_{it} (b)'  \beta_i \right)^2 \right] >0 \right) \\
		\leq & \Pr \left( \sup_{\beta \in \mathcal{B}(\epsilon) } \max_{b < b^0} \left(	2 \frac{1}{NT}\sum_{i=1}^N  \sum_{t=b+1}^{b^0} u_{it} W_{it}'(\beta_{i,A} - \beta_{i,B} ) - \underline{m} \frac{b^0 - b}{3T}  \right)  >0 \right) \\
		= &  \Pr \left( \sup_{\beta \in \mathcal{B}(\epsilon) } \max_{b < b^0} \left(	2 \frac{1}{N} \frac{1}{b^0-b}\sum_{i=1}^N  \sum_{t=b+1}^{b^0} u_{it} X_{it}' (\gamma_{i,A} - \gamma_{i,B})  - \frac{\underline{m}}{3}  \right)  >0 \right) \\
		= &  \Pr \left( \sup_{\beta \in \mathcal{B}(\epsilon) } \max_{b < b^0} \left(	2 \frac{1}{N} \frac{1}{b^0-b}\sum_{i=1}^N  \sum_{t=b+1}^{b^0} u_{it} X_{it}' (\gamma_{i,A} - \gamma_{i,B} )   \right)  > \frac{\underline{m}}{3} \right) \\
		\leq & \Pr \left( \sup_{\beta \in \mathcal{B}(\epsilon) } \max_{b < b^0} \left(	2 \frac{1}{N} \frac{1}{b^0-b}\sum_{i=1}^N  \sum_{t=b+1}^{b^0} u_{it} X_{it}'(\gamma_{i,A}^0 - \gamma_{i,B}^0)    \right)  > \frac{\underline{m}}{6} \right)\\
		& +  \Pr \bigg( \sup_{\beta \in \mathcal{B}(\epsilon) } \max_{b < b^0} \left(	2 \frac{1}{N} \frac{1}{b^0-b}\sum_{i=1}^N  \sum_{t=b+1}^{b^0} u_{it} X_{it}'\left[\gamma_{i,A} - \gamma_{i,B} - (\gamma_{i,A}^0 - \gamma_{i,B}^0 )\right]    \right)    > \frac{\underline{m}}{6} \bigg) \\
		\to & 0,
	\end{align*}
	where the last equality follows by applying \citet[][Lemma A.6]{bai&perron98} which is an extension of \citet{hajek1955generalization}. Here we use the observation that an $L_r$-bounded mixing sequence is an $L_p$ mixingale sequence for $1 \leq p < r$ as discussed in (Davidson, 1994, page 248). Thus, under Assumptions \ref{as-tail} and \ref{as-mixing}, $u_{it} X_{it}'(\gamma_{i,A}^0 - \gamma_{i,B}^0)$ is an $L_2$ mixingale, and we can apply \citet[][Lemma A.6]{bai&perron98}.
	
	Hence, for $b<b^0$, we have, 
	\begin{align*}
		\Pr \left( 	\frac{1}{NT}\sum_{i=1}^N \sum_{t=1}^T \left[y_{it} - W_{it} (b^0)' \hat \beta_i \right]^2 > \min_{b < b^0} \frac{1}{NT}\sum_{i=1}^N \sum_{t=1}^T \left[y_{it} - W_{it} (b)' \hat \beta_i \right]^2 \right)= o(1).
	\end{align*}
	With a similar argument, we can show that, for $b>b^0$,
	\begin{align*}
		\Pr \left( 	\frac{1}{NT}\sum_{i=1}^N \sum_{t=1}^T (y_{it} - W_{it} (b^0)' \hat \beta_i )^2 > \min_{b > b^0} \frac{1}{NT}\sum_{i=1}^N \sum_{t=1}^T (y_{it} - W_{it} (b)' \hat \beta_i )^2 \right)= o(1).
	\end{align*}
	This leads to our conclusion that
	\begin{align*}
		\Pr ( \tilde b \neq b^0 ) \to 0.
	\end{align*}
	
	\qed
	
	\section{Proof of Theorem \ref{thm-dml}}
	
	In the proof, for simplicity, we adopt an abused notation to let $X_{it, \hat{\mathbf{s}}_i}$ include an intercept, such that $\gamma_{i,\hat{\mathbf{s}}_i}$ and $\nu_{i,\hat{\mathbf{s}}_i}$ incorporate the coefficients of the intercept, namely $\alpha_i$ and $\eta_i$, respectively. Furthermore, we assume that $z_{it}$ is one dimensional. Under this revised notation, the private-effect estimator, e.g., using the main sample, can be written as  
	\begin{align*}
		\tilde \delta^m = 
		\left( \sum_{i=1}^N \sum_{t\in \mathcal{T}_m} \left(z_{it} - X_{it, \hat{\mathbf{s}}_i} (\tilde b)'\tilde{\nu}_i^a \right)^2 \right)^{-1} \times   \sum_{i=1}^N \sum_{t\in \mathcal{T}_m} \left(z_{it}  - X_{it, \hat{\mathbf{s}}_i} (\tilde b)'\tilde{\nu}_i^a \right)  \left( y_{it} - X_{it, \hat{\mathbf{s}}_i} (\tilde b)' \tilde{\gamma}_i^a  \right).
	\end{align*}
	
	Because $\Pr( \tilde b = b^0) \to 1$ and $\Pr ( \hat{\mathbf{s}}_i =  \mathbf{s}_i^0, \forall i) \to 1$ as $N,T \to \infty$ by Assumptions \ref{assume-sparse-dml} and \ref{assume-varcon}, we can ignore the uncertainty caused by estimating the breakpoint and selecting relevant units when examining the asymptotic properties of the private-effect estimator. Specifically, let $\hat \delta (\tilde{b}, \hat{\mathbf{s}}) $ be an estimator of $\delta^0$ obtained under $\tilde b$ and $\hat{\mathbf{s}}$, where $\hat{\mathbf{s}} = (\hat{\mathbf{s}}_1, \dots, \hat{\mathbf{s}}_N)$, whose true value is denoted as $\mathbf{s}^0 = (\mathbf{s}_1^0, \dots, \mathbf{s}_N^0)$. We observe that, for any $a$,
	\begin{align*}
		& \left| \Pr\left( \sqrt{NT} (\hat \delta (\tilde{b}, \hat{\mathbf{s}}) - \delta^0) <a \right) - \Pr \left( \sqrt{NT}(\hat \delta (b^0,\mathbf{s}^0 ) - \delta^0) < a \right) \right| \\
		\leq & \left| \Pr\left( \sqrt{NT} (\hat \delta (\tilde{b}, \hat{\mathbf{s}}) - \delta^0) <a , \tilde b = b^0,  \hat{\mathbf{s}} =\mathbf{s}^0 \right)  - \Pr \left( \sqrt{NT}(\hat \delta (b^0,\mathbf{s}^0 ) - \delta^0) <a \right) \right| \\
		& +  \Pr\left( \sqrt{NT} (\hat \delta (\tilde{b}, \hat{\mathbf{s}}) - \delta^0) <a , (\tilde b \neq b^0 \text{ or } \hat{\mathbf{s}} \neq \mathbf{s}^0 ) \right) \\
		= & \left| \Pr\left( \sqrt{NT} ( \hat \delta (b^0, \mathbf{s}^0) - \delta^0 ) <a , \tilde b = b^0,  \hat{\mathbf{s}} =\mathbf{s}^0 \right)  - \Pr \left( \sqrt{NT}(\hat \delta (b^0,\mathbf{s}^0 ) - \delta^0\right) <a ) \right| \\
		& +  \Pr\left( \sqrt{NT} (\hat \delta (\tilde{b}, \hat{\mathbf{s}}) - \delta^0) <a , (\tilde b \neq b^0 \text{ or } \hat{\mathbf{s}} \neq \mathbf{s}^0 ) \right)\\
		\leq & \left| \Pr \left( \tilde b = b^0,  \hat{\mathbf{s}} =\mathbf{s}^0 \right) - \Pr\left( \sqrt{NT} (\hat \delta (b^0, \mathbf{s}^0)  - \delta^0 ) <a , \text{ or } (\tilde b = b^0,  \hat{\mathbf{s}} =\mathbf{s}^0 ) \right)  \right| \\
		& +  \Pr  \left(\tilde b \neq b^0 \text{ or } \hat{\mathbf{s}} \neq \mathbf{s}^0 \right) \\
		\to & 0,
	\end{align*}
	 if $\Pr (\tilde b = b^0,  \hat{\mathbf{s}} =\mathbf{s}^0) \to 1$. Hence, $\hat \delta (\tilde{b}, \hat{\mathbf{s}})$ and $\hat \delta (b^0, \mathbf{s}^0)$ share the same asymptotic distribution, and we can simplify the analysis to study the estimator under known $b^0$ and $\mathbf{s}_i^0$, i.e.,
	\begin{align*}
		\tilde \delta^m(b^0,\mathbf{s}^0) = 
		\left( \sum_{i=1}^N \sum_{t\in \mathcal{T}_m} \left(z_{it}  - X_{it, \mathbf{s}_i^0} (b^0)'\tilde{\nu}_i^a \right)^2 \right)^{-1}  \sum_{i=1}^N \sum_{t\in \mathcal{T}_m}  \left(z_{it} - X_{it, \mathbf{s}_i^0} (b^0)'\tilde{\nu}_i^a \right) \left( y_{it} - X_{it, \mathbf{s}_i^0} (b^0)' \tilde{\gamma}_i^a  \right). 
	\end{align*}
	
	We note that 
	\begin{align*}
y_{it} - X_{it, \mathbf{s}_i^0} (b^0)' \tilde{\gamma}_i^a  
		=  X_{it, \mathbf{s}_i^0} (b^0)' (\gamma_{i,\mathbf{s}_i^0} - \tilde{\gamma}_i^a + \tilde{\nu}_i^a\delta^0 )   + (z_{it} - X_{it, \mathbf{s}_i^0} (b^0)'\tilde{\nu}_i^a) \delta^0 + u_{it}.
	\end{align*}
	It therefore holds that 
	\begin{align}\label{eq:delta-b0-s0}
		\tilde \delta^m(b^0,\mathbf{s}^0)     =& \delta^0 + \left( \sum_{i=1}^N \sum_{t\in \mathcal{T}_m} \left(z_{it}  - X_{it, \mathbf{s}_i^0} (b^0)'\tilde{\nu}_i^a \right)^2 \right)^{-1} \nonumber\\
		&  \quad \times  \sum_{i=1}^N \sum_{t\in \mathcal{T}_m} \left(z_{it} - X_{it, \mathbf{s}_i^0} (b^0)'\tilde{\nu}_i^a \right) \left(  X_{it, \mathbf{s}_i^0} (b^0)' (\gamma_{i,\mathbf{s}_i^0} - \tilde{\gamma}_i^a + \tilde{\nu}_i^a\delta^0 )  + u_{it} \right). 	
	\end{align}
	We first consider the numerator part of the right-hand side of~\eqref{eq:delta-b0-s0}. It holds that 
	\begin{align}\label{eq:numerator}
		&\sum_{i=1}^N \sum_{t\in \mathcal{T}_m}  \left(z_{it} - X_{it, \mathbf{s}_i^0} (b^0)'\tilde{\nu}_i^a \right) \left(  X_{it, \mathbf{s}_i^0} (b^0)' (\gamma_{i,\mathbf{s}_i^0} - \tilde{\gamma}_i^a + \tilde{\nu}_i^a\delta^0 )  + u_{it} \right)  \nonumber\\
		= &  \sum_{i=1}^N \sum_{t\in \mathcal{T}_m} \left( X_{it, \mathbf{s}_i^0} (b^0)' (\nu_{i,\mathbf{s}_i^0} - \tilde{\nu}_i^a ) + e_{it} \right) \left(  X_{it, \mathbf{s}_i^0} (b^0)'  (\gamma_{i,\mathbf{s}_i^0} - \tilde{\gamma}_i^a + \tilde{\nu}_i^a\delta^0 ) + u_{it} \right)  \nonumber\\
		= &  \sum_{i=1}^N \sum_{t\in \mathcal{T}_m} \left( X_{it, \mathbf{s}_i^0} (b^0)' (\gamma_{i,\mathbf{s}_i^0} - \tilde{\gamma}_i^a + \tilde{\nu}_i^a\delta^0 )  \right) \left( X_{it, \mathbf{s}_i^0} (b^0)' (\nu_{i,\mathbf{s}_i^0} - \tilde{\nu}_i^a ) \right) \nonumber \\
		& +    \sum_{i=1}^N \sum_{t\in \mathcal{T}_m} \left( X_{it, \mathbf{s}_i^0} (b^0)' (\gamma_{i,\mathbf{s}_i^0} - \tilde{\gamma}_i^a + \tilde{\nu}_i^a\delta^0 )  \right) e_{it}   +    \sum_{i=1}^N \sum_{t\in \mathcal{T}_m}  u_{it} \left( X_{it, \mathbf{s}_i^0} (b^0)' (\nu_{i,\mathbf{s}_i^0} - \tilde{\nu}_i^a ) \right)  \\
		& +    \sum_{i=1}^N \sum_{t\in \mathcal{T}_m}  u_{it} e_{it}.
	\end{align}
	We note that $
		\nu_{i,\mathbf{s}_i^0} - \tilde{\nu}_i^a
		=  -\left( \sum_{t\in \mathcal{T}_a} X_{it, \mathbf{s}_i^0} (b^0)  X_{it, \mathbf{s}_i^0} (b^0)'  \right)^{-1}  \sum_{t\in \mathcal{T}_a} X_{it, \mathbf{s}_i^0} (b^0)  e_{it}$,
	and 
	\begin{align*}
	&\gamma_{i,\mathbf{s}_i^0} - \tilde{\gamma}_i^a + \tilde{\nu}_i^a\delta^0 \\
		= & \gamma_{i,\mathbf{s}_i^0}^* - \tilde{\gamma}_i^a + (\tilde{\nu}_i^a - \nu_{i,\mathbf{s}_i^0})\delta^0 \\
		= & - \left( \sum_{t\in \mathcal{T}_a} X_{it, \mathbf{s}_i^0} (b^0)  X_{it, \mathbf{s}_i^0} (b^0)'  \right)^{-1}  \sum_{t\in \mathcal{T}_a} X_{it, \mathbf{s}_i^0} (b^0)  \tilde u_{it} 
		 + \left( \sum_{t\in \mathcal{T}_a} X_{it, \mathbf{s}_i^0} (b^0)  X_{it, \mathbf{s}_i^0} (b^0)'  \right)^{-1}  \sum_{t\in \mathcal{T}_a} X_{it, \mathbf{s}_i^0} (b^0)  e_{it} \delta^0 \\
		= & - \left( \sum_{t\in \mathcal{T}_a} X_{it, \mathbf{s}_i^0} (b^0)  X_{it, \mathbf{s}_i^0} (b^0)'  \right)^{-1}  \sum_{t\in \mathcal{T}_a} X_{it, \mathbf{s}_i^0} (b^0)  u_{it}.
	\end{align*}
	Then, the first term of \eqref{eq:numerator} is 
	\begin{align*}
		&    \frac{1}{\sqrt{NT}}  \sum_{i=1}^N (\gamma_{i,\mathbf{s}_i^0}^* -  \tilde{\gamma}_i^a)' \left( \sum_{t\in \mathcal{T}_m} X_{it, \mathbf{s}_i^0} (b^0)  X_{it, \mathbf{s}_i^0} (b^0)'  \right) \left( \nu_{i,\mathbf{s}_i^0} - \tilde{\nu}_i^a  \right) \\
		=&  \frac{1}{\sqrt{NT}}  \sum_{i=1}^N\sum_{t\in \mathcal{T}_a}  u_{it} X_{it, \mathbf{s}_i^0} (b^0)' \left( \sum_{t\in \mathcal{T}_a} X_{it, \mathbf{s}_i^0} (b^0)  X_{it, \mathbf{s}_i^0} (b^0)'  \right)^{-1} \\
		& \times  \left( \sum_{t\in \mathcal{T}_m} X_{it, \mathbf{s}_i^0} (b^0)  X_{it, \mathbf{s}_i^0} (b^0)'  \right) \left( \sum_{t\in \mathcal{T}_a} X_{it, \mathbf{s}_i^0} (b^0)  X_{it, \mathbf{s}_i^0} (b^0)'  \right)^{-1}  \sum_{t\in \mathcal{T}_a} X_{it, \mathbf{s}_i^0} (b^0)  e_{it} \\
	  = &\frac{1}{\sqrt{NT}}  \sum_{i=1}^N\sum_{t\in \mathcal{T}_a}  u_{it} X_{it, \mathbf{s}_i^0} (b^0)' \left( \sum_{t\in \mathcal{T}_a} X_{it, \mathbf{s}_i^0} (b^0)  X_{it, \mathbf{s}_i^0} (b^0)'  \right)^{-1}  \sum_{t\in \mathcal{T}_a} X_{it, \mathbf{s}_i^0} (b^0)  e_{it}   + o_p(1),
	  \end{align*}
	where the second equality above follows from the fact that 
	\begin{align}\label{eq:second-equality}
		& \bigg| \frac{1}{\sqrt{NT}}  \sum_{i=1}^N\sum_{t\in \mathcal{T}_a}  u_{it} X_{it, \mathbf{s}_i^0} (b^0)' \left( \sum_{t\in \mathcal{T}_a} X_{it, \mathbf{s}_i^0} (b^0)  X_{it, \mathbf{s}_i^0} (b^0)'  \right)^{-1} \nonumber\\
		& \times  \left( \sum_{t\in \mathcal{T}_m} X_{it, \mathbf{s}_i^0} (b^0)  X_{it, \mathbf{s}_i^0} (b^0)'  \right) \left( \sum_{t\in \mathcal{T}_a} X_{it, \mathbf{s}_i^0} (b^0)  X_{it, \mathbf{s}_i^0} (b^0)'  \right)^{-1}  \sum_{t\in \mathcal{T}_a} X_{it, \mathbf{s}_i^0} (b^0)  e_{it} \nonumber\\
		&  -\frac{1}{\sqrt{NT}}  \sum_{i=1}^N\sum_{t\in \mathcal{T}_a}  u_{it} X_{it, \mathbf{s}_i^0} (b^0)' \left( \sum_{t\in \mathcal{T}_a} X_{it, \mathbf{s}_i^0} (b^0)  X_{it, \mathbf{s}_i^0} (b^0)'  \right)^{-1}  \sum_{t\in \mathcal{T}_a} X_{it, \mathbf{s}_i^0} (b^0)  e_{it} \bigg| \nonumber\\
		=& \bigg| \frac{1}{\sqrt{NT}}  \sum_{i=1}^N\sum_{t\in \mathcal{T}_a}  u_{it} X_{it, \mathbf{s}_i^0} (b^0)' \left( \sum_{t\in \mathcal{T}_a} X_{it, \mathbf{s}_i^0} (b^0)  X_{it, \mathbf{s}_i^0} (b^0)'  \right)^{-1} \nonumber\\
		& \times \left(  \left( \sum_{t\in \mathcal{T}_m} X_{it, \mathbf{s}_i^0} (b^0)  X_{it, \mathbf{s}_i^0} (b^0)'  \right) -  \left( \sum_{t\in \mathcal{T}_a} X_{it, \mathbf{s}_i^0} (b^0)  X_{it, \mathbf{s}_i^0} (b^0)'  \right) \right) \nonumber\\
		& \times    \left( \sum_{t\in \mathcal{T}_a} X_{it, \mathbf{s}_i^0} (b^0)  X_{it, \mathbf{s}_i^0} (b^0)'  \right)^{-1} \sum_{t\in \mathcal{T}_a} X_{it, \mathbf{s}_i^0} (b^0)  e_{it} \bigg| \nonumber\\
		\leq  &  \frac{1}{\sqrt{NT}}  \sum_{i=1}^N \left\Vert \frac{1}{\sqrt{T}} \sum_{t\in \mathcal{T}_a}  u_{it} X_{it, \mathbf{s}_i^0} (b^0)'\right\Vert \cdot \left\Vert \left( \frac{1}{T} \sum_{t\in \mathcal{T}_a} X_{it, \mathbf{s}_i^0} (b^0)  X_{it, \mathbf{s}_i^0} (b^0)'  \right)^{-1} \right\Vert_2 \nonumber\\
		& \times \left\Vert  \left( \frac{1}{T}\sum_{t\in \mathcal{T}_m} X_{it, \mathbf{s}_i^0} (b^0)  X_{it, \mathbf{s}_i^0} (b^0)'  \right) -  \left( \frac{1}{T} \sum_{t\in \mathcal{T}_a} X_{it, \mathbf{s}_i^0} (b^0)  X_{it, \mathbf{s}_i^0} (b^0)'  \right) \right\Vert_2 \nonumber\\
		& \times \left\Vert  \left( \frac{1}{T} \sum_{t\in \mathcal{T}_a} X_{it, \mathbf{s}_i^0} (b^0)  X_{it, \mathbf{s}_i^0} (b^0)'  \right)^{-1} \right\Vert_2 \left\Vert \frac{1}{\sqrt{T}}\sum_{t\in \mathcal{T}_a} X_{it, \mathbf{s}_i^0} (b^0)  e_{it} \right\Vert \\
		=& o_p \left( \frac{\sqrt{N}}{\sqrt{T}}  \right).\nonumber
	\end{align}
	Here, the last equality above holds, because the first term in~\eqref{eq:second-equality} is $o_p(1)$ if $s/ (NT) \to 0$, and the second term is $o_p(\sqrt{N/T})$ due to Assumption \ref{assume-xx-stable}. The order of the first term can be derived by noting that Assumption \ref{assume-ue-iid} implies
	\begin{align*}
		& E \left(  \frac{1}{\sqrt{NT}}  \sum_{i=1}^N\sum_{t\in \mathcal{T}_a}   u_{it} X_{it, \mathbf{s}_i^0} (b^0)' \left( \sum_{t\in \mathcal{T}_a} X_{it, \mathbf{s}_i^0} (b^0)  X_{it, \mathbf{s}_i^0} (b^0)'  \right)^{-1}  \sum_{t\in \mathcal{T}_a} X_{it, \mathbf{s}_i^0} (b^0)  e_{it}  \right) \\
		= &  \frac{1}{\sqrt{NT}}  \sum_{i=1}^N s_i E(  u_{it} e_{it} ) = 0, %O\left( \frac{s}{\sqrt{NT}} \right).
	\end{align*}
	and 
	\begin{align*}
		& E  \left( \frac{1}{\sqrt{NT}}  \sum_{i=1}^N\sum_{t\in \mathcal{T}_a}   u_{it} X_{it, \mathbf{s}_i^0} (b^0)' \left( \sum_{t\in \mathcal{T}_a} X_{it, \mathbf{s}_i^0} (b^0)  X_{it, \mathbf{s}_i^0} (b^0)'  \right)^{-1}  \sum_{t\in \mathcal{T}_a} X_{it, \mathbf{s}_i^0} (b^0)  e_{it}  \right)^2  \\
		= &  \frac{1}{NT}  \sum_{i=1}^N s_i E(  u_{it}^2 e_{it}^2 ) = O\left( \frac{s}{NT} \right).
	\end{align*}

	Now we examine the second term on the right-hand side of~\eqref{eq:numerator}. By Assumptions \ref{assume-sparse-dml} and \ref{assume-varcon}, w.p.a.1, it is equal to
	\begin{align*}
		\frac{1}{\sqrt{NT}}  \sum_{i=1}^N (\gamma_{i,\mathbf{s}_i^0} - \tilde{\gamma}_i^a + \tilde{\nu}_i^a\delta^0 )' \sum_{t\in \mathcal{T}_m} X_{it, \mathbf{s}_i^0} (b^0) e_{it},
	\end{align*}
	and using cross-fitting under serial independence condition in Assumption~\ref{assume-ue-iid}, this term has a zero mean and its variance is
	\begin{align*}
		& E  \left( \frac{1}{\sqrt{NT}}  \sum_{i=1}^N (\gamma_{i,\mathbf{s}_i^0} - \tilde{\gamma}_i^a + \tilde{\nu}_i^a\delta^0 )' \sum_{t\in \mathcal{T}_m} X_{it, \mathbf{s}_i^0} (b^0) e_{it}  \right)^2  \\
		\leq  &\frac{1}{NT} E  \sum_{i=1}^N   \left\Vert 
		\left( \sum_{t\in \mathcal{T}_a} X_{it, \mathbf{s}_i^0} (b^0)  X_{it, \mathbf{s}_i^0} (b^0)'  \right)^{-1}  \sum_{t\in \mathcal{T}_a} X_{it, \mathbf{s}_i^0} (b^0)  u_{it} \right\Vert_2^2  \\
		& \times 
		E\left( \max_{i} \lambda_{\max}  \left( \sum_{t\in \mathcal{T}_m} X_{it, \mathbf{s}_i^0} (b^0) e_{it} \right) \left( \sum_{t\in \mathcal{T}_m} X_{it, \mathbf{s}_i^0} (b^0) e_{it}   \right)'  \right)  \\
		=&  O\left( \frac{s^2}{NT}\right).
	\end{align*}
Thus, we have the second term being of $O_p (s/\sqrt{NT})$.	
The third term of~\eqref{eq:numerator} can be analyzed similarly as the second term, and it is also of the same order $O_p(s/\sqrt{NT})$.
	
Finally, the last term of~\eqref{eq:numerator} determines the asymptotic distribution:
	\begin{align*}
		\frac{1}{\sqrt{NT}}  \sum_{i=1}^N \sum_{t\in \mathcal{T}_m}  u_{it} e_{it} 
		=  \sqrt{\frac{|\mathcal{T}_m| }{T}}\frac{1}{\sqrt{N |\mathcal{T}_m|}}  \sum_{i=1}^N \sum_{t\in \mathcal{T}_m}  u_{it} e_{it} 
		\to_d  N(0, E(u_{it}^2 e_{it}^2)/2).
	\end{align*}
To sum up, it follows that the numerator follows
	\begin{align*}
		& \frac{1}{\sqrt{NT}}\sum_{i=1}^N \sum_{t\in \mathcal{T}_m} \left(  X_{it, \mathbf{s}_i^0} (b^0)' (\gamma_{i,\mathbf{s}_i^0} - \tilde{\gamma}_i^a + \tilde{\nu}_i^a\delta^0 ) + u_{it} \right) \left(z_{it} - X_{it, \mathbf{s}_i^0} (\tilde b^0)'\tilde{\nu}_i^a \right)  \\
		= & O_p \left(  \frac{s}{\sqrt{NT}}  \right) +  \frac{1}{\sqrt{NT}}  \sum_{i=1}^N \sum_{t\in \mathcal{T}_m}  u_{it} e_{it} \\ 
		\to_d &  N(0, E(u_{it}^2 e_{it}^2)/2).
	\end{align*}
	
	Next, we examine the denominator of~\eqref{eq:delta-b0-s0}. By the law of large numbers, it holds that 
	\begin{align*}
		\frac{1}{NT} \sum_{i=1}^N \sum_{t\in \mathcal{T}_m}  \left(z_{it} - X_{it, \mathbf{s}_i^0} (b^0)'\tilde{\nu}_i^a \right)^2
		\to_p \frac{1}{2}E( e_{it}^2).
	\end{align*}	
	Combining the numerator and denominator, we thus have 
	\begin{align*}
		\sqrt{NT} \left(\tilde \delta^m - \delta^0\right) \to_d N \left( 0, 2 \left( E( e_{it}^2) \right)^{-1} E(u_{it}^2 e_{it}^2) \left( E( e_{it}^2) \right)^{-1} \right).
	\end{align*}
	Then, the final estimator after cross-fitting follows  
	\begin{align*}
		\tilde \delta =& \frac{1}{2}( \tilde \delta^m + \tilde \delta^a) \\
		=& \frac{1}{2} \left(\frac{1}{2}E( e_{it}^2) \right)^{-1}\left(  \frac{1}{\sqrt{NT}}  \sum_{i=1}^N \sum_{t\in \mathcal{T}_m}  u_{it} e_{it} +  \frac{1}{\sqrt{NT}}  \sum_{i=1}^N \sum_{t\in \mathcal{T}_a}  u_{it} e_{it} \right) + o_p(1) \\
		=&  \frac{1}{\sqrt{NT}}  \sum_{i=1}^N \sum_{t=1}^T  u_{it} e_{it} + o_p(1) \\
		\to_d &  N\left( 0, \left( E( e_{it}^2) \right)^{-1} E(u_{it}^2 e_{it}^2) \left( E( e_{it}^2) \right)^{-1} \right) .
	\end{align*}
	
	\qed
	
	\section{Useful Lemmas}	
	In this section, we provide useful lemmas from \citet{DzemskiOkui24}.

	We first present the following lemma which is a slightly more general version of \citet[Lemma A.15][]{DzemskiOkui24}. The proof is almost identical to that of \citet[Lemma A.15][]{DzemskiOkui24} and is thus omitted.
	\begin{lemma}
		\label{lem:fuk-nagaev-dependent-b}
		Suppose that $X_{it}$ is a strongly mixing process with zero mean for each $i=1, \dots, N$  with tail probabilities $\sup_{i = 1, \dotsc, N} P (|X_{it}| > x) \leq \exp (1-(x/a)^{d_1})$ and with strong mixing coefficients $\sup_{i = 1, \dotsc, N} a_i [t] \leq \exp ( -bt^{d_2} ) $, where $a$, $b$, $d_1$, and $d_2$ are positive constants. Let $\mathbb{P}_N$ denote a sequence of sets of probability measures that satisfy the above conditions with given values of $a$, $b$, $d_1$, and $d_2$. 
		Let \begin{align*}
			s_T^2   = \max_{1 \leq i \leq N} \max_{1 \leq t \leq T} \left(\E (X_{it}^2) + 2 \sum_{s>t} \left| \E( X_{it} X_{is})\right| \right) .
		\end{align*}
		Assume that $s_T^2 < C_s \log^{a_s} N$ for constants $C_s$ and $0 \leq a_s \leq 1$ which do not depend on $N, T$ nor $P$.
		Then, it holds that for any constant $C>0$, as $N,T\to \infty$ with $NT^{-\delta_2} \to 0$ for some $\delta_2>0$,
		\begin{align*}
			\sup_{P \in \mathbb{P}_N} P \left( \max_{1 \leq i \leq N}  \left| \frac{1}{T}\sum_{t=1}^T X_{it} \right| \geq  C T^{-1/2} \log N  \right) \to 0 .
		\end{align*}
	\end{lemma}

	We also use the following version where a random variable is indexed by $i$, $j$, and $t$.
	
	\begin{lemma}
		\label{lem:fuk-nagaev-dependent-ij-b}
		Suppose that $X_{ijt}$ is a strongly mixing process with zero mean for each $i=1, \dots, N$ and $j=1, \dots, N$ with tail probabilities $\sup_{i = 1, \dotsc, N, j = 1, \dots, N} P (|X_{ijt}| > x) \leq \exp (1-(x/a)^{d_1})$ and with strong mixing coefficients $\sup_{i = 1, \dotsc, N, j =1, \dots, N} a_{ij} [t] \leq \exp ( -bt^{d_2} ) $, where $a$, $b$, $d_1$, and $d_2$ are positive constants. Let $\mathbb{P}_N$ denote a sequence of sets of probability measures that satisfy the above conditions with given values of $a$, $b$, $d_1$, and $d_2$. 
		Let \begin{align*}
			s_T^2   = \max_{i = 1, \dotsc, N}\max_{j = 1, \dotsc, N} \max_{t = 1, \dotsc, T} \left(E (X_{ijt}^2) + 2 \sum_{s>t} \left| E( X_{ijt} X_{ijs})\right| \right) .
		\end{align*}
		Assume that $s_T^2 < C_s \log^{a_s} N$ for constants $C_s$ and $0 \leq a_s \leq 1$ which do not depend on $N, T$ nor $P$.
		Then, it holds that for any constants $C'>0$ and $0<c<1$, as $N,T\to \infty$ with $N^2T^{-\delta} \to 0$ for some $\delta>0$,
		\begin{align*}
			\sup_{P \in \mathbb{P}_N} P \left( \max_{b=1, \dots T} \max_{i = 1, \dotsc, N} \max_{j = 1, \dotsc, N}  \left| \frac{1}{T}\sum_{t=1}^b X_{ijt} \right| \geq C' T^{-1/2} 2 \log N  \right) \to 0 .
		\end{align*}
	\end{lemma}
	
	\begin{proof}
		The lemma follows Lemma \ref{lem:fuk-nagaev-dependent-b} by regarding the double index $(i,j)$ in this lemma as a single index $i$ with length $N^2$ in Lemma \ref{lem:fuk-nagaev-dependent-b}.
	\end{proof}
	
	We note that condition $s_T^2 < C_s \log^{a_s} N$ holds under Assumptions \ref{as-tail} and \ref{as-mixing}.
	
	\begin{lemma}[Lemma A.7 in \citet{DzemskiOkui24}]
		\label{lem:exp-tail}
		Suppose that two random variances $X_1$ and $X_2$ satisfy $P(|X_a| >x ) \leq C_a \exp( - b_a x^{d_a}) $ for $a=0,1$, then $P(|X_1 X_2| >x ) \leq C \exp( - b x^{d}) $ for some positive constants $C$, $b$, and $d_2$, and $P(|X_1 + X_2| >x ) \leq C' \exp( - b' x^{d'}) $ for some constants $C'$, $b'$, and $d'$.	
	\end{lemma}

	\begin{lemma}[Lemma A.9 in \citet{DzemskiOkui24}]
		\label{lem:mixing-g}
		Suppose that $(x_{it}, w_{it}, v_{it})$ is a strong mixing sequence over $t$ with mixing coefficients $\sup_i a_i [t] \leq C \exp (-at^{d}) $ for constants $C$, $a$ and $d$, then so is $ g( x_{it}, w_{it}, v_{it})$ where $g$ is a measurable function. 
	\end{lemma}

	\begin{lemma}[Lemma A.12 in \citet{DzemskiOkui24}]
		\label{lem:moment}
		Suppose that a random variable $X$ satisfies that $P (|X| > x) < C \exp (-(x/a)^{d})$ for some $C, a >0 $ and $d>1$. Then, for any integer $p$, $E (| X|^p ) <M$ for $M$ depending only on $C,a, d$ and $p$. 
	\end{lemma}

\section{Additional empirical results}
\label{app:two-groups}	

This section presents the results from the specification with the latent group structure. The number of groups is $G=2$.

\begin{figure}[htp]\caption{Heatmap of R\&D spillover in the two regimes}\label{fig:empirical-network-g2} 
	\begin{center}
		\begin{tabular}{cc}
			\includegraphics[scale=0.55]{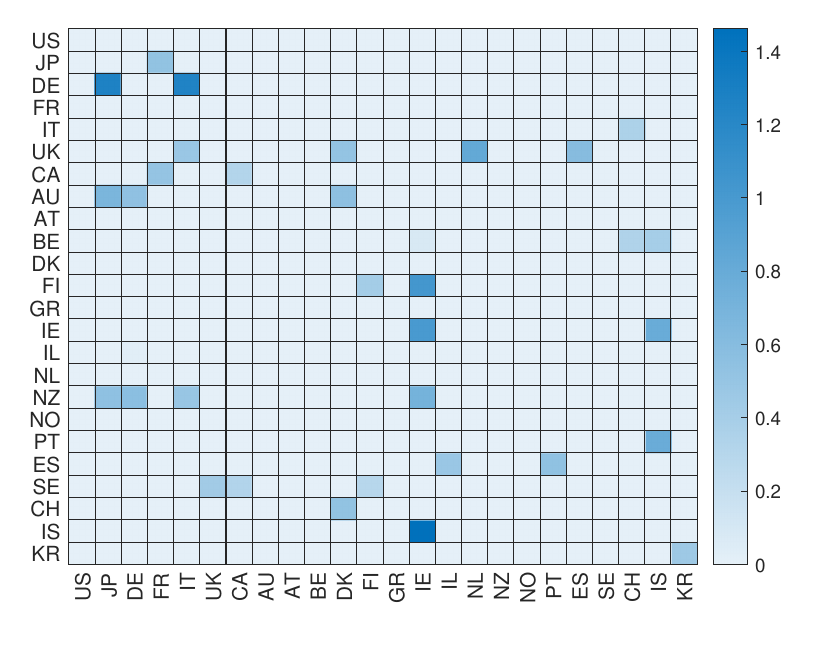} &
			\includegraphics[scale=0.55]{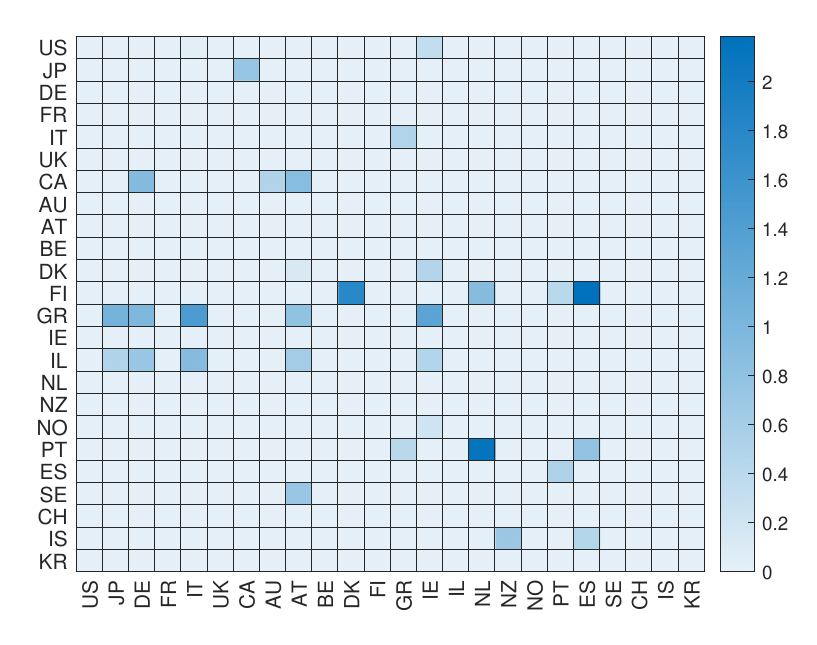}\\
			Before the break & After the break			
	\end{tabular}\end{center}
\end{figure}

Group 1 contains Greece, Portugal, France, Belgium, Switzerland, UK, Sweden, Finland, Norway, Canada, US, New Zealand, Australia Korea, and Israel, characterized by weaker effect of human capital effect (approximately 0.024). Group 2 consists of Austria, Germany, Netherlands, Denmark, Ireland, Japan, Iceland, Italy and Spain, whose effect of human capital is of much larger size (approximately 2.37)

Such group membership estimates are related to the geographic locations to some extent but not perfectly, and they do not coincide with the G7 clustering either. Overall, the estimates of private and spillover effects confirm the robustness of the results on the one hand; on the other hand, they suggest that the difference between the two specifications is minor.

\end{document}